\documentclass[manuscript,screen,nonacm]{acmart}
\setcopyright{none}
\renewcommand\footnotetextcopyrightpermission[1]{\vspace{16pt}}
\usepackage{booktabs}
\usepackage{tabularx}
\usepackage{placeins}
\usepackage{enumitem}
\usepackage{caption}
\usepackage{xspace}
\usepackage{cleveref}
\usepackage{tikz}
\usetikzlibrary{arrows.meta,positioning,calc,backgrounds}
\usepackage{pgfplots}
\usepgfplotslibrary{groupplots}
\pgfplotsset{compat=1.17}

\providecommand{\Description}[1]{}
\newcolumntype{L}[1]{>{\raggedright\arraybackslash}p{#1}}
\newcolumntype{Y}{>{\raggedright\arraybackslash}X}
\usepackage[ruled,vlined,resetcount]{algorithm2e}

\definecolor{darkcommentgray}{rgb}{0.25, 0.25, 0.25}
\definecolor{proccommentgray}{rgb}{0.3, 0.3, 0.3}

\SetAlgoLined
\SetNoFillComment
\DontPrintSemicolon

\SetCommentSty{mycommfont}
\SetKwComment{Comment}{/* }{ */}
\SetKwProg{Proc}{procedure\unskip}{}{\KwRet}

\SetKwProg{Glob}{Global variables and constants}{:}{\KwRet}
\SetKwProg{Upon}{upon}{ do}{}
\SetKwProg{Fn}{function}{:}{\KwRet}
\crefname{algocf}{algorithm}{algorithms}
\Crefname{algocf}{Algorithm}{Algorithms}

\makeatletter
\renewcommand{\ACM@linecountL}{}
\renewcommand{\ACM@linecountR}{}
\fancypagestyle{plain}{%
  \fancyhf{}%
  \fancyhead[L]{\ACM@linecountL}%
  \fancyhead[R]{\ACM@linecountR}%
  \fancyfoot[C]{\if@ACM@printfolios\footnotesize\thepage\fi}%
}
\makeatother

\newcommand{\hash}[1]{\mathcal{H}(#1)}          
\newcommand{\quorum}{n{-}f}                      

\newcommand{\rstat}[1]{\textsc{#1}} 

\newcommand{\sysname}{Vantage\xspace}

\newcommand{\proposer}{\mathsf{Proposer}}  
\newcommand{\gfull}{\mathsf{commit\mbox{-}full}}
\newcommand{\gcore}{\mathsf{commit\mbox{-}core}}
\newcommand{\gopen}{\mathsf{open\mbox{-}tip}}
\newcommand{\gskip}{\mathsf{skip}}
\newcommand{\echoskip}{\rstat{echo-skip}}
\newcommand{\noready}{\rstat{no-ready}}
\newcommand{\skipvote}{\rstat{skip-vote}}
\newcommand{\skipseal}{\rstat{skip-seal}}
\newcommand{\wish}{\rstat{wish}}
\newcommand{\reswit}{\rstat{resolver-witness}}
\newcommand{\rwish}{\rstat{resolver-wish}}
\newcommand{\fastseal}{\rstat{fast-seal}}
\newcommand{\backed}{\textnormal{\textsc{Backed}}}  

\title{Vantage: Availability-Graded Broadcast for Signature-Free BFT}


\author{Nikita Polyanskii}
\affiliation{
  \institution{IOTA Foundation}
  \country{Germany}
}
\email{nikitapolyansky@gmail.com}

\begin{document}

\begin{abstract}
Digital signatures make blocks and votes transferable evidence:
one party can prove to another what a third party said.  Authenticated
channels convince only the direct receiver: it learns who sent a message but
cannot pass that proof on.  A proposer that references another party's block
therefore cannot prove to the voters who published that block, so the
high-throughput signature-free protocols we survey first complete an
availability vote or a reliable broadcast for every published data block;
only then may a proposer reference it.

We present \sysname{}, a partially synchronous Byzantine fault-tolerant (BFT)
protocol for $n\ge 3f{+}1$ parties, at most $f$ of them Byzantine, that uses
only authenticated channels and collision-resistant hashing.  Parties publish
blocks on hash-linked author lanes; each view's proposer pairs a \emph{core
manifest} of lane prefixes that a quorum has acknowledged to it with an
optimistic \emph{tip manifest} of blocks it has just received.  A new
primitive, \emph{Availability-Graded Broadcast} (AGB), makes the core
irrevocable on a quorum of directly received responses while grading,
rather than blocking on, the tip.  A view that AGB leaves open or silent is
decided by its own on-demand instance of information-theoretic agreement.  A
block may thus be referenced in a proposal one message delay after
publication, as in signed optimistic designs.  Let $\delta$ be the actual message delay after
the Global Stabilization Time (GST).  When all parties are correct, a
proposal commits within $2\delta$ of its send if all $n$ first responses
support core and tip, and within $3\delta$ on the two-round quorum path, so a block published in
time for the next proposal is sequenced $3\delta$ to $4\delta$ after
publication.  Safety holds under asynchrony; liveness holds after GST.  On an
emulated ten-region wide-area network (WAN) with 100 parties, \sysname{} has
the lowest p50 latency of the compared signed and signature-free protocols
at every offered load up to its peak, and it sequences 250k 512-byte
transactions per second at a median end-to-end latency below 500\,ms.
\end{abstract}

\maketitle
\pagestyle{plain}

\section{Introduction}\label{sec:introduction}

State-machine replication (SMR) lets a group of replicas --- we call them
\emph{parties} --- maintain one transaction log even when some are faulty.
Modern high-throughput
systems let every party publish transaction blocks in parallel, rather than
routing every transaction through one leader, and consensus orders short
references to those blocks~\cite{narwhal,autobahn,mysticeti,sailfish}.

\paragraph{When may a block be referenced?}
There are two basic choices.  With \emph{certified admission}, a proposer
waits until a quorum of parties confirm that they hold the block.  This is robust,
but the confirmation round must finish before ordering can begin, and
because every reference is already confirmed, voters can treat the proposal
as \emph{indivisible}: one value, accepted or not as a whole.  With
\emph{uncertified admission}, the next proposer references the block as soon
as it receives it from the author.  This removes that delay when all
messages arrive, but if the proposal is still indivisible, a voter that
missed one block can only wait for it, fetch it, or reject the proposal and
every other block in it.  Signed protocols can take the second choice
because signatures make fetched data attributable; every high-throughput
signature-free protocol we know takes the first.  This paper shows that the
second choice does not require signatures.

\paragraph{The cost of signatures.}
Votes, certificates, and blocks in signed BFT protocols carry signatures
under quantum-vulnerable schemes such as Ed25519 and BLS, which fit in
48--96 bytes and, for BLS, aggregate to constant size.  The post-quantum
standards give up both: an ML-DSA-44 signature is 2{,}420 bytes and an
SLH-DSA-128s signature is 7{,}856 bytes --- $38\times$ and $123\times$ an
Ed25519 signature~\cite{fips204,fips205} --- and neither standard
aggregates, so with $n$ parties of which up to $f$ are Byzantine, a quorum
certificate carries $n-f$ signatures: at $n=100$ and $f=33$ it alone carries
$162$--$526$\,kB.  Time adds to size: hash-based signing takes 0.3--1\,s per
SLH-DSA-128s signature on a recent x86 core, and the faster parameter sets
pay with 17\,kB signatures~\cite{nistsigzoo}.  Replacing classical
signatures with a post-quantum standard on the consensus critical path has
been measured to roughly double commit latency at 50 nodes~\cite{simpleit}.
Both costs are structural and land on every vote, certificate, and block.
Authenticated channels avoid them: a symmetric per-message tag is 16 bytes,
cheap to compute, and needs no public-key infrastructure (PKI).

\paragraph{What signatures buy.}
Signatures provide transferable attribution: after receiving signed bytes
from any party, a voter can verify who signed them.  An authenticated channel
tells a receiver who sent a message, but the receiver cannot pass that proof
on, and a hash proves that fetched bytes match a reference, not who published
them.  Repair can therefore restore missing bytes but not authorship.
Without transferable attribution, waiting can stall the view, repair supplies
data but not evidence, and rejection discards otherwise available blocks, so
the signature-free systems we survey certify each orderable block or wait for
availability votes before proposing it~\cite{sailfishpp,simpleit}, forgoing
the one-delay admission of optimistic signed designs.  Those designs have
their own bottleneck: in Autobahn's optimistic-tip path, a replica that
lacks a proposed tip's data requests it from the leader and votes only after
it arrives~\cite[\S5.5.2 and Lemma~7]{autobahn}, which is also how the
Autobahn artifact we build on behaves; and BBCA-Chain votes on the leader's
proposal as one indivisible value~\cite{bbcachain}.  In both, a view's
progress rests on one leader's proposal arriving whole;
\Cref{sec:leader-relay} isolates what this relay costs under partial
dissemination.

\paragraph{Our approach.}
Vantage admits another author's block one message delay after publication
without signatures and without first completing an availability vote or a
broadcast for it --- to our knowledge, the first high-throughput protocol to
do so.  The claim concerns the fault-free data path under the conditions of
\Cref{sec:closest-comparisons}, not faulty-view or resolver latency.  Each
party publishes its blocks as a hash-linked \emph{author lane}, and a view's
proposer names lane prefixes in two parts: a \emph{core manifest} of
prefixes that a quorum has already acknowledged to it directly, and an
optimistic \emph{tip manifest} of newer blocks received directly from their
authors.  Voters answer with graded responses: support for core and
tip together, or for the core alone when a tip is still missing at the
response deadline, so a block that some voters lack costs at most the
optimistic part.  The proposer never serves data as a voting precondition.
A quorum of responses makes the core irrevocable, and a quorum with one grade
decides the tip immediately.  No transferable proof is involved: every count
is of first-hand (directly received) responses.  We call this two-level
primitive
\emph{Availability-Graded Broadcast} (AGB).  A view left open by
mixed grades or silent by a faulty proposer is resolved by its own on-demand
Information-Theoretic HotStuff (IT-HS) instance~\cite{iths}, which validates
candidate outcomes against each receiver's own recorded responses for that
view.  A silent view whose proposer merely crashed is skipped by a quorum of
skip votes without the resolver.

Our contributions are:
\begin{itemize}[leftmargin=1.4em,topsep=2pt,itemsep=2pt]
  \item \textbf{Availability-Graded Broadcast (\Cref{sec:primitive}).}
        We define a broadcast primitive whose outcome has two levels: a
        quorum of first-hand responses makes the core irrevocable, and
        a quorum that grades the tip unanimously decides at once whether the
        tip is included or dropped.  We prove that correct parties agree on
        the core and on the tip decision without timing assumptions, and that
        a correct proposer's view terminates after the Global Stabilization
        Time (GST).  For a faulty proposer the primitive alone does not
        guarantee a common completion; the resolver below supplies one.
  \item \textbf{Signature-free SMR at $n\ge 3f{+}1$ (\Cref{sec:protocol}).}
        We build \sysname{} on AGB over hash-linked author lanes, and resolve
        each view that AGB leaves open or silent with its own on-demand IT-HS
        instance~\cite{iths}.  Safety needs no timing assumptions: correct
        parties \emph{seal} (commit) one common outcome per view, their
        output sequences are prefixes of one another, and every output block
        was published by its author and remains retrievable.  After GST,
        every view is eventually sealed at all correct parties and every
        valid block of a correct author is eventually included.
  \item \textbf{One-delay admission and latency (\Cref{sec:closest-comparisons,sec:primitive}).}
        Let $\delta$ be the actual message delay after GST, and let all
        parties be correct.  A block that another party published directly
        to the proposer may be referenced $\delta$ after publication; the
        proposal seals $2\delta$ after its send when all $n$ first
        responses support both core and tip, or $3\delta$ on the two-round
        quorum path.  A block published just in time for the next proposal
        (\emph{aligned}; \Cref{sec:closest-comparisons}) is therefore
        sequenced between $3\delta$ and
        $4\delta$ after publication once earlier views are output; admitting
        it through the core instead, after its acknowledgment quorum, costs
        $5\delta$ to $6\delta$.  The two-delay seal cannot be improved: with $f\ge 1$, no
        broadcast decides one message delay after its send, even in
        all-correct executions~\cite{arxunauth}.  These bounds hold at any
        $n\ge 3f{+}1$; at $n=3f{+}1$, the minimum for partially synchronous
        consensus~\cite{dls}, $3\delta$ equals the lowest
        publication-to-sequencing latency of any protocol we compare under
        the same assumptions, signed ones included
        (\Cref{sec:closest-comparisons}).
  \item \textbf{Implementation and evaluation.}
        We implement \sysname{}, two Autobahn variants (seamless and
        optimistic all-to-all), and two Simple-IT variants (Bracha and
        optimistic reliable broadcast) in one Rust binary derived from the
        Autobahn artifact~\cite{autobahn}, so all five share one transport,
        batching, storage, client, and metrics stack; Bluestreak and
        Sailfish++ run from their own artifact~\cite{starfish}.  With 100
        parties on an emulated ten-region WAN, \sysname{} sequences up to
        about 268k tx/s and has the lowest p50 latency at every offered load
        up to that peak: 0.49\,s at 250k offered tx/s, against 0.54\,s for
        Bluestreak, which reaches about 240k tx/s at that load, and 0.64\,s
        for optimistic Autobahn.  At low load \sysname{}'s wire bytes per
        sequenced byte are one fourteenth of optimistic Autobahn's, since it
        carries no per-lane availability proofs or quorum certificates
        (\Cref{sec:evaluation}).  At $n=20$, when Byzantine authors publish
        blocks to only a few parties, a \sysname{} proposer never relays data
        and its p50 latency stays at 0.46\,s while optimistic Autobahn's
        reaches 2.3\,s (\Cref{sec:leader-relay}).
\end{itemize}

\paragraph{Roadmap.}
\Cref{sec:overview} walks through one view and its recovery cases;
\Cref{sec:closest-comparisons} compares the closest protocols under one
latency model; \Cref{sec:model,sec:primitive,sec:protocol} give the model, the
primitive, and the resolver with their guarantees; \Cref{sec:evaluation}
evaluates the implementation; and \Cref{sec:discussion} discusses
limitations.  Complete specifications, proofs, comparison and evaluation
details, and extended related work are in the appendices.

\section{Overview}\label{sec:overview}

\paragraph{Data path: author lanes.}
Every party continuously publishes its data blocks on its own \emph{author
lane}, a hash chain sent block by block to all parties
(\Cref{fig:dissemination}).  For each author, a party records the longest
valid lane prefix it has received directly from that author.  When a proposal
later names a lane prefix, every party whose directly received prefix covers
it \emph{acknowledges} it in its response to the proposal.  Because channels
are authenticated, direct receipt tells the receiver who published the block,
and $f{+}1$ matching acknowledgments tell any party that counts them that at
least one correct party received the prefix from its author and retains it; a
quorum of $\quorum$ acknowledgments guarantees $f{+}1$ correct holders.  No
block waits for a broadcast or a certification round before it may be
proposed.

\begin{figure*}
    \centering
\begin{tikzpicture}[
    >={Stealth[length=2mm]},
    y=1.15cm,
    block/.style={draw, minimum size=5.5mm, inner sep=0pt, fill=white},
    core/.style={block, fill=black!16, thick},
    tip/.style={block, dashed, fill=black!4},
    hlink/.style={->, gray, shorten >=0.7pt, shorten <=0.7pt},
    plabel/.style={anchor=east, font=\small},
    blabel/.style={font=\scriptsize, anchor=south, inner sep=1.5pt},
]
  \draw[line width=1.6pt, black!45] (0.45,-0.55) -- (0.45,-4.4);
  \node[font=\scriptsize\itshape, text=black!55] at (0.45,-4.62) {genesis};

  \foreach \x [count=\i] in {1.0,2.2,3.3,4.8,5.9}
    \node[block] (b1-\i) at (\x,-1) {};
  \foreach \x [count=\i] in {1.0,2.3,3.7,5.1}
    \node[block] (b2-\i) at (\x,-2) {};
  \foreach \x [count=\i] in {1.0,2.0,3.0,4.1,5.3,6.6}
    \node[block] (b3-\i) at (\x,-3) {};
  \foreach \x [count=\i] in {1.0,2.5,3.8}
    \node[block] (b4-\i) at (\x,-4) {};

  \foreach \p in {1,...,4} \node[plabel] at (0.2,-\p) {$p_{\p}$};
  \foreach \i in {1,...,5} \node[blabel] at (b1-\i.north) {$b_1^{\i}$};

  \foreach \l in {1,...,4} \draw[hlink] (b\l-1.west) -- (0.5,-\l);
  \foreach \i/\j in {2/1,3/2,4/3,5/4}     \draw[hlink] (b1-\i.west) -- (b1-\j.east);
  \foreach \i/\j in {2/1,3/2,4/3}         \draw[hlink] (b2-\i.west) -- (b2-\j.east);
  \foreach \i/\j in {2/1,3/2,4/3,5/4,6/5} \draw[hlink] (b3-\i.west) -- (b3-\j.east);
  \foreach \i/\j in {2/1,3/2}             \draw[hlink] (b4-\i.west) -- (b4-\j.east);

  \node[core] at (b1-3) {}; \node[core] at (b2-2) {};
  \node[core] at (b3-2) {}; \node[core] at (b4-2) {};
  \node[tip]  at (b1-5) {}; \node[tip]  at (b2-4) {};
  \node[tip]  at (b3-6) {}; \node[tip]  at (b4-3) {};

  \begin{scope}[shift={(7.45,-1.7)}]
    \node[core, minimum size=4mm] at (0.15,0) {};
    \node[anchor=west, font=\scriptsize] at (0.5,0) {last block of a core prefix named in $C_v$};
    \node[tip, minimum size=4mm] at (0.15,-0.5) {};
    \node[anchor=west, font=\scriptsize] at (0.5,-0.5) {tip block named in $T_v$};
    \draw[hlink] (0,-1.0) -- (0.3,-1.0);
    \node[anchor=west, font=\scriptsize] at (0.5,-1.0) {hash link $b_i^{k}\!\to\!b_i^{k-1}$};
  \end{scope}
\end{tikzpicture}
    \caption{Each party $p_i$ publishes an author lane $b_i^1,b_i^2,\ldots$
    from a shared genesis; lanes advance asynchronously to a ragged frontier.  The
    snapshot is from the point of view of $\proposer(v)$, whose view proposal
    $B_v=(C_v,T_v)$ names at most one prefix per author in each manifest: the
    core manifest $C_v$ (gray), whose entries a quorum has acknowledged, and
    the optimistic tip manifest $T_v$ (dashed), newer blocks received directly
    from their authors.}
    \label{fig:dissemination}
    \Description{Four author lanes leave a shared genesis and advance to
    ragged frontiers; the last block of each core prefix is gray and each tip
    block is dashed.}
\end{figure*}
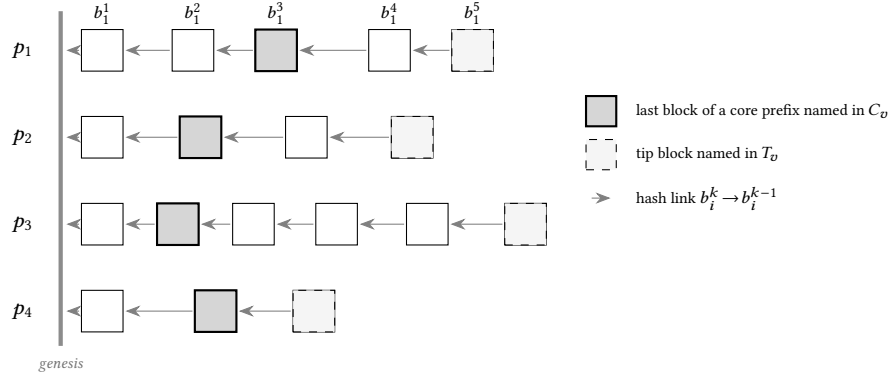

\paragraph{One view.}
Views have round-robin proposers.  The proposer of view $v$ broadcasts one
\emph{view proposal} $B_v=(C_v,T_v)$ made of two \emph{manifests}, each naming
at most one lane prefix per author.
The \emph{core manifest} $C_v$ names prefixes that the proposer holds directly
from their authors and for which it has itself counted a quorum ($\quorum$) of
acknowledgments, given in the responses to earlier proposals that named those
prefixes.  The optimistic \emph{tip manifest} $T_v$ names newer blocks that
the proposer received directly from their authors, each extending its author's
core entry (or genesis when the author has none).

This per-view broadcast and its graded responses are \sysname{}'s instance of
\emph{Availability-Graded Broadcast} (AGB), the primitive of
\Cref{sec:primitive}.  Every party answers in Bracha's two-round ECHO/READY
pattern~\cite{bracha} (\Cref{fig:agb-decision}):
\begin{itemize}[leftmargin=1.4em,topsep=2pt,itemsep=2pt]
  \item an \emph{ECHO} of $B_v$, with grade~$1$ if the party can verify core
        and tip --- each named prefix was received directly from its author or
        acknowledged by $f{+}1$ parties, and every tip block is held --- and
        grade~$0$ if it can verify the core but not the tip by a local
        deadline, because the tip is missing, conflicts with what the party
        received from that author, or does not extend the core.  A failed tip
        lowers the grade instead of blocking the view.  The ECHO also carries
        the party's acknowledgments of the named prefixes it holds directly
        from their authors; these feed later core manifests;
  \item a \emph{READY} for $B_v$ once a quorum of ECHOs names $B_v$, with
        grade~$1$ or $0$ if a quorum of those ECHOs share that grade, and the
        mixed grade otherwise.  A mixed READY is provisional: if one grade
        reaches a quorum before the party's READY deadline, the party sends
        one single-grade refinement for the same $B_v$.
\end{itemize}
A quorum of READYs for $B_v$, grades ignored, \emph{completes} the view: by
quorum intersection, $B_v$ is the only proposal that can complete in view $v$,
and its core becomes an irrevocable part of the view's output.  A single-grade
READY quorum also \emph{seals} the view (a \emph{direct seal}): grade~$1$
selects the full outcome $\gfull$, core and tip; grade~$0$ selects the
core-only outcome $\gcore$.  We call the primitive alone --- proposal, ECHO,
READY, completion, and direct seals --- \emph{Direct AGB}; the rest of this
section is Vantage's composition around it.  The first addition is the
\emph{fast seal}: grade-$1$ ECHOs for $B_v$ from all $n$ parties seal $\gfull$
one message delay before any READY quorum can form.  When all parties are
correct and the network is timely, a proposal therefore seals within $2\delta$
of its send, where $\delta$ is the actual message delay, and within $3\delta$
on the READY path (\Cref{fig:commit-chain}).  Views are pipelined: the
proposer of view $v{+}1$ sends $B_{v+1}$ as soon as it has processed $B_v$,
without waiting for view $v$'s responses.

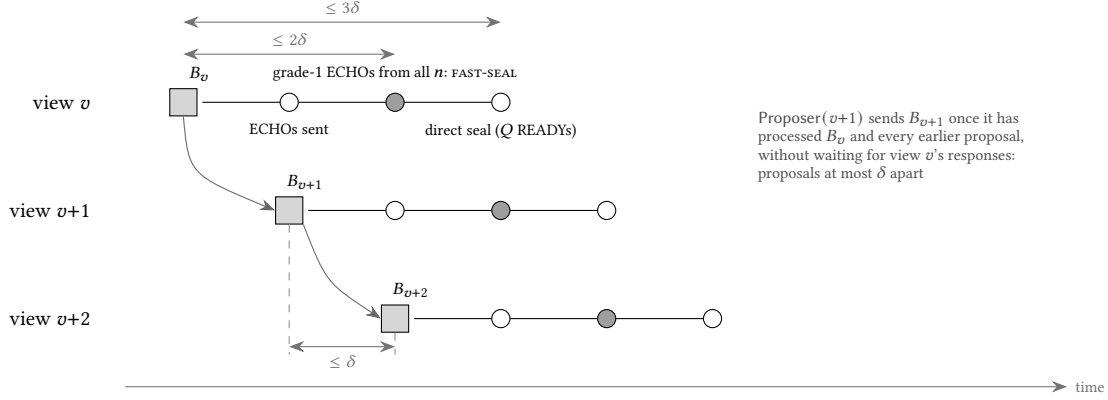
\begin{figure*}[t]
    \centering
\begin{tikzpicture}[
    >={Stealth[length=2mm]},
    x=1.4cm, y=1.06cm,
    evt/.style={draw, circle, fill=white, minimum size=2.4mm, inner sep=0pt},
    fast/.style={evt, fill=black!38},
    prop/.style={draw, fill=black!16, minimum size=3.6mm, inner sep=0pt},
    lane/.style={font=\small, anchor=east},
    stg/.style={font=\scriptsize, anchor=north},
    ann/.style={font=\scriptsize, anchor=south},
    tlabel/.style={font=\scriptsize, text=black!70},
]
  \node[lane] at (-0.25,0) {view $v$};
  \node[prop] (pv) at (0.55,0) {};
  \node[ann] at (0.7,0.16) {$B_v$};
  \draw (0.73,0) -- (3.55,0);
  \node[evt] at (1.55,0) {};
  \node[stg] at (1.55,-0.16) {ECHOs sent};
  \node[fast] at (2.55,0) {};
  \node[ann] at (2.55,0.16) {grade-$1$ ECHOs from all $n$: \fastseal};
  \node[evt] at (3.55,0) {};
  \node[stg] at (3.55,-0.16) {direct seal ($Q$ READYs)};
  \draw[<->, black!55] (0.55,0.6) -- (2.55,0.6)
    node[midway, above, font=\scriptsize] {$\le 2\delta$};
  \draw[<->, black!55] (0.55,1.0) -- (3.55,1.0)
    node[midway, above, font=\scriptsize] {$\le 3\delta$};

  \node[lane] at (-0.25,-1.35) {view $v{+}1$};
  \node[prop] (pv1) at (1.55,-1.35) {};
  \node[ann] at (1.7,-1.19) {$B_{v+1}$};
  \draw (1.73,-1.35) -- (4.55,-1.35);
  \node[evt] at (2.55,-1.35) {};
  \node[fast] at (3.55,-1.35) {};
  \node[evt] at (4.55,-1.35) {};

  \node[lane] at (-0.25,-2.7) {view $v{+}2$};
  \node[prop] (pv2) at (2.55,-2.7) {};
  \node[ann] at (2.7,-2.54) {$B_{v+2}$};
  \draw (2.73,-2.7) -- (5.55,-2.7);
  \node[evt] at (3.55,-2.7) {};
  \node[fast] at (4.55,-2.7) {};
  \node[evt] at (5.55,-2.7) {};

  \draw[->, black!60] (pv.south) .. controls (0.6,-0.9) .. (pv1.west);
  \draw[->, black!60] (pv1.south east) .. controls (1.9,-2.25) .. (pv2.west);
  \node[tlabel, align=left, anchor=west] at (5.9,-0.55)
    {$\proposer(v{+}1)$ sends $B_{v+1}$ once it has\\
     processed $B_v$ and every earlier proposal,\\
     without waiting for view $v$'s responses:\\
     proposals at most $\delta$ apart};

  \draw[dashed, black!45] (pv1.south) -- (1.55,-3.15);
  \draw[dashed, black!45] (pv2.south) -- (2.55,-3.15);
  \draw[<->, black!55] (1.55,-3.05) -- (2.55,-3.05)
    node[midway, below, font=\scriptsize] {$\le\delta$};

  \draw[->, black!55] (0,-3.55) -- (8.9,-3.55)
    node[anchor=west, font=\scriptsize, text=black!55] {time};
\end{tikzpicture}
    \caption{Pipelined Vantage views in the all-correct case.  Grade-$1$
    ECHOs from all $n$ parties arrive everywhere within two delays of the
    proposal's send and fast-seal $\gfull$; the direct seal from $Q=n-f$ READYs
    follows within three.  The proposer of $v{+}1$ sends $B_{v+1}$ once it
    has processed $B_v$ and every earlier proposal, without waiting for view
    $v$'s ECHOs or READYs, so consecutive proposals are at most one delay
    apart; formal view entry through the WISH pacemaker (\Cref{sec:primitive})
    arms the fallback timers and triggers the send only when an earlier
    proposal has not arrived.}
    \label{fig:commit-chain}
    \Description{Three overlapping views show proposal send, ECHO, fast-seal,
    and direct-seal events.  Each fast seal is two message delays after the
    proposal's send and each direct seal three; proposals are one delay apart.}
\end{figure*}

\begin{figure}
\centering
\begin{tikzpicture}[
    >={Stealth[length=2mm]},
    resp/.style={draw, rounded corners=1pt, fill=black!3, font=\scriptsize,
                 inner sep=2.6pt},
    outc/.style={draw, thick, fill=black!16, font=\scriptsize,
                 inner sep=2.6pt},
    inp/.style={draw, fill=black!16, font=\scriptsize, inner sep=2.6pt},
    ghost/.style={draw, dashed, fill=black!4, font=\scriptsize,
                  inner sep=2.6pt},
    cond/.style={font=\scriptsize, text=black!55, align=center,
                 fill=white, inner sep=1.5pt},
    band/.style={font=\scriptsize, text=black!55},
]
  \node[band] at (0.85,1.5) {view $v$ proposal};
  \node[band] at (3.9,1.5) {ECHO --- one per party};
  \node[band] at (6.9,1.5) {READY --- from $Q$ ECHOs};
  \node[band] at (12.4,1.5) {terminal seal --- one per view};

  \node[inp] (bv) at (0.85,-1.25) {$B_v=(C,T)$};

  \node[resp] (e1)  at (3.9,0)    {grade-$1$ ECHO};
  \node[resp] (e0)  at (3.9,-2.5) {grade-$0$ ECHO};
  \node[resp] (esk) at (3.9,-3.9) {\echoskip{}};

  \node[resp] (r1) at (6.9,0)     {READY-$1$};
  \node[resp] (rm) at (6.9,-1.25) {READY-mix};
  \node[resp] (r0) at (6.9,-2.5)  {READY-$0$};
  \node[resp] (nr) at (6.9,-3.9)  {\noready{}};

  \node[ghost] (gopen) at (9.3,-1.25) {$\gopen$};
  \node[cond, below=1pt of gopen] {core fixed, tip open};
  \node[resp]  (svote) at (9.3,-3.9) {\skipvote{}};

  \draw[rounded corners=3pt, black!70] (10.85,0.45) rectangle (13.95,-4.35);
  \node[outc] (full)  at (12.4,0)    {$\gfull(C,T)$};
  \node[outc] (core)  at (12.4,-2.5) {$\gcore(C)$};
  \node[outc] (gskip) at (12.4,-3.9) {$\gskip$};

  \draw[->] (bv.east) -- node[cond, midway] {core and tip verified} (e1.west);
  \draw[->] (bv.east) -- node[cond, midway] {core verified;\\tip missing at the deadline} (e0.west);
  \draw[->] (bv.south) .. controls (0.85,-3.2) and (2.3,-3.9) ..
    node[cond, midway] {core invalid, or no proposal\\by the echo deadline} (esk.west);

  \draw[->] (e1.east) -- node[cond, above=1pt] {$Q$ grade-$1$ ECHOs} (r1.west);
  \draw[->] (e0.east) -- node[cond, below=1pt] {$Q$ grade-$0$ ECHOs} (r0.west);
  \draw[->] (e1.south east) .. controls (5.3,-0.5) .. (rm.west);
  \draw[->] (e0.north east) .. controls (5.3,-2.0) .. (rm.west);
  \node[cond] at (5.45,-1.25) {$Q$ ECHOs,\\both grades};
  \draw[->] (esk.east) -- (nr.west);
  \node[cond, anchor=south] at (5.4,-3.64) {no READY by the\\READY deadline};

  \draw[->] (r1.east) -- node[cond, above=1pt] {$Q$ grade-$1$ READYs} (full.west);
  \draw[->] (r0.east) -- node[cond, above=1pt] {$Q$ grade-$0$ READYs} (core.west);
  \draw[->] (rm.east) -- (gopen.west);
  \node[cond, anchor=south] at (7.95,-0.98) {$Q$ READYs,\\not all one grade};
  \draw[->, dashed, black!60] (gopen.east) .. controls (10.2,-1.25) and (10.55,-0.45)
    .. node[cond, pos=0.55, above left=-1pt] {resolver} (full.south west);
  \draw[->, dashed, black!60] (gopen.east) .. controls (10.2,-1.25) and (10.55,-2.05)
    .. (core.north west);

  \draw[->, dashed, black!60] (e1.north) .. controls (4.6,1.0)
    and (10.2,1.0) .. (full.north);
  \node[cond, anchor=south] at (7.6,0.9)
    {grade-$1$ ECHOs from all $n$ parties: \fastseal{}};

  \draw[->] (nr.east) -- (svote.west);
  \node[cond, anchor=south] at (8.1,-3.64) {own \noready{},\\$Q$ \echoskip{}s};
  \draw[->] (svote.east) -- (gskip.west);
  \node[cond, anchor=south] at (10.15,-3.64) {$Q$ \skipvote{}s};

  \node[band, anchor=west, align=left] at (-0.1,-4.95)
    {$Q=\quorum$.  Each party sends one echo-stage and one ready-stage response
     per view; READY-mix is refined once if a grade reaches $Q$ before the
     READY deadline.};
\end{tikzpicture}
\caption{One Vantage view $v$.  Left to right: the proposal, each party's
  echo-stage response, its ready-stage response formed from $Q$ ECHOs, and
  the terminal seal, at most one per view.  Any READY quorum completes the
  view and fixes the core; a single-grade quorum seals it $\gfull$ or
  $\gcore$, and a quorum without a single grade leaves the tip open
  ($\gopen$) for the resolver (\Cref{fig:resolution}).  Two branches are
  Vantage shortcuts around Direct AGB: grade-$1$ ECHOs from all $n$ parties
  fast-seal $\gfull$ one message delay before READYs, and a quorum of
  grounded skip votes seals a crash-only silent view as $\gskip$.
  \Cref{app:primitive} gives the complete rules.}
\Description{Flow diagram with four columns: a proposal branches into
  grade-one, grade-zero, and echo-skip ECHOs; quorums of ECHOs form READY-1,
  READY-0, READY-mix, or no-ready; quorums of READYs select one terminal seal
  per view among commit-full, commit-core, and skip, while a quorum without a
  single grade leaves the intermediate open-tip state for the resolver.}
\label{fig:agb-decision}
\end{figure}

\paragraph{Recovery: open and silent views.}
Direct AGB leaves two cases unsealed.  If READY grades stay mixed, a completed
view is \emph{open}: its core is fixed, but its tip is neither included nor
discarded.  If the proposer is faulty, a view may be \emph{silent}: parties
time out and send refusals, \echoskip{} in place of an ECHO and \noready{} in
place of a READY.  Vantage seals both cases with two mechanisms, again without
transferable evidence.

The first is a shortcut for a view whose proposer merely crashed.  A party
that has itself sent \noready{} and itself counted a quorum of \echoskip{}s
broadcasts a \skipvote{}; we call such a vote \emph{grounded}, since it rests
only on the party's own evidence.  A quorum of grounded votes
\emph{skip-seals} the view as $\gskip$ (\Cref{fig:agb-decision}).  A skip vote
is a persistent stance: a party that has voted to skip never endorses a
non-skip outcome for that view in the resolver below, and vice versa, so a
skip quorum and a conflicting resolver quorum cannot both form.

Every other unsealed view $u$ --- a resolver \emph{target} --- is decided by
its own on-demand instance of Information-Theoretic HotStuff
(IT-HS)~\cite{iths}, started once the data path has advanced three views past
$u$ (\Cref{fig:resolution}).  Candidates are $\gfull(C,T)$, $\gcore(C)$, and
$\gskip$, each proposed only by a party whose own recorded responses for $u$
justify it.  A party endorses (ECHOs) a candidate only if it is consistent
with what it itself sent in view $u$: a party that sent a grade-$1$ READY
never endorses $\gcore$, one that sent grade~$0$ never endorses $\gfull$, and
only a party that sent \noready{} endorses $\gskip$.  Any READY quorum that
could have direct-sealed $u$ shares a correct party with any resolver ECHO
quorum, and that party refuses one of the two; a candidate that conflicts with
a possible direct seal therefore never gathers a quorum.  This local check
stands in for the certificate a signed protocol would attach.  Because it is
local, a party that joins the resolver late cannot repeat it, so before IT-HS
locks on a candidate or carries it into a later resolver view, a quorum of
parties must have witnessed its endorsement quorum first-hand (\backed{},
\Cref{sec:protocol}).  Targets run concurrently with one another and with the
data path, and later proposals carry no resolution state.

The fast seal needs the same kind of guard.  A party that counts grade-$1$
ECHOs from all $n$ parties seals $\gfull$ without waiting for READYs, so every
party that sends a grade-$1$ ECHO installs an \emph{optimistic lock}: it
refuses every conflicting resolver candidate until it has counted $f{+}1$
echo-stage responses that are not grade-$1$ ECHOs for $B_v$.  One such
response would not do, because a Byzantine party may report different grades
to different parties; among $f{+}1$ there is a correct party, whose single
ECHO is the same everywhere, so no party can have counted all $n$ grade-$1$
ECHOs.

\paragraph{Sequencing and output.}
A party \emph{seals} each view at most once, with whichever result arrives
first --- a direct, fast, or skip seal or a resolver decision;
\Cref{sec:protocol} shows that they never disagree.  Sealed views are
\emph{sequenced} in view order: once all earlier views are sealed, a completed
view's core blocks are appended --- even while its own tip is still open ---
and its tip blocks follow if the view seals $\gfull$.  \emph{Output} is local
delivery after retrieval: a party that lacks a sequenced block requests it
from the other parties, and every sealed manifest is backed by first-hand
acknowledgments from a correct party that retains its blocks
(\Cref{sec:model}).  We say that a resolver \emph{decides} a target and that a
party \emph{seals} a view.

\section{Closest Comparisons}\label{sec:closest-comparisons}

\Cref{sec:overview} gave Vantage's fault-free data-path latencies: a tip is
admitted one message delay after publication, and the proposal seals $2\delta$
after its send on the fast path or $3\delta$ on the READY path.  This section
places those numbers beside the closest protocols under one set of
assumptions.  The outcome is that Vantage's tip path reaches the lowest
publication-to-sequencing latency in the comparison, signed protocols
included, while the signature-free baselines spend two delays before a block
can be proposed at all.

We measure from a correct nonleader publishing a block to a different correct
party proposing the object that orders it --- a consensus proposal or a
directed-acyclic-graph (DAG) vertex --- which we call the block's
\emph{carrier}; in \sysname{}, the carrier names a lane prefix that contains
the block.  We write $A$ for publication-to-carrier latency (\emph{admission})
and $L$ for publication-to-sequencing latency.  Every row uses the same
favorable setting: all parties are correct after GST, there is no queueing,
earlier output is complete, and publication is \emph{aligned} with the next
carrier opportunity, so $A$ and $L$ exclude any wait for the protocol's
cadence.  For Vantage this is the steady state of \Cref{cor:tx-latency}:
every earlier view is output and each proposal is sent as soon as its tip is
eligible.  Here $\delta$ is the actual post-GST message delay.
``Signature-free'' means no digital signatures or transferable signature
certificates; every row still assumes authenticated channels.

\Cref{tab:data-path-comparison} has two groups.  Mempool-Simple-IT, in its
faster Opt-RBC variant, and Sailfish++ are the signature-free protocols whose
latencies~\cite{simpleit} derives in the same message-delay terms; both
certify a block, by an
availability-vote quorum or by optimistic reliable broadcast (Opt-RBC), before
a carrier may reference it.  FinWhale (Mysticeti with a fast path), Autobahn,
and BBCA-Chain are signed protocols whose carrier may reference a block one
delay after publication --- a directly received DAG block, an optimistic lane
tip, or a one-trip best-effort-broadcast (BEB) block --- which is the
admission latency Vantage claims without signatures.  Autobahn requires a
proof of availability (PoA) for the tip's parent; BBCA denotes Byzantine
Broadcast with Complete--Adopt.
\begin{table}[t]
\caption{Aligned cross-author data latency in the common setting of
\Cref{sec:closest-comparisons}.  $A$ is publication-to-carrier (admission)
latency and $L$ is publication-to-sequencing latency; lower is better.
Starred values need responses from all $n$ parties.}
\label{tab:data-path-comparison}
\centering
\normalsize
\(
  \text{publish}
  \xrightarrow{\ A\ }
  \text{first legal carrier}
  \longrightarrow
  \text{sequenced at }L
\)
\smallskip

\setlength{\tabcolsep}{4pt}
\renewcommand{\arraystretch}{1.22}
\begin{tabularx}{\linewidth}{@{}L{3.1cm}cYcc@{}}
\toprule
Protocol / path
& \shortstack{Signature-\\free}
& Required before carrier
& \shortstack{Admission\\$A$}
& \shortstack{Aligned total\\$L$} \\
\midrule
Mempool-Simple-IT (Opt-RBC)~\cite{simpleit}
& Yes
& Availability-vote quorum
& $2\delta$
& $5\delta$ \\

Sailfish++ (Opt-RBC)~\cite{sailfishpp,simpleit}
& Yes
& Opt-RBC for every vertex
& $2\delta$
& \shortstack{$5\delta$ for $n{-}f{-}1$ vertices\\$7\delta$ for the remaining $f$} \\

\addlinespace[2pt]
\textbf{\sysname{} (optimistic lane tip)}
& \textbf{Yes}
& \textbf{Directly published lane tip; no PoA}
& $\boldsymbol{\delta}$
& \shortstack{$\boldsymbol{3\delta^{\ast}}$ fast (all-to-all)\\$\boldsymbol{4\delta}$ quorum (READY)} \\
\addlinespace[2pt]
FinWhale (Mysticeti $+$ fast path)~\cite{finwhale,mysticeti}
& No
& Direct receipt of signed DAG block
& $\delta$
& \shortstack{$3\delta$ fast\\$4\delta$ quorum} \\

Autobahn (optimistic lane tip)~\cite{autobahn}
& No
& Current lane tip; parent has PoA
& $\delta$
& \shortstack{$3\delta^{\ast}$ fast (all-to-all)\\$4\delta^{\ast}$ fast (linear)\\$6\delta$ quorum (slow path)} \\

BBCA-Chain~\cite{bbcachain}
& No
& One-trip BEB block
& $\delta$
& $4\delta$ \\

\bottomrule
\end{tabularx}
\end{table}

Starred values are zero-fault fast paths that need responses from all $n$
parties; one non-voter disables them.  FinWhale's fast path is not starred:
it commits on $n-p$ references with $n=3f+2p-1$, so at the displayed $p=1$
it runs at $n=3f{+}1$ and survives one non-voter~\cite{finwhale}.  Sailfish++
sequences $n-f-1$ of each round's nonleader vertices in $5\delta$ and the
remaining $f$ two delays later~\cite{sailfishpp,simpleit}.  Derivations and
protocol-by-protocol qualifications appear in \Cref{app:comparison}; the
table compares latency under these assumptions, not fault tolerance or
overall dominance.

The highlighted row is the lane-tip path, the one-delay admission of
\Cref{sec:introduction}.  The core is the other manifest in the same
proposal: a core entry first waits for its acknowledgment quorum, so it
sequences at $5\delta$ aligned through the fast seal (\Cref{cor:tx-latency}).
The tip row's $3\delta^{\ast}$ assumes that the carrier is sent as soon as the
tip is eligible and that all $n$ parties respond; the Direct AGB READY path
gives $4\delta$.  The two paths race at run time and a party seals with
whichever completes first: the starred value waits for the slowest of all $n$
responses, the READY path for two successive quorum steps, so under
heterogeneous delays the nominally slower path can finish first
(\Cref{sec:decision-anatomy} measures the split in the local $n=20$ study).
Without alignment, a block that misses the current proposal waits at most one
delay for the next, so the tip bound is $4\delta$ through the fast seal
(\Cref{cor:tx-latency}); on the READY path each of these figures is one delay
larger.

The signature-free baselines spend two delays before a carrier: an
availability-vote quorum in Mempool-Simple-IT and Opt-RBC for each Sailfish++
vertex.  Vantage instead admits a directly published lane tip after one
delay, and a receiver that lacks the tip lowers its grade instead of
withholding its ECHO.  The certified baselines buy stronger pre-carrier
availability with those delays: a Vantage tip manifest may remain open and
use the resolver (\Cref{sec:protocol}).

The three signed rows pay differently.  Autobahn also admits a current lane
tip after one delay but requires a PoA for its parent and signed votes; its
displayed $3\delta$ endpoint, like Vantage's, needs all $n$ participants, and
its two-phase slow path takes $6\delta$.  FinWhale's $p=1$ fast path reaches
the same $3\delta$ and survives one non-voter, using signed DAG blocks.
BBCA-Chain's carrier is a one-trip BEB block followed by a three-trip BBCA
decision.  At $n=3f{+}1$, Vantage therefore \emph{matches} the lowest aligned
endpoint in the comparison without signatures; it does not strictly dominate
these protocols or prove a lower bound.

The admission rule also shapes leader load.  A \sysname{} proposal names
lane prefixes by reference; an under-disseminated tip costs its grade while
the core is retained, and missing bytes are repaired from lane authors and
holders, so no party's response waits on data supplied by the proposer.
When a leader instead admits data that lagging replicas must fetch from it
before voting, the missing-data fanout concentrates on that leader
(optimistic all-to-all Autobahn); when the proposal is voted as one
indivisible value, a single under-disseminated block stalls or forfeits the
whole view (BBCA-Chain).  \Cref{sec:leader-relay} measures the relay cost
under a partial-dissemination stress.

\FloatBarrier

\section{Model and Preliminaries}\label{sec:model}

The bounds above are proved in the model fixed here.  \Cref{sec:primitive}
then specifies Direct AGB and its guarantees, and \Cref{sec:protocol} the
resolver and output rules that complete Vantage; the proofs are in
\Cref{app:primitive,app:protocol}.

\paragraph{Parties and faults.}
The protocol runs among a fixed set $\Pi=\{p_1,\dots,p_n\}$ of $n\ge 3f{+}1$
parties. A party is \emph{correct} if it follows the protocol and
\emph{Byzantine} otherwise.  A static, computationally bounded adversary
chooses before the execution a fixed set of at most $f$ parties,
which may deviate arbitrarily under its coordination.  Thus a correct party
is one that is never corrupted during the execution.
As $n\ge 3f{+}1$, at least $n-f\ge 2f{+}1$ parties are correct.
Agreement uses only that at most $f$ parties are ever faulty.  The static
choice is used where $f{+}1$ matching first-hand statements certify that one
correct party holds a block --- availability witnesses and resolver origin
bits --- because retrievability and repair need that party to remain correct
and keep serving the block.  Core entries, backed by $n-f$
acknowledgments, keep $n-2f\ge f{+}1$ correct holders under any $f$
corruptions; the $f{+}1$ threshold does not, so an adversary that corrupts a
holder after it acknowledged could remove the only correct copy.  We do not
treat that case.

\paragraph{Communication.}
Every pair of parties shares a reliable, authenticated point-to-point channel:
the recipient of a message learns the identity of its sender, the adversary can
neither forge nor tamper with messages between correct parties, and every message
sent between correct parties is eventually delivered. Channels are not assumed to
be first-in, first-out (FIFO). Authentication is \emph{not transferable}: that $q$ received $m$ on its
channel from $p$ is not evidence $q$ can relay to a third party. This is the
relevant difference from the signature-based setting, and the reason
non-equivocation must be re-established through quorums.  A concrete link may
use a pairwise message-authentication code (MAC), but protocol messages contain
no transferable vector of recipient-specific MACs: forwarded authentication
tags are never accepted or counted as evidence.  A broadcast includes
its sender: a party delivers its own message to itself immediately and
counts it as a first-hand message from itself.

\paragraph{Timing.}
We assume partial synchrony~\cite{dls}: there is an unknown Global
Stabilization Time (GST) and a known bound $\Delta$ such that every message
between correct parties sent at time $t$ arrives by
$\max(t,\mathrm{GST})+\Delta$. Safety never rests on timing; liveness is
guaranteed after GST. For a particular execution, let $\delta\le\Delta$ denote
its actual upper bound after GST: every correct-to-correct message sent at time
$t$ arrives by $\max(t,\mathrm{GST})+\delta$. The protocol knows $\Delta$ but
not $\delta$; latency statements use $\delta$ to state the execution's actual
delay.  This is an explicit \emph{in-flight-message convention}: even a
correct-to-correct message sent before GST is delivered by
$\mathrm{GST}+\delta$ in the execution-specific analysis.
Local clocks are monotone throughout the execution and measure elapsed real
time without drift after GST; they need not have a common origin.  Thus a
timer armed for duration $\theta$ before GST has at most $\theta$ local time
remaining at GST and expires within at most $\theta$ real time thereafter.
At a local deadline, a party processes every message delivered by that
deadline, and evaluates the positive rules those messages enable, before
taking the timeout rule.  The
protocol uses the constants
\[
  \theta_E=3\Delta,\qquad \theta_R=4\Delta
\]
for a missing echo-stage response and a missing initial ready-stage response --- the
smallest values for which the derivations in
\Cref{thm:agb-live,lem:responsive-pacing} hold as written, tight when
$\delta=\Delta$; they
are not claimed to be protocol lower bounds.  These
constants affect only timer-triggered responses; positive evidence fires a
rule as soon as its local condition holds.

\paragraph{Execution fairness.}
Liveness is stated for fair, non-Zeno executions.  Every finite real-time
interval contains only finitely many protocol events; local time and armed
timers at correct parties advance; every action that remains enabled at a
correct party is eventually taken; and a correct holder eventually serves a
request issued by a correct party.  These assumptions rule out an execution
that performs infinitely many view transitions in finite time or permanently
starves an enabled local action.  They impose no delivery deadline before GST.

\paragraph{Cryptography.}
Beyond the ideal authenticated-channel abstraction, the protocol's sole
cryptographic primitive is a collision-resistant hash function
$\mathcal{H}\colon\{0,1\}^*\to\{0,1\}^{\lambda}$ with security parameter
$\lambda$.  For a post-quantum instantiation, collision resistance is required
against quantum polynomial-time adversaries.  The protocol uses no digital
signatures, public-key infrastructure, or threshold cryptography.  Relative to
the ideal channel abstraction, all guarantees hold except with negligible
probability in $\lambda$, with a hash collision as the only protocol-level
computational failure event.  A concrete MAC-based channel implementation
additionally relies on pairwise-MAC unforgeability; a post-quantum
implementation requires that property against quantum polynomial-time
adversaries as well.  We treat $\hash{x}$ as a succinct, collision-free
reference to $x$.

\paragraph{Instance and wire encoding.}
An execution fixes a session identifier $\mathsf{sid}$, which binds the
membership, fault threshold, and application epoch; a maximum data payload
length $\ell_{\max}$; and a deterministic, state-independent validity
predicate $\textsc{BlockOK}$ for data blocks (state-dependent transaction
checks are performed later by deterministic execution).  Every protocol
message carries $\mathsf{sid}$; a party rejects a message from another
session before storing or counting it.  All messages use one canonical,
prefix-free byte encoding $\mathsf{enc}$ with author-indexed vectors in
increasing author order, and every hash input begins with a distinct fixed
domain tag and $\mathsf{sid}$ (for example, \texttt{data-block} or
\texttt{view-proposal}), so objects of different types cannot share an
encoding.  A message with a noncanonical encoding, more entries than its
type's cap (\Cref{app:primitive}), or a data payload longer than
$\ell_{\max}$ is malformed and is never counted.

\paragraph{Quorums.}
A \emph{quorum} is any set of $\quorum{}$ parties. We use repeatedly that, for
$n\ge 3f{+}1$:
\begin{enumerate}[label=(\roman*),leftmargin=2.2em,topsep=2pt,itemsep=1pt]
  \item every quorum contains at least $n-2f\ge f{+}1$ correct parties, and the
        correct parties always form a quorum; and
  \item any two quorums $Q_1,Q_2$ satisfy
        $\lvert Q_1\cap Q_2\rvert \ge 2(n{-}f)-n = n-2f \ge f{+}1$, hence share at least one
        correct party.
\end{enumerate}

\section{Availability-Graded Broadcast}\label{sec:primitive}

\Cref{sec:overview} introduced the view proposal $B_v=(C_v,T_v)$ and its
graded responses (\Cref{fig:agb-decision}).  This section states the rules of
Direct AGB precisely and separates them from Vantage's proposer policy and
shortcuts.  For a correct proposer, each core entry carries $\quorum$
exact-position acknowledgments (R2 below); the tip manifest is not certified
before proposal.  Every manifest is sent by value in the proved protocol.  Full
lane syntax, publication, acknowledgment, repair, and the value-domain hooks
appear in \Cref{app:primitive}.

Direct AGB separates \emph{completion} from \emph{sealing}.  Completion fixes
one proposal and makes its core irrevocable.  A homogeneous grade-$1$ READY
quorum seals $\gfull(C_v,T_v)$, a homogeneous grade-$0$ quorum seals
$\gcore(C_v)$, and a residual mixed quorum leaves $\gopen(C_v,T_v)$ for the
Vantage composition.  Direct AGB never outputs $\gskip$ and does not promise
totality for a Byzantine proposer.  It is not Feldman--Micali
gradecast~\cite{feldmanmicali}: gradecast grades a receiver's confidence in
one value, whereas AGB fixes a common core and grades only the inclusion of
an optional suffix.

\subsection{Compact protocol}

All statements are broadcast over authenticated channels and counted only
when received directly from distinct authors.  Let $Q=n-f$.  Each party sends
one immutable echo-stage response and one initial ready-stage response per
view; an initial READY-mix may refine once to a homogeneous READY for the same
proposal.  The complete event and direct result follow these four rules.

\begin{enumerate}[label=\textbf{R\arabic*.},leftmargin=2.8em,topsep=3pt,itemsep=3pt]
  \item \textbf{Propose.}  The designated proposer sends one well-formed
        $(C,T)$.  A correct Vantage proposer forms $C$ from directly
        published lane frontiers with $Q$ exact-position ECHO
        acknowledgments and forms $T$ from directly published extensions of
        those frontiers.
  \item \textbf{Echo.}  A receiver fixes the first proposal received directly
        from the proposer.  After the proposal becomes active through formal
        entry or the responsive frontier, the receiver immediately ECHOs
        grade $1$ if core and tip pass their local checks
        (\Cref{fig:agb-decision}, left).  If it
        has not sent that ECHO by its bounded echo fallback deadline, it sends
        grade $0$ when the core passes, or \echoskip{} otherwise.
        If no active well-formed proposal exists by the absolute deadline
        $e_i(v)+3\Delta$, it sends \echoskip{}.  A proposal ECHO also carries
        availability claims outside its immutable identity.  An exact-position
        bit acknowledges the named direct prefix; after chain verification the
        claim also covers its ancestors.  A sparse integer $d>0$ acknowledges
        the directly published prefix ending $d$ heights behind, after the
        anchored ancestry is verified.  Repair can provide bytes but
        cannot create direct provenance or acknowledgments.
  \item \textbf{Ready.}  On $Q$ ECHOs for one proposal, a party broadcasts
        READY-$1$ or READY-$0$ if that
        grade already has $Q$ ECHOs, and READY-mix otherwise.  A provisional
        mix refines once if one grade reaches $Q$ before all ECHO responses or
        the deadline; otherwise the refinement window closes.  If the party
        has sent no READY by $e_i(v)+4\Delta$, it emits \noready{}.
  \item \textbf{Complete and direct-seal.}  Any $Q$ READYs naming one proposal,
        grades ignored and each author counted once, complete it.  A
        homogeneous $Q$ additionally emits the corresponding full or core
        direct result (\Cref{fig:agb-decision}, right).  If the completing state has neither homogeneous
        quorum, the Vantage caller quarantines its carried non-exact-quorum tip
        witnesses.  Late READYs continue to be counted after completion.
\end{enumerate}

Two further rules of a correct Vantage proposer are policy rather than part
of the primitive, and no receiver checks them.  If a direct extension has a
general but not exact acknowledgment quorum, $T$ names the least-height such
extension first, so that a stable exact-position census forms.
Completed-open quarantine drops an author whose carried tip never gathered
$Q$ exact acknowledgments from ordinary later tip manifests, never from the
core, and retries at most one quarantined author per proposal through $T$ on
an exponential proposer-turn schedule (\Cref{sec:discussion}).

The WISH pacemaker is the only formal entry mechanism.  A \wish{} message
carries a view high-watermark.  A party relays that high-watermark after
first-hand support from $f+1$ authors and, after support from $Q$ authors,
enters each missing view up to that high-watermark.  Immediately before
sending either its echo-stage response for view $v-1$ or its initial
ready-stage response for view $v-2$, a party checks whether that send completes
the pair --- its echo-stage response for view $v-1$ together with its initial
ready-stage response for view $v-2$ --- and if so raises its own \wish{}
high-watermark through $v+1$.  A
separate responsive frontier activates consecutive proposals as
soon as they arrive; it is scheduling state, never counted evidence.  After
GST, correct entries to a view lie within $2\delta$, within $\delta$ when all
parties are correct, and one view can delay advancement by at most
$4\Delta+\delta$.  The complete pacemaker and timer proof are in
\Cref{sec:pacemaker}.

Vantage adds one shortcut outside Direct AGB.  Before sending a grade-$1$
ECHO for a proposal $B=(C,T)$, a party installs an optimistic lock on that
full payload.  The lock releases permanently after $f+1$
first-hand nonmatching echo-stage responses.  Counting grade-$1$ ECHOs from all $n$
parties submits $\gfull(C,T)$ to Vantage's $\rstat{try-seal}$ arbiter without
creating completion, a direct result, or a resolution witness.  The ordinary
Direct AGB READY path continues in parallel (\Cref{fig:commit-chain});
\Cref{sec:decision-anatomy} measures which path finishes first in the local
$n=20$ study.

\subsection{Guarantees}

The \emph{pacemaker cutoff} below is the finite view after which the pacemaker
supplies the entry and proposal-timing premises of conditional termination.

\begin{theorem}[Direct AGB and fast-seal summary]\label{thm:main-agb}
Without timing assumptions, no two correct parties complete different
proposals in one view; every direct result is compatible with the unique
completion, and a completed core is valid, author-backed, and permanently
retrievable.  Every correct-proposer view past the pacemaker cutoff completes
at all correct parties when its persistent core checks hold.
In an all-correct responsive suffix, a prefix-extending proposal
fast-seals full within $2\delta$ of its send; the Direct AGB READY path seals
full within $3\delta$.
\end{theorem}

The safety proof is quorum intersection over first-hand, one-shot ECHOs and
READYs.  A READY quorum contains a correct sender backed by an ECHO quorum;
two such quorums intersect in a correct immutable response.  For the
Vantage-only shortcut, the all-$n$ grade-$1$ ECHO set fixes every correct
ECHO, while the optimistic locks exclude a conflicting resolution-bearing
proposal's ECHO quorum in either temporal order.  The full safety,
conditional-termination, and lock-compatibility
proofs are in \Cref{sec:agb-guarantees,lem:fast-seal}.

For a block aligned with the next proposal opportunity, tip admission costs
$\delta$ and the fast seal costs another $2\delta$, giving the $3\delta$ endpoint in
\Cref{tab:data-path-comparison}; the READY endpoint is $4\delta$.  A
core entry needs an anchoring proposal and its ECHO census before a
later proposal can carry that prefix in $C$.  Its aligned endpoint is therefore
$5\delta$ through fast seal, or $6\delta$ when that later proposal follows the
READY path.  With arbitrary arrival phase, the proved tip endpoints are
$4\delta$ in the worst phase and $3.5\delta$ in expectation under the
uniform-phase convention; the corresponding core bounds are $7\delta$ and
$6.5\delta$.  These latencies
do not bound mixed or Byzantine resolution.  Direct AGB uses
$O(n^2)$ logical messages per view --- at most $3n^2+O(n)$ with READY
refinements --- and the proved by-value protocol uses
$O(n^3\lambda)$ bits, excluding lane publication, repair, and the resolution
plane.

\section{Vantage Resolution and Output}\label{sec:protocol}

Direct AGB promises no common outcome for a faulty proposer
(\Cref{sec:primitive}): a completed view can remain open when its READY grades
do not become homogeneous, and a silent or partially disseminated view may not
complete at all.  This section makes the recovery mechanisms sketched in
\Cref{sec:overview} precise (\Cref{fig:resolution}).  The first-result-wins
rule of \Cref{sec:overview} is the \(\rstat{try-seal}\) arbiter, which
combines four compatible submissions: a Direct AGB full/core result, the all-\(n\) fast-seal full result, a
grounded skip result, or a target-local resolver decision.  The
arbiter stores the first submission and treats later compatible submissions as
idempotent.  It is Vantage composition state, not a second Direct AGB input.
Different unresolved views run independent resolver instances, so no global
fallback height or later AGB proposal serializes them.

\begin{figure}[t]
    \centering
\begin{tikzpicture}[
    >={Stealth[length=2mm]},
    x=1.0cm, y=1.1cm,
    box/.style={draw, rounded corners=1pt, fill=black!7,
                font=\scriptsize, align=center, inner sep=2.5pt},
    phase/.style={draw, circle, fill=white, minimum size=2.6mm, inner sep=0pt},
    backed/.style={draw, rounded corners=1pt, fill=black!20,
                   font=\scriptsize, align=center, inner sep=2.5pt},
    final/.style={draw, circle, fill=black!35, minimum size=3mm, inner sep=0pt},
    ann/.style={font=\scriptsize, text=black!65, align=center},
]
  \node[box] (agb)  at (0,0)    {AGB target $u$};
  \node[box] (open) at (2.35,0) {open or\\silent};
  \draw[->] (agb) -- (open);

  \node[box] (skip) at (5.4,1.3) {$Q$ grounded\\\skipvote{}s};
  \node[final] (skipdone) at (7.6,1.3) {};
  \node[ann, anchor=west] at (7.85,1.3) {$\gskip$};
  \draw[->] (open.north) |- node[ann, pos=0.75, above] {crash-only silent view} (skip.west);
  \draw[->] (skip) -- (skipdone);

  \node[box] (cand)    at (4.75,0) {candidates:\\full, core, skip};
  \node[box] (wish)    at (7.35,0) {WISH quorum:\\enter view $r$};
  \node[box] (propose) at (9.7,0)  {PROPOSE $R$};
  \draw[->] (open.east) -- node[ann, below] {otherwise} (cand.west);
  \draw[->] (cand) -- (wish);
  \draw[->] (wish) -- (propose);

  \node[phase] (echo) at (0.3,-1.9) {};
  \node[ann, below=1pt of echo] {ECHO};
  \node[backed] (witness) at (3.3,-1.9) {$Q$ witnesses:\\\backed{}};
  \draw[->] (propose.south) -- ++(0,-0.5) -| (echo.north);
  \draw[->] (echo) -- node[ann, above] {$Q$ consistent ECHOs} (witness);
  \node[phase] (k1)   at (5.3,-1.9) {}; \node[ann, below=1pt of k1]   {KEY1};
  \node[phase] (k2)   at (6.5,-1.9) {}; \node[ann, below=1pt of k2]   {KEY2};
  \node[phase] (k3)   at (7.7,-1.9) {}; \node[ann, below=1pt of k3]   {KEY3};
  \node[phase] (lock) at (8.9,-1.9) {}; \node[ann, below=1pt of lock] {LOCK};
  \node[final] (done) at (10.3,-1.9) {}; \node[ann, below=1pt of done] {$Q$ DONE};
  \draw[->] (witness) -- (k1);
  \draw[->] (k1) -- (k2);
  \draw[->] (k2) -- (k3);
  \draw[->] (k3) -- (lock);
  \draw[->] (lock) -- (done);
  \node[ann, anchor=west] at (10.6,-1.9) {$\rstat{seal}(u)\to X$};

  \draw[->, black!55] (-0.7,-3.05) -- (12.3,-3.05);
  \node[ann, anchor=west] at (-0.7,-3.33)
    {later AGB views and other target-local resolvers continue concurrently};
\end{tikzpicture}
    \caption{Resolution mechanisms for a silent or open target.  A clean
    crash-only view can skip after grounded post-ready votes.  Otherwise the
    target's own WISH/IT-HS instance turns a locally checked fresh value into a
    stable backed value before KEY1 and decides it after the remaining IT-HS
    phases.  Other target-local resolvers and later data-plane views continue in
    parallel, but ordered output cannot cross an unresolved cursor position.}
    \label{fig:resolution}
    \Description{A silent view can use a grounded skip-vote shortcut or an
    independent per-target WISH and IT-HS agreement instance.}
\end{figure}
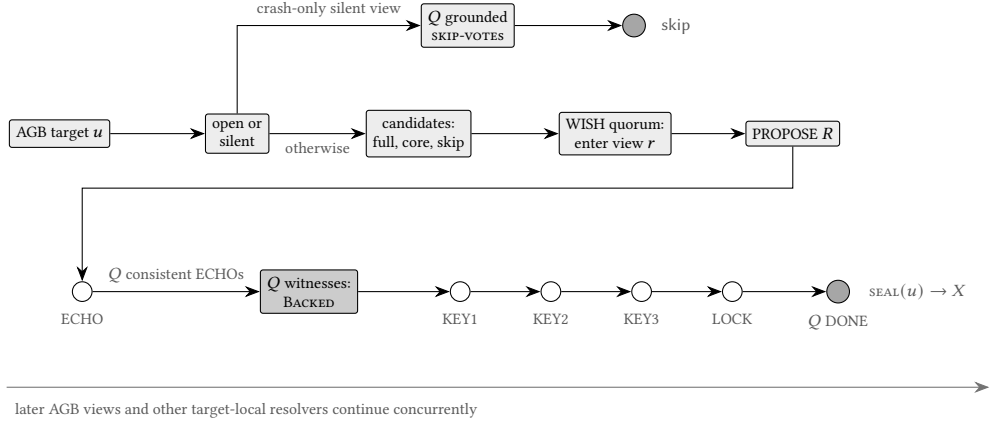

The reusable resolver contract is per target: all correct calls choose one
exact outcome; a full, core, or skip call permanently excludes every future
incompatible READY quorum; a non-skip call supplies author-backed manifests;
and every correct party eventually obtains the chosen result.  Direct AGB
composed with any service satisfying these properties has one compatible
terminal outcome, and a completed view can never resolve as skip.  Vantage
realizes the contract with the guards and agreement instance below;
the formal resolver properties and generic composition theorem are in
\Cref{sec:resolver-contract}.

The local data-plane proposal horizon is the larger of the next responsive
proposal opportunity and the WISH high-watermark.  Once it reaches \(u+3\), a party scans its
first-hand target state.  A candidate needs \(Q=n-f\) ready-stage responses.
A full candidate additionally needs \(f+1\) matching grade-\(1\) ECHOs; a
core candidate needs \(Q\) initial ready-stage responses whose first counted
form is not grade \(1\) and
\(f+1\) matching ECHOs of either grade; skip needs \(Q\) \noready{}
responses.  These finite candidate sets only grow.  A nonempty set at an unsealed
party starts a typed per-target \rwish{} exchange: \(f+1\) high-watermarks amplify the own
WISH and \(Q\) enter the resolver view; its round-robin primary proposes on
\(Q\) accepted IT-HS key suggestions (\Cref{sec:commit-rule}).  An initial
WISH is retransmitted until entry.  Without local evidence, remote traffic from fewer than $f+1$ parties
cannot allocate a target instance.  A party that already holds a compatible terminal
seal nevertheless remains a participant: one peer WISH activates it and makes
it broadcast its own initial WISH.  It fresh-ECHOs only a resolver value whose
submitted outcome equals that seal.
This prevents a selectively
delivered local seal from removing a correct party needed by the others.

Every target has a persistent stance
\(z_i(u)\in\{\mathsf{free},\mathsf{non\mbox{-}skip},
\mathsf{skip\mbox{-}voted}\}\).  Before a fresh non-skip resolver ECHO, a
correct party checks its final target echo- and ready-stage history, required data, and
active optimistic lock, then atomically claims the non-skip stance.  Before
sending \skipvote{}, it instead claims skip-voted; the sends are mutually
exclusive in either order.  The fresh ECHO carries origin \(1\) for full only
after a matching grade-\(1\) target ECHO, and for core after a matching target
ECHO of either grade.  A fresh non-skip ECHO quorum is usable only with at
least \(f+1\) origin-\(1\) senders, preserving publication provenance.
Every resolver ECHO also binds the exact proposal key \(q\): \(q=0\) denotes
fresh validation and \(q>0\) denotes a \backed{} value carried from an earlier
resolver view.  ECHO quorums match \((u,r,q,x)\), so fresh and carried ECHOs
for the same value cannot be merged.

After a usable fresh ECHO quorum, a correct party broadcasts one first-hand
witness for that target, resolver view, and value.  It relays its own witness
after \(f+1\) matching witnesses; \(Q\) make the value \backed{}.  The
causally first correct witness must follow a genuine fresh ECHO quorum, while
the relay makes \backed{} eventually common.  A WISH sender also receives the
responder's current own WISH and, once per coordinate, its applicable retained
resolver sends through the requested view plus one, together with its DONE.
This send-upon-join rule closes races in which phase traffic arrived before
target activation.  Only a \backed{} value may enter KEY1
or be re-proposed after a resolver view change (\Cref{fig:resolution}).  Full
guards, the IT-HS
\textsc{AcceptKey}/\textsc{OpenLock} rules, and the proof are in
\Cref{app:protocol,app:commit-rule}; \Cref{alg:main-composition} lists the
rules in event form.

\begin{algorithm*}[t]
\SetAlgoSkip{}
\caption{Compact Vantage composition at party \(p_i\); \(Q=n-f\)}
\label{alg:main-composition}
\scriptsize\linespread{0.70}\selectfont
\Comment{\(v\): data-plane view; \(u\): unresolved target view; \(r\):
  target-local resolver view; \(\hash{R}\): digest of a resolver value
  \(R\); \(X\): a submitted outcome; \(\backed{}(u,R)\)
  abbreviates \(\backed{}_i(u,\hash{R})\).}
\Upon{genesis}{
  enter view \(1\) and broadcast \(\langle\wish,2\rangle\); every view then
  runs Direct AGB (R1--R4, \Cref{sec:primitive}) with the wrappers below
}
\Upon{about to send the echo-stage response for view \(v-1\) or the initial
  ready-stage response for view \(v-2\), whichever completes that pair}{
  raise the own \wish{} high-watermark through \(v+1\) and piggyback it
}
\Upon{\(f+1\) direct \wish{} high-watermarks support \(x\) above the own wish}{
  persist and broadcast the own \wish{} high-watermark through \(x\)
}
\Upon{\(Q\) direct \wish{} high-watermarks support \(x\)}{
  invoke \(\rstat{enter}(v)\) for every locally missing \(v\le x\), in order
}
\Upon{about to send a grade-\(1\) ECHO for proposal \(B\) in view \(v\)}{
  install the optimistic lock on \(B\); it releases once \(f+1\) counted
  echo-stage responses for \(v\) are not grade-\(1\) ECHOs for \(B\)
}
\Upon{Direct AGB emits \(\rstat{complete}(v)\to(C_v,T_v)\)}{
  store the irrevocable core; run \textsc{OutputCursor}
}
\Upon{Direct AGB emits \(\rstat{direct-seal}(v)\to X\), grade-\(1\) ECHOs for
  one proposal \(B\) are counted from all \(n\) parties, or \(Q\)
  \(\langle\skipvote,v\rangle\) statements are counted}{
  submit \(X\), \(\gfull(B)\), or \(\gskip\), respectively, to
  \(\rstat{try-seal}(v,\cdot)\)
}
\Upon{own \noready{} and \(Q\) direct \echoskip{} for target \(u\), with
  \(z_i(u)=\mathsf{free}\) and no non-skip seal}{
  atomically persist \(z_i(u)\gets\mathsf{skip\mbox{-}voted}\) and broadcast
  \(\langle\skipvote,u\rangle\)
}
\Upon{the proposal horizon reaches \(u+3\) and local AGB evidence justifies
  a first candidate for unresolved \(u\)}{
  allocate the target instance; broadcast \(\langle\rwish,u,1\rangle\) and
  retransmit it every \(5\Delta\) until entry
}
\Upon{resolver WISHes for an unallocated target \(u\) arrive from \(f+1\)
  parties, or one arrives while \(p_i\) holds a terminal seal for \(u\)}{
  allocate the instance; in the second case also broadcast
  \(\langle\rwish,u,1\rangle\)
    \Comment*[r]{the resolver rules below act only on allocated targets}
}
\Upon{\(f+1\) resolver WISHes support \(r\) above the own WISH}{
  raise and broadcast the own WISH through \(r\)
}
\Upon{a resolver WISH \((u,r)\) from \(p_j\)}{
  send \(p_j\) the current own WISH; replay once each applicable own
  SUGGEST, PROOF, PROPOSE, phase statement, and witness through \(r+1\), plus
  the own DONE, if any
}
\Upon{\(Q\) resolver WISHes support \(r\)}{
  enter each missing target-local view through \(r\); send the IT-HS
  suggestion/proof; arm the \(5\Delta\) no-proposal and \(11\Delta\)
  full-view timers
}
\Upon{the no-proposal timer expires without a primary proposal, or the
  full-view timer expires before a decision}{
  raise and broadcast the own WISH through \(r+1\)
}
\Comment{Resolver statements are counted through one view above the own
  WISH.  The ECHO send and the ECHO, KEY, and LOCK quorum rules act only
  while their view is current; a quorum already counted at entry fires on
  entry.  Witness, \backed{}, WISH, and DONE relay and decision are
  view-independent.}
\Upon{the primary has \(Q\) accepted suggestions}{
  re-propose the value of the maximum positive key, or advance its persistent
  cursor and propose the next locally justified fresh candidate
}
\Upon{the first proposal of the primary of \((u,r)\), for \(R\) with key
  \(q\)}{
  store it and stop the no-proposal timer; once \(r\) is current and the
  proposal is lock-compatible (rechecked as proofs, backing, target data, or
  the current view change): if fresh (\(q=0\)), require the
  target-history/data/stance checks and attach the origin bit; if carried
  (\(q>0\)), require \(\textsc{Backed}(u,R)\); broadcast the resolver ECHO
  for \((u,r,q,\hash{R})\)
}
\Upon{\(Q\) ECHOs matching \((u,r,0,\hash{R})\), and \(f+1\)
  origin-1 ECHOs if \(R\) is non-skip}{
  broadcast \(\langle\reswit,u,r,\hash{R},R\rangle\)
}
\Upon{\(f+1\) matching witnesses}{relay the own witness once}
\Upon{\(Q\) matching witnesses}{latch \(\textsc{Backed}(u,R)\)}
\Upon{\(Q\) ECHOs matching the stored proposal \((u,r,q,\hash{R})\), and
  \(\textsc{Backed}(u,R)\) if \(q=0\)}{
  update the KEY1 state and broadcast KEY1
}
\Upon{successive \(Q\) KEY1, KEY2, KEY3, and LOCK statements}{
  update the IT-HS keys/lock and emit the next phase, ending with DONE
}
\Upon{\(f+1\) matching DONEs}{relay the own DONE once}
\Upon{\(Q\) matching DONEs for \(R\)}{
  decide \(R\) for \(u\) and submit its outcome to \(\rstat{try-seal}(u,\cdot)\)
}
\Fn{\(\rstat{try-seal}(u,X)\)}{
  \If{\(u\) has no terminal outcome}{persist \(X\), authorize its named lane
  roots, fire \(\rstat{seal}(u)\to X\), and run \textsc{OutputCursor}}
}
\Proc{\textsc{OutputCursor}$()$}{
  in view order: stop at a view that is neither completed nor sealed; at a
  completed or non-skip-sealed view, retrieve its lane prefixes and emit the
  core once; stop while the tip is open; after a full, core, or skip seal
  append the tip, omit it, or emit nothing, then advance; rerun when named
  lane prefixes arrive
}
\end{algorithm*}

\subsection{Concurrent target-local resolver}\label{sec:commit-rule}

Resolver views are separate from data-plane views.  For target $u$, view~1's
primary is the committee member immediately after $\proposer(u)$; later views
continue round robin.  A party acts on a resolver view's proposal, ECHO, KEY,
and LOCK quorums only while that view is current (a quorum counted before
entry fires on entry); witness, \backed{}, WISH, and DONE relay and decision are
view-independent.  On entry, parties send the primary their IT-HS KEY3
suggestions and broadcast KEY1 lock-opening proofs.  A positive suggestion
attaches its bounded value and must satisfy \textsc{AcceptKey} and
\textsc{Backed}.  The primary re-proposes a maximum positive key; when all
accepted keys are zero, it advances a persistent target-local cursor and
proposes the next locally justified candidate.  A $5\Delta$ timer changes view
only if no proposal arrived, avoiding the full timeout for a silent primary;
after a proposal, the $11\Delta$ IT-HS view timer remains.

The \textsc{Backed} step makes external validity stable.  Checking only a
local $f+1$ subset would find a correct origin but would not make the value
stable: the Byzantine senders among those witnesses may never repeat them, so
a party that joins later could not reproduce the check.  In contrast, a correct
party that counts $Q$ witnesses counted at least $n-2f\ge f+1$ correct
senders.  Their broadcasts make every correct party relay and eventually
count $Q$.  A carried value can therefore be validated from
\textsc{Backed}, rather than by rerunning a target-local predicate that may
differ across parties.  The IT-HS key and lock argument then gives agreement
and view-change safety.  \Cref{app:commit-rule} makes that argument
self-contained: same-view uniqueness and persistent-lock safety prove
agreement, while persistent KEY2 and KEY1 histories prove positive-key
acceptance and lock opening.

Across $R$ entered resolver views for one target, the protocol and at-most-once
catch-up use $O(Rn^2)$ distinct logical messages and $O(Rn^3\lambda)$ bits in
the proved by-value encoding; retransmitting the same initial WISH before entry
adds one all-to-all send per timer expiration.  The costs of $k$ overlapping
targets add.  There is no global resolution queue, but
we prove neither a bound on simultaneously active targets nor a constant
general target latency; \Cref{sec:resolver-evaluation} measures the resolver
under a sustained mixed-open attack, and \Cref{sec:discussion} states what that
leaves open.

\begin{theorem}[Target-local resolver summary]\label{thm:main-direct-resolver}
For each target, no two correct parties decide different resolver outcomes.
If one correct party decides $R$, every correct party eventually either decides
$R$ or already holds its compatible terminal seal.  Under partial
synchrony, every unresolved target eventually obtains one common valid full,
core, or skip outcome at all correct parties.  Every non-skip decision is author-backed and
future-compatible with Direct AGB.
\end{theorem}

\subsection{End-to-end guarantees}\label{sec:correctness}

\begin{theorem}[Vantage correctness]\label{thm:main-correctness}
Safety holds without timing assumptions: every correct party seals each view
at most once, no two correct parties seal different outcomes, and their output
sequences are always prefix-comparable.  Every output block is valid,
published with its lane prefix by its encoded author to a correct party,
permanently retrievable, and obtained before local output.  Under partial
synchrony, the fixed static fault set, fair repair, and the non-Zeno execution
assumption, every view is eventually sealed, every finite common output prefix
is eventually produced everywhere, and every valid block of a correct author
is eventually included.
\end{theorem}

The proof separates compatibility across sealing mechanisms from totality.  Quorum
intersection makes direct results agree; optimistic locks exclude conflicts
between fast-seal and resolver submissions; persistent stances exclude a skip-vote quorum
from any fresh ECHO quorum for a non-skip resolver value.  Stable backing and
IT-HS make each target's resolver choice identical at all correct parties;
targets need no relative order because the prefix cursor already processes
data-plane views in order.
For a clean crash-only silent view whose first correct entry is $s\ge GST$ and
whose correct entries finish by $s+2\delta$, all correct parties skip-seal by
$s+4\Delta+3\delta$; simultaneous entry gives $4\Delta+\delta$.  This is a
per-view bound, not a bound for a burst of failures.
\Cref{sec:resolver-evaluation} measures both a transient crash and the
mixed-open case.  Full statements and proofs appear in \Cref{app:correctness}.

\section{Implementation and Evaluation}\label{sec:evaluation}

\subsection{Implementation and methodology}\label{sec:implementation}\label{sec:methodology}

We implement Vantage in Rust.  The implementation is public at
\url{https://github.com/polinikita/vantage-bft}, and the benchmark harness
with its configurations and recorded campaigns at
\url{https://github.com/polinikita/wan-bench}.
The binary is derived from the Autobahn artifact~\cite{autobahn}.  The two
Autobahn variants (optimistic all-to-all and seamless) extend that
artifact's own protocol implementation, and Vantage and the two Simple-IT
variants (Bracha and optimistic reliable broadcast) are added as protocol
modules, so all five share one Transmission Control Protocol (TCP)
transport, batching, worker, storage, client, and metrics stack.
A deterministic task handles the protocol while networking and storage run on
separate tasks.  Bluestreak and Sailfish++ run from the Starfish
artifact~\cite{starfish} at \url{https://github.com/iotaledger/starfish}.  All campaigns use canonical one-byte committee
identifiers (compact IDs).  A matched check at $n=30$ finds that the inactive
target-local resolver leaves throughput, latency, and CPU unchanged
(\Cref{app:resolver-load}).

The shared transport realizes the model's authenticated channels with a
pairwise keyed-hash tag on each frame.  Tags are verified on arrival and never
become protocol evidence.  The current implementation also checks that counted
sender and proposer fields match the authenticated connection identity.  Every
evaluated workload emits identity-consistent fields, so this validation does
not change the protocol path measured by the retained campaigns.  The results
include channel-authentication cost; \Cref{sec:discussion} reports a separate
cost check.

Unless stated otherwise, cloud runs place one \texttt{c5d.2xlarge} validator in
the Amazon Web Services (AWS) \texttt{eu-west-1a} availability zone and inject
the ten-region round-trip time (RTT) matrix of \Cref{tab:rtt-matrix} over
private addresses.  The netem rules add delay only and cap no bandwidth, so
the fleet keeps its single-zone datacenter bandwidth: the throughput results
combine WAN latency with datacenter bandwidth.  All shared-binary variants
use the same fleet.  Sailfish++
also runs there; Bluestreak uses \texttt{m5d.2xlarge} because it exhausts memory
on the shared instance type before reaching its highest accepted load
(\Cref{app:baseline-fleet}).  Q4 and Q5 are controlled native-Linux studies.

Transactions are random 512-byte values submitted open loop.  Each cloud run
warms up for 30 seconds and measures for 120 seconds.  A transaction is
\emph{materialized} at a validator when its block is sequenced in that
validator's output and the block's bytes are held locally.  We report the
median validator's submission-to-materialization latency (end-to-end
latency), central processing unit (CPU) use, and outbound framed bytes per
sequenced byte.  Cloud cells are
medians across the available repetitions of up to three fresh-fleet campaigns.
When measured throughput is within roughly one percent of offered load, we say
that the protocol sustains the load: small over- and undershoots reflect work
crossing the measurement-window boundary.  The benchmarking harness, exact
hardware, repetition counts, and metric definitions appear in
\Cref{app:evaluation}; the harness repository contains the corresponding
configurations and records.

The implementation uses digest naming, positional acknowledgment bitmaps, and
transport batching.  These encoding choices do not change the protocol rules.
The common encoding uses $O(n^2\lambda+n^3)$ control bits; sparse claims and
by-value recovery remain within the proved $O(n^3\lambda)$ worst case
(\Cref{sec:optimizations}).

Q1 (\Cref{sec:committee-scaling}) and Q2 (\Cref{sec:throughput-sweep}) test
the end-to-end claim of \Cref{sec:introduction} --- lowest p50 latency at
every accepted load --- across committee sizes at low load and under
increasing load at $n=100$; in both, no view uses a fallback decision.  Q3
(\Cref{sec:leader-relay}) measures the leader-relay cost that
\Cref{sec:closest-comparisons} attributes to admitting data a leader must
serve.  Q4 (\Cref{sec:decision-anatomy}) follows one fault-free view, checking
the stage order that the $3\delta$/$4\delta$ endpoints of
\Cref{tab:data-path-comparison} assume and measuring the fast/READY split.  Q5
(\Cref{sec:resolver-evaluation}) asks whether the recovery path of
\Cref{sec:protocol} --- the grounded skip shortcut for crashed proposers and
the target-local resolver for mixed-open views --- keeps pace with crash and
mixed-open faults.

\subsection{Q1: Fault-free committee scaling}\label{sec:committee-scaling}

We run seven protocol variants at $n\in\{10,20,50,100\}$ under a fixed aggregate
100 tx/s.  All variants sustain approximately the offered load.

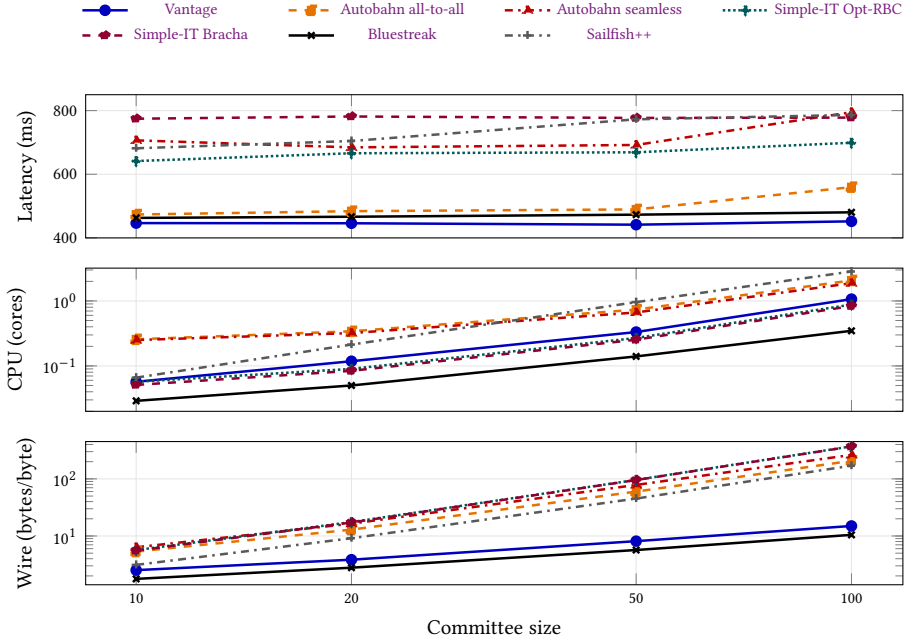
\begin{figure}[t]
  \centering
  \begin{tikzpicture}
\begin{groupplot}[
  group style={group size=1 by 3, vertical sep=0.40cm},
  width=0.82\textwidth,
  height=0.23\textwidth,
  xmode=log,
  log basis x=10,
  xmin=8.5,
  xmax=118,
  xtick={10,20,50,100},
  xticklabels={10,20,50,100},
  grid=major,
  grid style={black!10},
  tick label style={font=\scriptsize},
  label style={font=\small},
  title style={font=\small},
  legend columns=4,
  legend style={font=\scriptsize,draw=none,fill=none,/tikz/every even column/.append style={column sep=0.45cm}},
  every axis plot/.append style={line width=0.95pt,mark size=1.7pt},
]
\nextgroupplot[
  legend to name=committee-scaling-legend,
  ylabel={Latency (ms)},
  ymin=400,
  ymax=850,
  xticklabels={},
]
  \addplot[blue!75!black,solid,mark=*] table[x=n,y=vlat] {figures/committee-scaling.dat};
  \addlegendentry{Vantage}
  \addplot[orange!90!black,dashed,mark=square*] table[x=n,y=aolat] {figures/committee-scaling.dat};
  \addlegendentry{Autobahn all-to-all}
  \addplot[red!75!black,dashdotted,mark=triangle*] table[x=n,y=aslat] {figures/committee-scaling.dat};
  \addlegendentry{Autobahn seamless}
  \addplot[teal!70!black,densely dotted,mark=diamond*] table[x=n,y=solat] {figures/committee-scaling.dat};
  \addlegendentry{Simple-IT Opt-RBC}
  \addplot[purple!75!black,dashed,mark=pentagon*] table[x=n,y=sblat] {figures/committee-scaling.dat};
  \addlegendentry{Simple-IT Bracha}
  \addplot[black,solid,mark=x] table[x=n,y=bllat] {figures/committee-scaling.dat};
  \addlegendentry{Bluestreak}
  \addplot[gray!70!black,dashdotted,mark=+] table[x=n,y=sflat] {figures/committee-scaling.dat};
  \addlegendentry{Sailfish++}

\nextgroupplot[
  ylabel={CPU (cores)},
  ymode=log,
  log basis y=10,
  ymin=0.02,
  ymax=3.2,
  xticklabels={},
]
  \addplot[blue!75!black,solid,mark=*] table[x=n,y=vcpu] {figures/committee-scaling.dat};
  \addplot[orange!90!black,dashed,mark=square*] table[x=n,y=aocpu] {figures/committee-scaling.dat};
  \addplot[red!75!black,dashdotted,mark=triangle*] table[x=n,y=ascpu] {figures/committee-scaling.dat};
  \addplot[teal!70!black,densely dotted,mark=diamond*] table[x=n,y=socpu] {figures/committee-scaling.dat};
  \addplot[purple!75!black,dashed,mark=pentagon*] table[x=n,y=sbcpu] {figures/committee-scaling.dat};
  \addplot[black,solid,mark=x] table[x=n,y=blcpu] {figures/committee-scaling.dat};
  \addplot[gray!70!black,dashdotted,mark=+] table[x=n,y=sfcpu] {figures/committee-scaling.dat};

\nextgroupplot[
  ylabel={Wire (bytes/byte)},
  xlabel={Committee size},
  ymode=log,
  log basis y=10,
  ymin=1.4,
  ymax=450,
]
  \addplot[blue!75!black,solid,mark=*,unbounded coords=discard]
    table[x=n,y=vwire] {figures/committee-scaling.dat};
  \addplot[orange!90!black,dashed,mark=square*] table[x=n,y=aowire] {figures/committee-scaling.dat};
  \addplot[red!75!black,dashdotted,mark=triangle*] table[x=n,y=aswire] {figures/committee-scaling.dat};
  \addplot[teal!70!black,densely dotted,mark=diamond*] table[x=n,y=sowire] {figures/committee-scaling.dat};
  \addplot[purple!75!black,dashed,mark=pentagon*] table[x=n,y=sbwire] {figures/committee-scaling.dat};
  \addplot[black,solid,mark=x] table[x=n,y=blwire] {figures/committee-scaling.dat};
  \addplot[gray!70!black,dashdotted,mark=+] table[x=n,y=sfwire] {figures/committee-scaling.dat};
\end{groupplot}
\node at ([yshift=0.95cm]group c1r1.north) {\ref{committee-scaling-legend}};
\end{tikzpicture}
  \caption{Fault-free committee scaling at 100 transactions per second.
  Most cells are medians across three fresh-fleet repetitions; a few baseline
  cells use fewer.  The $n{=}100$ cell joins from the throughput
  campaign's 100~tx/s point, and instance types are listed
  in \Cref{app:baseline-fleet}.  Wire cost is bytes sent per sequenced
  byte.  The committee axis and the processor and wire panels are
  logarithmic.}
  \Description{Three plots compare latency, processor use, and wire efficiency
  as the committee grows from 10 to 100.  Vantage has the lowest median
  latency and Bluestreak the lowest processor and wire cost.}
  \label{fig:committee-scaling}
\end{figure}

Vantage has the lowest p50 and p99 latency at every committee size
(\Cref{fig:committee-scaling,tab:latency-percentiles}); its p50 stays near
0.45\,s as the committee
grows from 10 to 100.  At $n=100$ it is about 30\,ms faster than Bluestreak at
both p50 and p99, while using roughly three times its CPU and 1.4 times its wire
cost.  Against optimistic all-to-all Autobahn, Vantage is about 0.1\,s faster,
uses about half the CPU, and sends roughly one fourteenth as many bytes.  The
wire difference comes mainly from the per-lane availability proofs in each
Autobahn consensus cut.

Progress cadence helps explain the latency ordering.  At $n=100$, Vantage
advances through about 14 proposal views per second, compared with about 10 DAG
rounds for Bluestreak, about four for Sailfish++, and about eight consensus
slots for each Autobahn variant (\Cref{app:cadence}).  These units are not equivalent work;
they identify only the next opportunity for a fresh data reference to enter
consensus.  More frequent opportunities shorten that initial wait, while
admission rules and the later commit path account for the remaining difference.
The pipelined cadence is the run-time form of the one-delay proposal spacing in
\Cref{fig:commit-chain}.
Summed over validators and repetitions, the all-$n$ fast seal supplies
67--70\% of Vantage's local seals at $n\in\{10,20,50\}$ and 72--73\% at
$n=100$; the rest are Direct AGB READY seals (\Cref{sec:primitive}).

\subsection{Q2: Throughput scaling at \texorpdfstring{$n=100$}{n=100}}
\label{sec:throughput-sweep}

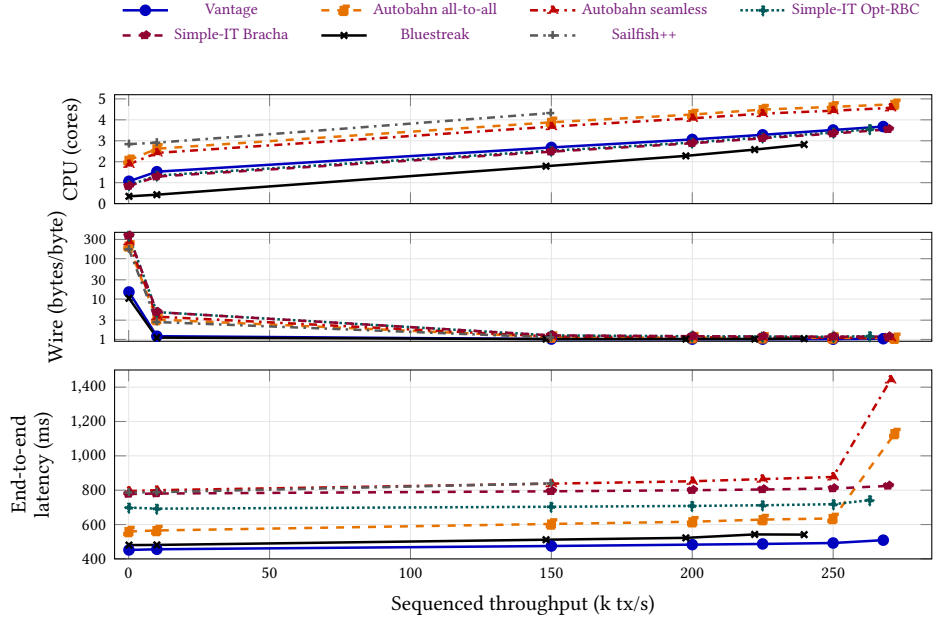
\begin{figure}[t]
  \centering
  \begin{tikzpicture}
\begin{groupplot}[
  group style={group size=1 by 3, vertical sep=0.38cm},
  width=0.82\textwidth,
  height=0.20\textwidth,
  xmin=-5,
  xmax=285,
  xtick={0,50,100,150,200,250},
  grid=major,
  grid style={black!10},
  tick label style={font=\scriptsize},
  label style={font=\small},
  legend columns=4,
  legend style={font=\scriptsize,draw=none,fill=none,
    /tikz/every even column/.append style={column sep=0.35cm}},
  every axis plot/.append style={line width=0.95pt,mark size=1.7pt},
  unbounded coords=jump,
]
\nextgroupplot[
  legend to name=throughput-sweep-legend,
  ylabel={CPU (cores)},
  ymin=0,
  ymax=5.2,
  ytick={0,1,2,3,4,5},
  xticklabels={},
]
  \addplot[blue!75!black,solid,mark=*] table[x=vtps,y=vcpu] {figures/throughput-sweep.dat};
  \addlegendentry{Vantage}
  \addplot[orange!90!black,dashed,mark=square*] table[x=aotps,y=aocpu] {figures/throughput-sweep.dat};
  \addlegendentry{Autobahn all-to-all}
  \addplot[red!75!black,dashdotted,mark=triangle*] table[x=astps,y=ascpu] {figures/throughput-sweep.dat};
  \addlegendentry{Autobahn seamless}
  \addplot[teal!70!black,densely dotted,mark=diamond*] table[x=sotps,y=socpu] {figures/throughput-sweep.dat};
  \addlegendentry{Simple-IT Opt-RBC}
  \addplot[purple!75!black,dashed,mark=pentagon*] table[x=sbtps,y=sbcpu] {figures/throughput-sweep.dat};
  \addlegendentry{Simple-IT Bracha}
  \addplot[black,solid,mark=x] table[x=bltps,y=blcpu] {figures/throughput-sweep.dat};
  \addlegendentry{Bluestreak}
  \addplot[gray!70!black,dashdotted,mark=+] table[x=sftps,y=sfcpu] {figures/throughput-sweep.dat};
  \addlegendentry{Sailfish++}

\nextgroupplot[
  ylabel={Wire (bytes/byte)},
  ymode=log,
  log basis y=10,
  ymin=0.9,
  ymax=450,
  ytick={1,3,10,30,100,300},
  yticklabels={1,3,10,30,100,300},
  xticklabels={},
]
  \addplot[blue!75!black,solid,mark=*] table[x=vtps,y=veff] {figures/throughput-sweep.dat};
  \addplot[orange!90!black,dashed,mark=square*] table[x=aotps,y=aoeff] {figures/throughput-sweep.dat};
  \addplot[red!75!black,dashdotted,mark=triangle*] table[x=astps,y=aseff] {figures/throughput-sweep.dat};
  \addplot[teal!70!black,densely dotted,mark=diamond*] table[x=sotps,y=soeff] {figures/throughput-sweep.dat};
  \addplot[purple!75!black,dashed,mark=pentagon*] table[x=sbtps,y=sbeff] {figures/throughput-sweep.dat};
  \addplot[black,solid,mark=x] table[x=bltps,y=bleff] {figures/throughput-sweep.dat};
  \addplot[gray!70!black,dashdotted,mark=+] table[x=sftps,y=sfeff] {figures/throughput-sweep.dat};

\nextgroupplot[
  height=0.27\textwidth,
  ylabel={\shortstack{End-to-end\\latency (ms)}},
  xlabel={Sequenced throughput (k tx/s)},
  ymin=400,
  ymax=1500,
  ytick={400,600,800,1000,1200,1400},
]
  \addplot[blue!75!black,solid,mark=*] table[x=vtps,y=vlat] {figures/throughput-sweep.dat};
  \addplot[orange!90!black,dashed,mark=square*] table[x=aotps,y=aolat] {figures/throughput-sweep.dat};
  \addplot[red!75!black,dashdotted,mark=triangle*] table[x=astps,y=aslat] {figures/throughput-sweep.dat};
  \addplot[teal!70!black,densely dotted,mark=diamond*] table[x=sotps,y=solat] {figures/throughput-sweep.dat};
  \addplot[purple!75!black,dashed,mark=pentagon*] table[x=sbtps,y=sblat] {figures/throughput-sweep.dat};
  \addplot[black,solid,mark=x] table[x=bltps,y=bllat] {figures/throughput-sweep.dat};
  \addplot[gray!70!black,dashdotted,mark=+] table[x=sftps,y=sflat] {figures/throughput-sweep.dat};
\end{groupplot}
\node at ([yshift=0.95cm]group c1r1.north) {\ref{throughput-sweep-legend}};
\end{tikzpicture}
  \caption{Throughput sweep at $n=100$.  Cells are medians across campaign
  repetitions on fresh fleets, on the fleets of \Cref{app:baseline-fleet}.
  The x-axis is measured sequenced throughput; end-to-end latency is
  submission-to-materialization p50.  The last point of each series is its
  highest accepted load; the next offered rung overloads and is not
  plotted.}
  \Description{CPU, wire efficiency, and latency are plotted against sequenced
  throughput.  Vantage retains the lowest p50 latency, with wire cost inside
  the band the accepted variants converge to at high loads; Bluestreak uses
  the least processor time.}
  \label{fig:throughput-sweep}
\end{figure}

We increase offered load from 100 to 300,000 tx/s at $n=100$.  A point is
accepted when median sequenced throughput reaches at least 95\% of offered
load; the first lower point is the
overload endpoint, reported in the text and not plotted.

Vantage has the lowest p50 latency at every accepted offered rate
(\Cref{fig:throughput-sweep}), remaining between about 0.45 and 0.51\,s.  It sustains 250k offered tx/s at about 0.49\,s
p50; at its highest accepted point it sequences about 268k tx/s at
0.51\,s before overloading at 300k.  Optimistic Autobahn reaches a similar
throughput, about 272k tx/s, but its p50 exceeds 1.1\,s.  Bluestreak reaches
about 240k tx/s at 0.54\,s before overloading, and Sailfish++ overloads at a
lower offered load.  Seamless Autobahn and both Simple-IT variants also accept
the 275k point and overload at 300k.  The shared-binary variants thus reach
nearly the same accepted peak, 263--272k, and fail at the same offered load
despite different consensus logic; we attribute this endpoint to the shared
transport path on this instance type rather than to the protocols, so
throughput parity among them is expected and latency is the discriminating
measurement.

Resource cost separates the designs.  Bluestreak uses the least CPU at every
accepted point.  It also sends the fewest bytes at low load; from 150k upward,
the accepted variants converge toward one wire byte per sequenced byte, with
only small differences among them.

The tail follows the same ordering.  Vantage stays near 0.65\,s p99 from 150k
through 250k.  Over the same interval Bluestreak rises from about 0.8 to
1.5\,s, while optimistic Autobahn reaches about 3.3\,s at its highest accepted
point (\Cref{tab:latency-percentiles}).  These are deployment measurements,
not protocol capacity bounds.

Across the three repetitions, the all-$n$ fast seal supplies 72--76\% of
Vantage's local seals at every accepted load through 250k and 41--45\% at
the 275k point; the remainder are Direct AGB READY seals, and no view in
these runs used a skip or resolver seal.

\subsection{Q3: What does optimistic leader relay cost?}
\label{sec:leader-relay}

We isolate the work created when a leader admits data before it is widely
disseminated, the leader-load effect of \Cref{sec:closest-comparisons}.  At
$n=20$ and $f=6$, every author receives the same share of a
uniform offered load $R$.  Each of six Byzantine publishers creates one batch
per $\Delta=200$\,ms, publishes it to five correct holders, keeps its header
visible, and refuses repair.  Including the publisher, each attacked lane has
six direct holders, one short of the $f+1=7$ availability threshold.  Correct
holders serve normally, and the groups cover every correct consensus leader.
There is no equivocation, invalid proposal, or extra load; correct authors
therefore supply the reachable $0.7R$ workload.

We compare optimistic all-to-all Autobahn, Vantage, and Simple-IT Opt-RBC.
Seamless Autobahn is outside this test because sub-threshold tips never enter
its certified path.  Optimistic Autobahn admits a locally held tip and makes
the leader serve missing bytes before replicas vote~\cite{autobahn}.  Vantage
can reject that tip while retaining its proposal core, whereas Simple-IT waits
for its availability path before admission.  The figure therefore separates
correct-author goodput from the attacked-lane data that Autobahn additionally
orders.

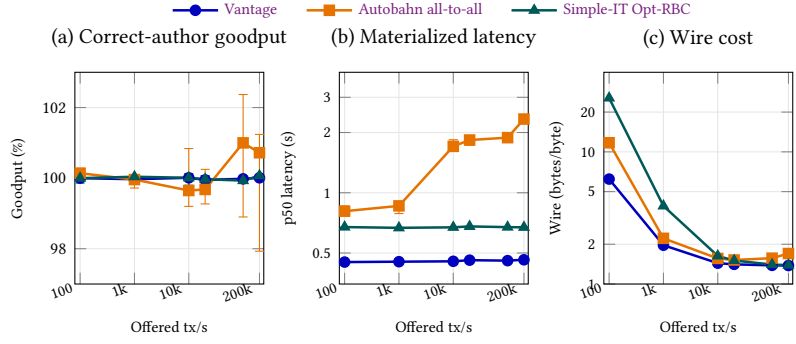
\begin{figure}[t]
  \centering
  \begin{tikzpicture}
\begin{groupplot}[
  group style={group size=3 by 1, horizontal sep=1.0cm},
  width=0.27\textwidth,
  height=0.29\textwidth,
  xmode=log,
  log basis x=10,
  xmin=80,
  xmax=240000,
  xtick={100,1000,10000,200000},
  xticklabels={100,1k,10k,200k},
  x tick label style={font=\scriptsize,rotate=20,anchor=east},
  grid=major,
  grid style={black!10},
  tick label style={font=\scriptsize},
  label style={font=\scriptsize},
  title style={font=\small},
  every axis plot/.append style={line width=0.95pt,mark size=1.7pt},
  legend columns=3,
  legend style={font=\scriptsize,draw=none,fill=none,
    /tikz/every even column/.append style={column sep=0.30cm}},
]
\nextgroupplot[
  title={(a) Correct-author goodput},
  ylabel={Goodput (\%)},
  xlabel={Offered tx/s},
  ymin=97,
  ymax=103,
  ytick={98,100,102},
  legend to name=leader-relay-legend,
]
  \addplot[blue!75!black,solid,mark=*,
    error bars/.cd,y dir=both,y explicit]
    table[x=offered,y=vgood,y error minus=vgoodm,y error plus=vgoodp] {figures/leader-relay.dat};
  \addlegendentry{Vantage}
  \addplot[orange!90!black,solid,mark=square*,
    error bars/.cd,y dir=both,y explicit]
    table[x=offered,y=optgood,y error minus=optgoodm,y error plus=optgoodp] {figures/leader-relay.dat};
  \addlegendentry{Autobahn all-to-all}
  \addplot[teal!70!black,solid,mark=triangle*,
    error bars/.cd,y dir=both,y explicit]
    table[x=offered,y=sgood,y error minus=sgoodm,y error plus=sgoodp] {figures/leader-relay.dat};
  \addlegendentry{Simple-IT Opt-RBC}

\nextgroupplot[
  title={(b) Materialized latency},
  xlabel={Offered tx/s},
  ylabel={p50 latency (s)},
  ymode=log,
  log basis y=10,
  ymin=0.35,
  ymax=4,
  ytick={0.5,1,2,3},
  yticklabels={0.5,1,2,3},
]
  \addplot[blue!75!black,solid,mark=*,
    error bars/.cd,y dir=both,y explicit]
    table[x=offered,y=vp50,y error minus=vp50m,y error plus=vp50p] {figures/leader-relay.dat};
  \addplot[orange!90!black,solid,mark=square*,
    error bars/.cd,y dir=both,y explicit]
    table[x=offered,y=opp50,y error minus=opp50m,y error plus=opp50p] {figures/leader-relay.dat};
  \addplot[teal!70!black,solid,mark=triangle*,
    error bars/.cd,y dir=both,y explicit]
    table[x=offered,y=sp50,y error minus=sp50m,y error plus=sp50p] {figures/leader-relay.dat};

\nextgroupplot[
  title={(c) Wire cost},
  xlabel={Offered tx/s},
  ymode=log,
  log basis y=10,
  ymin=1,
  ymax=40,
  ytick={1,2,5,10,20},
  yticklabels={1,2,5,10,20},
  ylabel={Wire (bytes/byte)},
]
  \addplot[blue!75!black,solid,mark=*]
    table[x=offered,y=vwire] {figures/leader-relay.dat};
  \addplot[orange!90!black,solid,mark=square*]
    table[x=offered,y=optwire] {figures/leader-relay.dat};
  \addplot[teal!70!black,solid,mark=triangle*]
    table[x=offered,y=swire] {figures/leader-relay.dat};
\end{groupplot}
\node at ([yshift=0.82cm]group c2r1.north) {\ref{leader-relay-legend}};
\end{tikzpicture}
  \caption{Leader-relay stress at $n=20,f=6$; each point is the median of
  up to three campaign repetitions, and whiskers
  span the repetitions.  Panel (a) reports correct-author throughput as a
  fraction of the reachable $0.7R$ load; points scatter within about a percent
  of 100\% because work in flight at the window boundary is counted on one
  side of it or the other.  Panel (b) shows materialization p50 in seconds;
  panel (c) shows wire bytes per sequenced byte.}
  \Description{Three plots compare correct-author goodput, materialization
  latency, and wire cost as offered load rises.  All protocols retain
  correct-author goodput at every offered load; optimistic Autobahn commits the
  attacked-lane payload and its latency rises sharply, while Vantage and
  Simple-IT reject that data and remain flat.}
  \label{fig:leader-relay}
\end{figure}

All three protocols sustain the full correct-author workload at every offered
rate (\Cref{fig:leader-relay}).  At 200k tx/s this is 140k correct-author tx/s.  Optimistic Autobahn also
orders nearly all attacked-lane data because a leader that admitted those tips
must relay them; Vantage and Simple-IT order none of it.  Autobahn's p50 reaches
about 2.3\,s, compared with 0.46\,s for Vantage and 0.67\,s for Simple-IT, and
it uses roughly twice their CPU.  The result exposes an admission-policy
tradeoff rather than a loss of correct-author goodput.

The tail shows where each policy pays (\Cref{tab:latency-percentiles}).
Simple-IT remains near 0.86\,s p99 because unavailable data never enters
consensus.  Vantage's median changes little, but its p99 rises to roughly
1.2--1.4\,s as a rejected tip returns through a later core --- the
tip-spoiling cost discussed in \Cref{sec:discussion}.  Optimistic
Autobahn's median and tail both grow because the admitted tips must be relayed;
at 200k its p99 is about 5.7\,s.  These are measurements of this withholding
profile, not protocol capacity bounds.

\subsection{Q4: Inside one Vantage view}\label{sec:decision-anatomy}

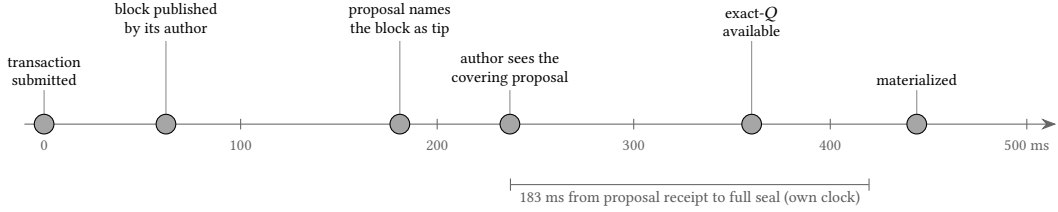
\begin{figure}[t]
  \centering
\begin{tikzpicture}[
    x=0.026cm, y=1cm,
    >={Stealth[length=2mm]},
    ev/.style={draw, circle, fill=black!35, minimum size=2.6mm, inner sep=0pt},
    lab/.style={font=\scriptsize, align=center, inner sep=1pt},
    lead/.style={black!45, line width=0.4pt},
    tick/.style={font=\scriptsize, text=black!60},
]
  \draw[->, black!60] (-10,0) -- (515,0);
  \foreach \t/\l in {0/0,100/100,200/200,300/300,400/400,500/500 ms} {
    \draw[black!60] (\t,0.07) -- (\t,-0.07);
    \node[tick, anchor=north] at (\t,-0.1) {\l};
  }
  \foreach \x/\lvl in {0/0.45,237/0.45,62/1.1,181/1.1,360/1.1,444/0.45}
    \draw[lead] (\x,0.13) -- (\x,\lvl);
  \node[ev] at (0,0) {};   \node[lab, anchor=south] at (0,0.47) {transaction\\submitted};
  \node[ev] at (62,0) {};  \node[lab, anchor=south] at (62,1.12) {block published\\by its author};
  \node[ev] at (181,0) {}; \node[lab, anchor=south] at (181,1.12) {proposal names\\the block as tip};
  \node[ev] at (237,0) {}; \node[lab, anchor=south] at (237,0.47) {author sees the\\covering proposal};
  \node[ev] at (360,0) {}; \node[lab, anchor=south] at (360,1.12) {exact-$Q$\\available};
  \node[ev] at (444,0) {}; \node[lab, anchor=south] at (444,0.47) {materialized};
  \draw[|-|, black!60] (237,-0.8) -- (420,-0.8)
    node[midway, below=1pt, lab, text=black!65]
    {183 ms from proposal receipt to full seal (own clock)};
\end{tikzpicture}
  \caption{One transaction's path in the local $n=20$ study at 100~tx/s:
  p50 stage latencies chained on one axis from submission.  Publication
  follows submission by about 62~ms
  (blocks are published every 100~ms); the send of the first proposal naming
  the block as tip (119~ms), the author's receipt of that proposal (175~ms),
  and exact-$Q$ availability (298~ms) are measured from publication;
  materialization is the direct
  submission-to-materialization median, 444~ms.  p99 values are in
  \Cref{app:decision-anatomy}.}
  \Description{A single timeline in milliseconds after submission with marks
  at submission, block publication, the send of the first proposal naming the
  block as tip, the author's receipt of that proposal, exact-quorum
  availability, and materialization, and a bracket for the
  proposal-receipt-to-full-seal interval.}
  \label{fig:vantage-decision-breakdown}
\end{figure}

To check directly that a proposal can cover a block before the block is
exact-$Q$-available, we run a local $n=20$ study at 100~tx/s with the same
RTT matrix injected in process (\Cref{app:decision-anatomy}).  Each of three
60-second runs materializes
all 6,000 submitted transactions after drain.  Every data-plane outcome is
full, and no resolver message or resolver seal appears.

\Cref{fig:vantage-decision-breakdown} chains the stage medians on one
timeline.  A transaction's block is published about 62~ms after submission.
The first proposal that names the block as tip is sent 119~ms after
publication --- all validators of this single-process run share one clock, so
the send is measured, not inferred --- and the author sees that proposal
175~ms after publication, before the block becomes exact-$Q$-available at
298~ms; submission to materialization measures 444~ms directly.  Both waits
follow from two cadences.  Authors publish a lane block every 100~ms (the
header delay, \Cref{app:methodology}), so a transaction's batch waits about
50~ms on average for the next block --- the measured 51~ms from first digest
to publication.  Proposals follow one another about every 70~ms, one one-way
delay apart (\Cref{fig:commit-chain}; about 14 views per second in the
$n=100$ runs, \Cref{app:cadence}), so a block that has reached the proposer
waits for the next one: the measured 119~ms is the 72~ms median one-way delay
plus a 47~ms wait.  Proposal receipt to local full seal takes 183~ms.

The proved aligned endpoints --- $3\delta$ fast or $4\delta$ READY for a
tip, $5\delta$ or $6\delta$ for a core entry (\Cref{tab:data-path-comparison},
\Cref{sec:primitive}) --- are not converted into measured milliseconds: under
the RTT matrix $\delta$ differs per pair, and the measured intervals average
over arrival phase.

In the representative matrix run (\Cref{fig:vantage-traffic-shares}), payload
and lane traffic account for about half of typed bytes, while AGB contributes about 40\% of bytes and 60\% of
logical messages.  Sequence state and primary--worker traffic make up nearly
all of the remainder.  The inactive target-local resolver contributes no bytes
or messages.

In a matched network-delay control (\Cref{fig:vantage-rtt-shape}), the
all-$n$ fast-seal shortcut supplies about 70\% of local seals under the
heterogeneous matrix and every observed seal under a uniform RTT; all outcomes
in both experiments are full.  This is the run-time race behind the two
\sysname{} lines in \Cref{tab:data-path-comparison}: fast-seal waits for the
slowest of all $n$ ECHOs while the $4\delta$ READY path waits only for
quorums.  Heterogeneous delays therefore route the remaining 30\% through
READY, while uniform delays let all-to-all finish first.

\begin{figure}[t]
  \centering
  \includegraphics[width=\linewidth]{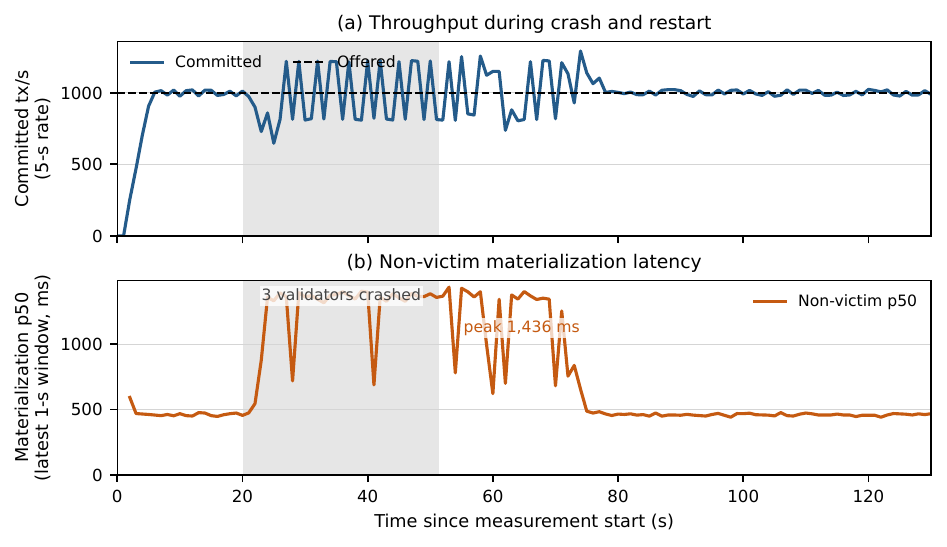}
  \caption{Pure transient crash at $n=10,f=3$ and 1,000 offered tx/s.
  The gray interval is the roughly 30-second outage.  Both curves are
  evaluated every second over the seven non-victims: panel (a) is their median
  trailing-5-second committed rate (sequenced transactions, as the axes
  label it), and panel (b) is their median
  latest-1-second materialization p50.  This is one diagnostic run, not a
  confidence interval.}
  \Description{Two time-series plots show the committed rate and one-second
  materialization p50 at the seven non-victims while three other validators are
  killed for about 30 seconds and restarted.  Throughput remains near the
  offered rate; latency rises during the outage and returns to its pre-crash level.}
  \label{fig:transient-crash}
\end{figure}

\subsection{Q5: Can recovery keep up?}\label{sec:resolver-evaluation}

\begin{figure}[t]
  \centering
  \includegraphics[width=\linewidth]{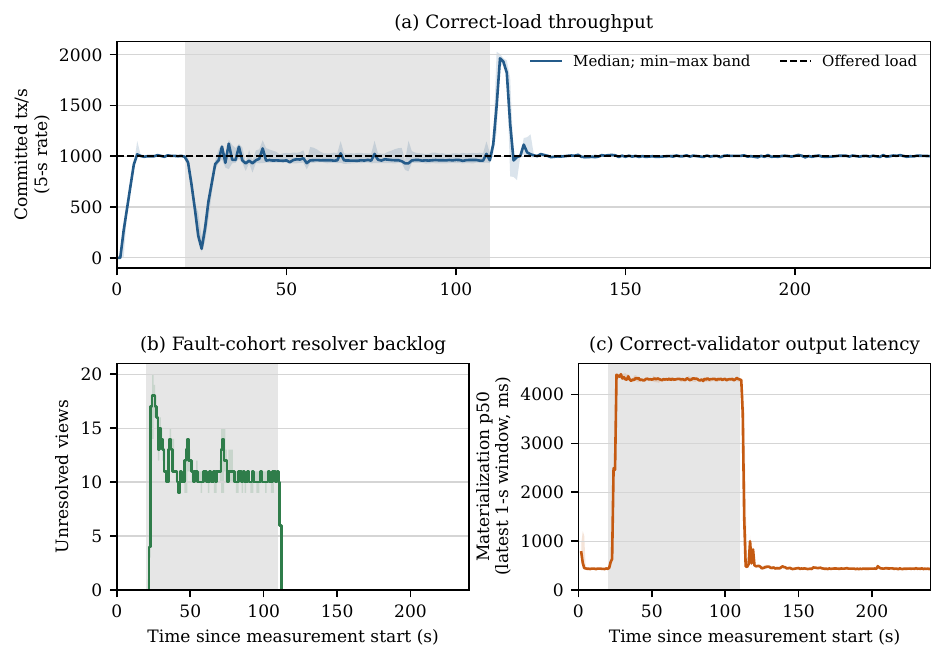}
  \caption{Sustained mixed-open recovery at $n=10,f=3$ over three fresh
  committees.  Each run first takes the median over the seven correct
  validators; curves then show the repetition median and min--max band.  The
  gray interval is Byzantine injection.  Panel (a) is the trailing-5-second
  correct-load committed rate, evaluated every second; panel (b) is the
  unresolved injected-target count; and panel (c) is the latest-1-second
  materialization p50.}
  \Description{Three time-series plots show correct-load throughput, the
  resolver backlog, and one-second p50 latency during a 90-second
  mixed-open attack.  The backlog remains around ten during injection and
  drains when injection stops; latency rises during injection and then returns
  to its earlier level.}
  \label{fig:direct-resolver-q5}
\end{figure}

We first isolate a pure transient crash at $n=10,f=3$.  The counted
1,000~tx/s load is placed on the seven non-victims; three otherwise-correct
validators are killed for about 30 seconds and then restarted with their
local state intact.  The native-Linux Docker deployment uses the same ten-region
netem matrix and $\Delta=200$~ms as the cloud campaigns.

As \Cref{fig:transient-crash} shows, the seven non-victims sustain the offered
1,000~tx/s before, during, and after the outage.  Their latest-one-second p50 is
about 0.46\,s before the crash, peaks near 1.4\,s shortly after restart, and
returns to about 0.46\,s.  Transactions delayed during the outage enter the
histogram only when they materialize, which is why the peak can follow restart;
the curve is a one-second measurement, not a cumulative or one-minute value.
Views whose proposer is a crashed victim are silent and take the skip shortcut
of \Cref{sec:protocol}.

\FloatBarrier

The harder test creates mixed-open targets continuously.  At $n=10,f=3$,
each of three Byzantine publishers proposes its own partially disseminated tip
and suppresses all of its protocol responses from seconds 20--110.  Correct
publishers offer the counted 1,000~tx/s; the faulty publishers add 600
uncounted tx/s.  We repeat the 240-second run on three fresh committees.

In \Cref{fig:direct-resolver-q5}, every injected mixed-open target seals at
every correct validator.  Targets arrive and seal at about four per second;
the unresolved backlog stays around ten rather than growing, then drains in
about two seconds after injection stops because target instances progress
concurrently.  Correct-load throughput sustains the offered 1,000~tx/s.
Targets seal everywhere in about 2.8\,s at the median and about 4\,s at the
slow end.  The latest-one-second materialization p50 rises from about 0.44\,s
to 4.3\,s while the attack is active and returns afterward.

The timing follows the resolver path.  A target activates when the proposal
horizon reaches $u+3$ (\Cref{sec:protocol}; \Cref{fig:resolution}).  Under this fault pattern it may then encounter up to
two silent resolver primaries, each consuming the $5\Delta=1$\,s no-proposal
timeout, before a correct-primary agreement completes.  The $11\Delta$
full-view timer bounds an unsuccessful proposed resolver view; a correct
primary can finish without waiting for that deadline, so it is not the main
delay in this experiment.  Because output remains in view order, a slow open
target also delays later transactions at the output cursor.  Faster activation,
faulty-primary turnover, or correct-primary agreement could reduce this
latency, but the present traces do not separate their individual costs.

These measurements answer the mechanism question for this attack: the
resolver keeps pace with the observed mixed-open arrival rate and drains the
remaining targets after the fault stops.  They do not bound an arbitrary
Byzantine arrival process, the number of simultaneously active targets, or a
target whose resolver repeatedly selects faulty primaries.

\FloatBarrier

\section{Discussion and Limitations}\label{sec:discussion}

\paragraph{Experimental scope.}
The cloud series report medians across up to three fresh-fleet repetitions; a
few baseline cells rest on fewer, and the transient-crash trace is one run.
The controlled deployment places validators in one AWS availability zone and
uses netem to reproduce a ten-region RTT matrix.  It therefore isolates latency
shape without modeling Internet loss or jitter.  The fault experiments (\Cref{sec:leader-relay,sec:resolver-evaluation}) cover
partial publication with repair refusal, transient crash/restart, and sustained
mixed-open views.

\paragraph{The cost of channel authentication.}
Our model assumes a reliable authenticated link for every pair, concretely a
pairwise MAC (\Cref{sec:model}); pairwise MACs in place of signatures go back
to Practical Byzantine Fault Tolerance (PBFT)~\cite{pbft}.  The transport implements it and our runs
keep it enabled (\Cref{sec:implementation}), so the reported numbers include
its cost.  That cost is small and bytes-proportional rather than
round-trip-bound.  In a single-host A/B --- ten validators in one process,
1000~tx/s, 512~B transactions, the ten-region RTT matrix, 180~s, two runs per
arm --- both configurations sustain the offered load at about 0.44\,s p50,
while authentication adds roughly 9\% CPU.  Its wire share depends on frame
size because batching amortizes one 16-byte tag: tags account for about 10\%
of the small frames at 100~tx/s and below 0.1\% of framed bytes at saturation.
The authenticated $n=100$ campaigns report no verification failure.

\paragraph{Resource contention and active targets.}
The Direct AGB homogeneous path does not logically wait for a resolver.  Each
unresolved target runs an independent instance, so different targets can
advance at network pace rather than waiting for one global service slot.  They
still share processor and network capacity with the data path, and an open tip
stops ordered output at its cursor position.  With $k$ active targets, the
resolver's message, wire, timer, and retained-state costs multiply by $k$.
Completed-open quarantine limits retries by correct proposers; it constrains
neither Byzantine proposers nor the rate at which they create target evidence.
The Q5 stress (\Cref{sec:resolver-evaluation}) shows matching arrival and seal rates for one $n=10$ attack, not
a bound under arbitrary arrivals or a saturation result.  A fixed target is
eventually selected and decided, but Byzantine primaries can consume resolver
views, so no constant general failed-view latency is claimed; only the clean
crash-only case of \Cref{cor:crash-skip} has a direct bound.

\paragraph{Formal retention.}
The protocol model retains the lane, target, and resolver evidence used by its
rules indefinitely.  Bounded-state recovery and reclamation are implementation
concerns outside the protocol and theorem statements.

\paragraph{Tip spoiling.}
One unverifiable entry makes the entire optional tip manifest fail a
receiver's local check, so a faulty author can prevent other authors' fresh
blocks in that manifest from being included.  The faulty author publishes its
lane-tip block only to the correct proposer and withholds matching
acknowledgments from every other correct party.  Repair may supply the bytes;
it cannot recreate direct proof of authorship.  Multimmit calls the same
withholding easy spoiling in Raptr's prefix voting~\cite{multimmit,raptr}.

A correct proposer's completion and core survive it.  For every $n\ge 3f{+}1$
with $f\ge1$, the proposer is the only correct party that can send grade $1$ for the resulting tip
manifest, so even if all $f$ Byzantine parties also send grade $1$, at most
$1+f<2f+1\le n-f$ grade-$1$ ECHOs exist: neither the all-$n$ fast-seal shortcut nor a
READY quorum of $\quorum$ grade-$1$ statements can form.  A party whose ECHO
quorum mixes grades sends READY-mix immediately, and ECHO and READY quorums
count both grades, so ordinary completion remains available.  If grade $0$
later reaches $\quorum$, each provisional sender refines to READY-$0$ and the
Direct AGB READY path seals $\gcore$.  Neither case treats repaired possession
as proof of authorship.  \Cref{sec:leader-relay} measures the cost: Vantage's
median changes little while its p99 rises to 1.2--1.4\,s as the rejected tip
returns through a later core.

A fully Byzantine schedule can still withhold enough ECHOs that neither grade
reaches $\quorum$ before the refinement window closes.  The residual READY-mix
then leaves the view $\gopen$: completed, core fixed, tip undecided.
Completed-open quarantine stops a correct proposer from reopening the same
lane on every turn --- it drops from ordinary later tip manifests any author
whose carried tip entry never gathered $\quorum$ acknowledgments naming that
exact prefix, then re-probes that author through $T$ after $1,2,4,\ldots$ of
its own proposer opportunities.  \Cref{sec:agb-interface} gives the selection
rule, the fabricated-witness case, and the per-author state cost.  The
discipline is caller-side and no receiver can check it.  Byzantine proposers
ignore it and keep creating residual open views on their own turns.

\paragraph{Reuse boundary.}
Direct AGB applies when one committee handles a value with two parts: a stable
prefix and an optional suffix that can be omitted safely.  Two candidates
elsewhere are a shared mempool and a checkpoint with a speculative log tail.
Each such use has to define its own eligibility rules and prove availability
for them; this paper proves only the author-lane instance.  AGB alone is not
reliable broadcast, not a cross-committee availability certificate, and not
consensus.  Vantage supplies the resolver, deterministic order, repair, and
input fairness around it (\Cref{cor:resolved-agb-log}).

\paragraph{Why the skip shortcut is narrow.}
Relaying a skip vote after $f+1$ votes is unsafe once a party has echoed a
non-skip proposal for that target.  Requiring a free stance removes the
conflict but can leave too few parties to complete the relay.  So the grounded
shortcut --- one vote per party, sent only after that party's own \noready{}
and a first-hand quorum of \echoskip{} --- covers a clean crash-only silent
view and nothing else.  Every other failed view goes through its target-local
resolver.

\paragraph{Provenance boundary.}
The provenance guarantee concerns validator-authored blocks: an output block
under index $i$ was published by $p_i$ to a correct party.  A Byzantine author
may still invent payload bytes or publish forks.  Authenticating external
clients is an application concern, and state-dependent invalid or replayed
transactions must be handled deterministically during execution.  None of that
weakens common block ordering.  None of it is a client-level no-creation
guarantee either.

\paragraph{Per-coordinate grades.}
One grade bit per tip entry would let quorum intersection select an outcome
for each two-valued coordinate, produce $(C,S)$ for some $S\subseteq T$, and
confine simple withholding to one author's lane.  The resolver stays:
Byzantine ECHOs can still make a coordinate mixed, and the optimistic lock and
resolution checks would need a new per-coordinate proof, which this paper does
not have.  The multi-height version is a monotone bitmap; unlike Multimmit's
rank rule, the individual questions remain two-valued at $n=3f{+}1$.

\section{Conclusion}\label{sec:conclusion}
A signature-free data path can carry another author's block one message delay
after receiving it directly, with no certificate for that block and no
availability vote completed beforehand.  Vantage splits each proposal into a
quorum-available core manifest and an optimistic tip manifest.  Completion
fixes the core; the first-hand responses decide whether the tip travels with
it.  When all $n$ parties return a grade-$1$ ECHO, a proposal seals two
message delays after its send, which puts a block published
in time for that proposal into the ordered log three delays after publication.
Only two paths in our survey reach that endpoint --- Autobahn's all-to-all fast
path and FinWhale's $p=1$ path (\Cref{tab:data-path-comparison}) --- and both
buy it with signed votes and a PKI where Vantage uses authenticated channels
and hashes.

Safety holds without timing assumptions: correct outputs never diverge, and
every output block is valid, author-backed, and retrievable.  After GST every
valid block of a correct author is eventually included, and a clean crash-only
silent view is skipped without depending on the resolver.  At $n=100$ on the emulated
ten-region WAN, Vantage holds the lowest p50 latency at every accepted load
among the compared protocol variants and sustains 250k tx/s at about 0.5\,s
p50, with a highest accepted throughput of about 268k tx/s
(\Cref{sec:throughput-sweep}).  In the mixed-open stress
(\Cref{sec:resolver-evaluation}), every injected mixed-open target seals, the resolver
keeps pace with this attack, and the backlog does not grow, although interval
p50 rises to about 4.3\,s while the attack is active.  This is not a general
rate or latency bound: active targets
share finite resources and faulty primaries can consume resolver views
(\Cref{sec:discussion}).

\begin{acks}
Large language models from Anthropic and OpenAI were used under the
author's direction for code development, for adversarial auditing of the
proofs, and for editing and restructuring the manuscript's prose.  The
protocol design, the proofs, and all claims are the author's, and the
author is responsible for all content.
\end{acks}

\bibliographystyle{ACM-Reference-Format}
\bibliography{signature-free}

\clearpage
\appendix

\section{Comparison Details}\label{app:comparison}

This appendix expands the normalized comparison of
\Cref{tab:data-path-comparison}.  The calculations use the same aligned
cross-author setting and do not turn favorable endpoints into general
fault-path bounds.

\paragraph{Mempool-Simple-IT}
A leader names a batch only after collecting $2f+1$ availability votes, and a
receiver checks $f+1$ locally~\cite{simpleit}.  This two-delay gate, two-delay
Opt-RBC, and one commit-vote delay give $5\delta$ under the roughly $5n/6$
optimistic threshold.  It certifies availability before proposal; Vantage
admits a directly published tip after one delay but may leave it open.

\paragraph{Sailfish++.}
Opt-RBC makes each nonleader vertex referenceable after two delays.  Then
$n-f-1$ ordinary correct nonleaders sequence in $5\delta$ and the remaining
$f$ late ones in $7\delta$---two classes, not a range~\cite{sailfishpp,simpleit}.
This supplies per-vertex RBC; Vantage's graded tip instead gives the
$3\delta$ all-to-all or $4\delta$ READY endpoint and may need target-local
resolution.

\paragraph{Autobahn.}
The leader may reference a directly received current lane tip after one delay
while its parent remains PoA-certified~\cite{autobahn}.  Linear fast consensus
adds three delays, giving $4\delta$; the two-phase slow path adds five,
giving $6\delta$~\cite{autobahn}.  On the sketched all-to-all path, Prepare
is the carrier; one delay spreads every vote and another lets each party form
the all-$n$ fast certificate locally.  Thus the derived total is
$\delta+2\delta=3\delta$, without a final certificate broadcast.  The paper
does not state this endpoint; like Vantage's all-$n$ fast-seal shortcut, one
non-voter disables it.
The optimistic tip is outside Autobahn's seamlessness guarantee: a missing
receiver fetches it and withholds its vote, while weak/strong refinement still
needs $f+1$ strong assertions in the PrepareQC.  Vantage lowers the grade
without suppressing ECHO or READY, but an open tip has no constant resolver
bound.

\paragraph{FinWhale, Mysticeti, and compressed variants.}
Cordial Miners begins this best-effort-DAG line~\cite{cordialminers}; Mysticeti
pipelines leaders~\cite{mysticeti}; Starfish supplies the live Push pacemaker
and multisigned Mysticeti-L~\cite{starfish}; Bluestreak compresses with
ordinary signatures~\cite{bluestreak}; and FinWhale adds a two-round fast
path~\cite{finwhale}.  Cross-author admission costs one delay, after which
Mysticeti's three-message decision gives $4\delta$.  FinWhale commits a
round-$r$ leader after $n-p$ round-$(r+1)$ references, where
$n=3f+2p-1$ and $1\le p\le f$; creating and receiving them gives $3\delta$.
At $p=1$, $n=3f{+}1$ and one non-voter is tolerated.  Its adopted Push pacemaker
bounds a round by $2\Delta+2\delta$ after GST.

Bluestreak sends one advance and block per party per round; Vantage sends one
ECHO and initial READY per party per view plus a proposal.  Both favorable
patterns are $2n^2+O(n)$; READY refinement raises Vantage's general bound to
$3n^2+O(n)$.  Bluestreak's happy case is $O(n^2\lambda+nP)$; Vantage's
positional bitmaps give $O(n^2\lambda+n^3)$ aligned, while its proved by-value
and ragged/batch worst case remains $O(n^3\lambda)$.

\paragraph{BBCA-Chain.}
One BEB trip plus the referencing leader's three-trip BBCA gives the
$4\delta$ favorable nonleader path~\cite{bbcachain}.  Timely correct-leader
BBCA tolerates $f$ faults; Vantage's two-delay carrier-to-seal shortcut assumes
all parties correct, and its READY path takes three delays.  BBCA-Chain waits
for the preceding BBCA and recovers through signed complete, adopt, or no-adopt
evidence.  Vantage overlaps views with an at-most-$\delta$ fresh-opportunity
gap but has no general recovery bound.  The $4\delta$ excludes BBCA-Chain's
optional blob-availability layer.

BBCA-Chain's view-advancement value below is likewise derived rather than
stated.  The general bound is $\Delta_\Gamma+\delta$: by
$\Delta_\Gamma$ after synchronized entry, every lagging correct party has
completed or adopted, or has probed and sent a new-view block; one actual
delay then delivers either a complete/adopt certificate or $2f+1$ correct
\textsc{noadopt} blocks everywhere.  The paper requires
$\Delta_\Gamma\ge\Delta_{\mathrm{sync}}+3\Delta$ and proves
$\Delta_{\mathrm{sync}}=\Delta$, so the minimum admissible timer gives the
reported $4\Delta+\delta$.

\paragraph{View advancement.}
\Cref{tab:data-path-comparison} omits per-view advancement, recorded here.
In a view-based protocol, once every correct party has entered view $v$ or
higher, the bound is the additional time until every correct party has
entered a view strictly higher than $v$, following the definition
of~\cite{simpleit}; for a round-based DAG it is the maximum time a correct
party can remain in one round.  One failed view or round therefore delays
progression by at most the bound once it applies, each view or round in a run
of failures incurs its own delay, and the bound says nothing about resolving
the failed value or draining a recovery queue.  The values are
$8\Delta+4\delta$ for Mempool-Simple-IT, $8\Delta+5\delta$ for Sailfish++
after a correct leader, and $10\Delta+\delta$ for Autobahn, all
from~\cite{simpleit}; $4\Delta+\delta$ for Vantage (\Cref{lem:view-duration});
$2\Delta+2\delta$ for FinWhale, the worst-case round duration proved for the
Push pacemaker with leader timeout $2\Delta$~\cite{starfish}, which FinWhale
adopts, while Mysticeti alone states no such bound; and $4\Delta+\delta$ for
BBCA-Chain, derived above.  Vantage's target-local resolver may take longer
than view advancement and has no fixed general bound (\Cref{sec:discussion}).

\FloatBarrier

\section{Direct AGB: Full Specification and Proofs}\label{app:primitive}

\sysname{} runs one pipelined AGB instance per view.  Its round-robin
proposer $\proposer(v)=p_{1+((v-1)\bmod n)}$ broadcasts
$B_v=(C_v,T_v)$.  The core $C_v$ and tip $T_v$ are author-indexed lane-frontier
manifests, encoded in increasing author order with at most one reference
$(k,h)$ per author in each (\Cref{fig:dissemination}); we omit the coordinate
and height when clear.  These fields are self-asserted and do not alone prove
authorship.  A correct proposer places an entry in its exact-$Q$-available core
only after receiving its valid lane prefix directly from the encoded author
and counting $\quorum$ exact-position author-bound acknowledgments.  For a faulty proposer,
``quorum-available'' is not certified; R2 instead checks that every
output-relevant entry is author-backed, which suffices for the retrievability
theorem below.  Hashes are distinct within the proposal and disjoint across
$C_v,T_v$; ordered output removes any cross-view repetition.
The proposal identifier $\hash{B_v}$ abbreviates
$\mathcal H(\texttt{view-proposal}\mathbin\Vert
\mathsf{enc}(\mathsf{sid},v,C_v,T_v))$ and therefore binds the view and the
manifests.

At its abstraction boundary, Direct AGB is parameterized by a deterministic
syntactic predicate $\textsc{Formed}_v(C,T)$ and local predicates
$\textsc{CoreOK}_i(C)$ and $\textsc{TipOK}_i(C,T)$.  The value domain also
supplies a stable item identity and a finite,
deterministic item sequence $\mathsf{Items}(M)$ for each formed manifest.  For every
prior-item set $D$, $\mathsf{Core}_D(C)$ is $\mathsf{Items}(C)$ after removing
identities in $D$ and repeated identities, and $\mathsf{Full}_D(C,T)$ appends
the likewise deduplicated items of $\mathsf{Items}(T)$.  Thus the core expansion
is a prefix of the full expansion, and every full item comes from $C$ or $T$.
Write $\textsc{Witnessed}(M)$ for the value-domain guarantee that every item in
$\mathsf{Items}(M)$ is externally valid and permanently retrievable.  A conforming eligibility
instantiation must prove
\[
  \textsc{CoreOK}_i(C)\Rightarrow\textsc{Witnessed}(C),\qquad
  \textsc{TipOK}_i(C,T)\Rightarrow\textsc{Witnessed}(T)
\]
whenever $p_i$ is correct.  Agreement uses only locality and determinism of
these predicates; availability inherits the displayed witness implications.

For conditional termination, $\textsc{CoreOK}$ and $\textsc{TipOK}$ are
persistent for a fixed proposal: once true at a correct
party, they remain true.  A conforming value domain also provides authorization
and fair repair: once a correct party authorizes a witnessed manifest or its
root, it eventually authorizes and obtains every item in that manifest's
deterministic expansion in every fair execution.

Vantage instantiates $\textsc{CoreOK}_i(C)$ as $\textsc{AuthorOK}_i(C)$ and
$\textsc{TipOK}_i(C,T)$ as
$\textsc{AuthorOK}_i(T)\wedge\textsc{TipAnchored}_i(C,T)$.  Under this
instantiation, $\textsc{Witnessed}(M)$ means that every named valid lane prefix was
published by its encoded author to a correct permanent holder;
\Cref{lem:availability-witness,lem:author-witness} prove the witness
implications.  The author-indexed syntax in \Cref{sec:agb-interface}
instantiates $\textsc{Formed}$, and the lane-chain traversal in
\Cref{sec:protocol} instantiates the expansion contract.  The standalone
safety arguments use only this contract and the initial-response and
same-proposal refinement rules.

\subsection{Data Dissemination}\label{sec:agb-diss}

Payloads are disseminated in data blocks below AGB.  An \emph{author lane} is
the author-indexed publication structure formed by one party's data blocks.  A
height-$k$ block of party $p_i$ has the canonical form
$b_i^k=(\mathsf{sid},i,k,h_i^{k-1},x)$, where $|x|\le\ell_{\max}$ and
$h_i^{k-1}$ is the predecessor hash (the session fixes the shared height-$0$
genesis).  Its identifier is
$h_i^k=\mathcal H(\texttt{data-block}\mathbin\Vert
\mathsf{enc}(b_i^k))$, abbreviated $\hash{b_i^k}$.  A \emph{valid lane prefix}
is a linear hash chain whose blocks share one author index, have consecutive
heights and matching predecessor hashes back to genesis, and satisfy
$\textsc{BlockOK}$.  A correct author's valid lane prefixes are mutually
comparable; a Byzantine author may publish forked lane prefixes.  A correct
party creates only valid extensions of its current lane frontier.  To publish
a new block it sends $\langle\rstat{publish},b_i^k\rangle$ to all parties.  A
recipient treats this as author publication only when the authenticated
channel sender and the block's encoded author are both $p_i$.

A hash $h=\hash{b_i^k}$ names both the block and, through the predecessor-hash
chain, its valid lane prefix $b_i^1,\ldots,b_i^k$.  A reference $(i,k,h)$ in a
well-formed fixed proposal, a local completion or seal, or a well-formed
resolver value under local validation authorizes repair of its root block,
as does the anchor of a counted availability claim (below).
After obtaining a block whose hash, session, encoded author, height, and
$\textsc{BlockOK}$ check match an authorized reference, the requester also
authorizes its predecessor reference; this process continues recursively to
genesis.  For each authorized missing hash, the party requests the block from
all parties.  A holder answers with a distinctly typed
$\langle\rstat{serve},h,b\rangle$ repair message.  A canonical, correct-session,
size-bounded, $\textsc{BlockOK}$ body whose hash is $h$ is cached even if its
encoded author or height does not match the currently authorized coordinate;
only an exact coordinate match advances the predecessor walk.  This prevents a
false coordinate from consuming the sole useful response for a later genuine
reference.  A reply that fails the hash or body checks is ignored and never
marks the hash as obtained.  No repair message establishes publication
provenance, even when its sender equals the encoded author; if the corresponding
\rstat{publish} message arrives later, it may upgrade the cached bytes.

A correct party records the first pending request for each requester--hash pair
even when it does not yet hold the block, and answers once if it later retains
that block.  A correct requester sends at most one request for a given
authorized hash to each peer, and a correct holder answers at most once for
each requester--hash pair.  Reliable delivery and retained pending requests
make retransmission unnecessary; the rule also bounds correct-party traffic
induced by repeated Byzantine requests.

For $h=h_a^k$, let $\textsc{DirectPub}_j(a,k,h)$ mean that $p_j$ holds the
valid lane prefix through $h$ and received every block in it through
\rstat{publish} messages directly from $p_a$.  The party retains this prefix
and advances its local direct frontier for $a$, but this event sends no
periodic protocol message.

For a formed proposal $B$, let
$R(B)=(r_1,\ldots,r_m)$ be the canonical ordered concatenation of its core and
tip lane references.  A proposal ECHO from $p_i$ carries
an authenticated availability claim $A_i(B)$ outside the immutable ECHO
identity.  For $r_j=(a,k,h)$, its exact-position bit may be set only when
$\textsc{DirectPub}_i(a,k,h)$ holds.  Otherwise it may contain one integer
$(j,d)$ with $0<d<k$ only when the predecessor walk from $h$ identifies the
unique ancestor $r'=(a,k-d,h')$ and $\textsc{DirectPub}_i(r')$ holds.  Thus the
integer means ``I directly received this exact ancestor of the named tip,''
not merely ``I hold some block at height $k-d$.''  A correct sender is silent
at that position if it cannot verify the ancestry or its direct provenance.

After verifying the canonical proposal body, a receiver treats an exact bit
immediately as the logical first-hand acknowledgment
$\langle\rstat{ack},a,k,h\rangle_i$.  Since \textsc{DirectPub} is a prefix
predicate, once the predecessor chain is verified the same claim also
acknowledges every derived ancestor in that prefix.  For $(j,d)$ the receiver
first obtains and verifies the named tip's predecessor chain, derives the exact
$h'$, and only then counts the acknowledgments for $(a,k-d,h')$ and its derived
ancestors; until that walk succeeds, the claim remains pending.  Repeated
ECHOs or repeated coordinates never count the same sender twice for one exact
tuple.  A receiver processes the availability field only on that sender's
first accepted echo-stage response for the view; later envelope variants do
not create additional claims.  A Byzantine ECHO author may
fabricate its own fields, but contributes at most one counted sender to any
exact tuple.  Hence a
correct acknowledgment still asserts validity, possession, and direct receipt
from the encoded author, and every counting party attributes it only to the
authenticated ECHO sender.

\begin{definition}[Availability]\label{def:available}
A lane reference $(a,k,h)$ is \emph{$q$-available} at a party if it
has received $\langle\rstat{ack},a,k,h\rangle$ first-hand from at least $q$
distinct parties.
\end{definition}

Call a reference \emph{exact-$\quorum$-available} at a party when its
$\quorum$ counted acknowledgments all came from exact-position bits for that
reference.  Unlike an integer claim, an exact bit needs no predecessor walk at
another receiver.  Correct Vantage proposers use this stronger, locally
checkable condition for core promotion; ordinary availability and
$\textsc{AuthorOK}$ count either encoding.

A valid lane prefix obtained only through repair is held and served but never
acknowledged.  Hash links prove content, not authorship: anyone can construct a
well-linked block carrying another party's index.  Authenticated
\rstat{publish} receipt establishes authorship locally, while $f{+}1$
author-bound acknowledgments establish that some correct recipient observed
such a publication.  Consequently, a block that $p_a$ never publishes to a
correct party gathers at most $f$ acknowledgments for $(a,k,h)$.

A reference is \emph{locally available} at a party when its valid lane prefix is
held by that party or it is $(f{+}1)$-available there; a manifest is locally
available when every entry in it is.  Define the separate provenance gate
\[
\begin{gathered}
  \textsc{AuthorOK}_j(a,k,h)\\[-2pt]
  \Longleftrightarrow\ \textsc{DirectPub}_j(a,k,h)\\[-2pt]
  {}\lor\ \bigl((a,k,h)\text{ is $(f{+}1)$-available at }p_j\bigr).
\end{gathered}
\]
For a manifest, $\textsc{AuthorOK}_j$ means that the predicate holds for every
entry at its encoded author coordinate.  $\textsc{AuthorOK}$ implies local availability;
the converse is false for a fetched lane prefix.  \emph{Retention rule:} a correct party
retains every valid lane prefix it acknowledges; every valid lane prefix whose local
holding it relies on when a protocol rule checks local availability,
$\textsc{AuthorOK}$, or tip anchoring; and every valid lane prefix it obtains by request for
such a check --- even when the prompting window has closed by the time the lane
prefix arrives.  It also retains an anchoring chain whose verified ancestry it
uses to emit or resolve an integer ACK claim.  Hash-correct bodies received during an authorized repair are
also cached so that a later exact coordinate can consume them.  Retained lane
prefixes are held and served indefinitely; local
sealing or output of the naming views is no discard signal, since another
correct party's seal and output may lag arbitrarily.  Bounded-state reclamation
is outside the model.
We call a reference \emph{retrievable} when some correct party permanently
retains its valid lane prefix.  In a fair execution, the repair rule below makes
every correct party that authorizes the exact reference eventually obtain it.

\Cref{alg:data-dissemination} gives the corresponding event handlers; the
prose above is normative.
\begin{algorithm*}[tbp]
\SetAlgoSkip{}
\caption{Authenticated data publication and repair at party $p_i$}
\label{alg:data-dissemination}
\footnotesize
\Comment{A stored block may carry a direct-publication mark, a repair mark,
  or both; only the first can establish provenance or an ECHO ACK claim.}
\Glob{}{
  maps of cached blocks and direct-publication marks;
  sets $\mathit{authorized}$, $\mathit{requested}$, $\mathit{pendingReq}$,
  $\mathit{retained}$, and $\mathit{answered}$;
  direct lane frontiers; per-tuple ECHO claim counts and pending exact and
  integer ACK claims
}

\Upon{the local application creates the next valid block $b_i^k$}{
  store $b_i^k$; process it as a direct publication from $p_i$;
  \textbf{broadcast} $\langle\rstat{publish},b_i^k\rangle$
}

\Upon{a direct $\langle\rstat{publish},b\rangle$ from $p_a$ whose encoded
  author is $a$}{
  \If{$b=(\mathsf{sid},a,k,h',x)$ is canonical, $|x|\le\ell_{\max}$,
    and $\textsc{BlockOK}(b)$}{
    $h\gets\hash{b}$; store $b$ and mark $(a,k,h)$ directly published by $p_a$
      \Comment*[r]{may upgrade repaired bytes}
  }
}

\Fn{\textsc{DirectPub}$_i(a,k,h)$}{
  \KwRet{$p_i$ holds the valid lane prefix through $h$ and every block in it
    carries a direct-publication mark from $p_a$}
}

\Upon{$\textsc{DirectPub}_i(a,k,h)$ first becomes true}{
  retain the valid lane prefix through $h$;
  advance the direct frontier of $a$ to $(a,k,h)$ if it extends that frontier
}

\Fn{\textsc{AckClaims}$_i(B)$}{
  let $R(B)=(r_1,\ldots,r_m)$ be the canonical lane-reference vector;
  initialize an $m$-bit exact vector and an empty sparse integer map;
  \ForEach{$r_j=(a,k,h)$}{
    \eIf{$\textsc{DirectPub}_i(a,k,h)$}{set exact bit $j$}{
      \If{the predecessor walk from $h$ verifies a greatest directly
        published ancestor $(a,k-d,h')$ with $0<d<k$}{
        add $(j,d)$ to the sparse map
      }
    }
  }
  \KwRet{the exact vector and sparse map}
}

\Upon{the first accepted authenticated proposal ECHO from $p_s$ with claim
  $A_s(B)$ and a formed canonical body $B$}{
  \ForEach{exact bit $j$ of $A_s(B)$}{
    count $p_s$ for $r_j$; record the pending claim $(s,j)$ with claimed
    height $k$; \textsc{Authorize}$(r_j)$
  }
  \ForEach{sparse entry $(j,d)$ of $A_s(B)$}{
    record the pending claim $(s,j,d)$ with claimed height $k-d$;
    \textsc{Authorize}$(r_j)$
  }
}

\Upon{a pending claim by $p_s$ on anchor $r_j=(a,k,h)$ with claimed height
  $k^\ast$, and a verified ancestor $(a,k',h')$ of $r_j$ with $k'\le k^\ast$}{
  count $p_s$ for $(a,k',h')$ unless already counted for that tuple
    \Comment*[r]{state rule; also fires when the chain is already cached}
}

\Proc{\textsc{Authorize}$(a,k,h)$}{
  add $(a,k,h)$ to $\mathit{authorized}$ if absent\;
  \eIf{a cached $b=(\mathsf{sid},a,k,h',x)$ satisfies $\hash{b}=h$,
    $\textsc{BlockOK}(b)$, the size bound, and either $k>1$ or $h'$ is the
    session genesis hash}{
    \If{$k>1$}{\textsc{Authorize}$(a,k-1,h')$}
  }{
    \ForEach{$p_j\in\Pi$ with $(j,h)\notin\mathit{requested}$}{
      add $(j,h)$ to $\mathit{requested}$;
      send $\langle\rstat{request},h\rangle$ to $p_j$
    }
  }
}

\Upon{$\langle\rstat{serve},h,b\rangle$ arrives for a requested hash $h$}{
  \If{$b=(\mathsf{sid},a,k,h',x)$ is canonical, $\hash{b}=h$,
    $|x|\le\ell_{\max}$, and $\textsc{BlockOK}(b)$}{
    cache $b$ as repaired data, without a direct-publication mark
      \Comment*[r]{coordinate-independent cache}
  }
  \Comment*[r]{a failed hash/body check changes no state}
}

\Upon{a cached $b=(\mathsf{sid},a,k,h',x)$ matches an authorized
  $(a,k,\hash{b})$ and either $k>1$ or $h'$ is the session genesis hash}{
  \textsc{Authorize}$(a,k,\hash{b})$
  \Comment*[r]{fires after either PUBLISH or SERVE; continues the walk}
}

\Upon{an authorized predecessor walk verifies a valid lane prefix through genesis}{
  retain every block in that lane prefix; invoke \textsc{TryServe} on their hashes
}

\Upon{a direct $\langle\rstat{request},h\rangle$ from $p_j$}{
  add $(j,h)$ to $\mathit{pendingReq}$ if $(j,h)\notin\mathit{answered}$;
  \textsc{TryServe}$(h)$
}

\Upon{a cached block with hash $h$ becomes retained}{
  \textsc{TryServe}$(h)$
}

\Proc{\textsc{TryServe}$(h)$}{
  \If{a retained block $b$ satisfies $\hash{b}=h$}{
    \ForEach{$(j,h)\in\mathit{pendingReq}\setminus\mathit{answered}$}{
      add $(j,h)$ to $\mathit{answered}$; remove it from $\mathit{pendingReq}$;
      send $\langle\rstat{serve},h,b\rangle$ to $p_j$
    }
  }
}

\Fn{\textsc{AuthorOK}$_i(a,k,h)$}{
  \KwRet{$\textsc{DirectPub}_i(a,k,h)$ or at least $f+1$ counted
    first-hand ECHO claims for $(a,k,h)$}
}
\end{algorithm*}

\begin{lemma}[Availability witnesses]\label{lem:availability-witness}
For every reference $(a,k,h)$:
\begin{enumerate}[label=\textup{(\roman*)},leftmargin=2.2em,topsep=2pt,itemsep=2pt]
  \item if a correct party relies on it in a local-availability,
        $\textsc{AuthorOK}$, or tip-anchoring check, then the reference is
        retrievable;
  \item in every fair execution, every correct party that authorizes a
        retrievable reference eventually obtains its valid lane prefix;
  \item if it is $(\quorum)$-available at one correct party, then it eventually
        becomes $(f{+}1)$-available at every correct party.
  \item if it is exact-$\quorum$-available at a correct party by time
        $\beta\ge\mathrm{GST}$, then it is $(f{+}1)$-available at every correct
        party by $\beta+\delta$.
\end{enumerate}
\end{lemma}
\begin{proof}
For (i), if the checking party holds the lane prefix, the retention rule applies;
otherwise it counted $f{+}1$ acknowledgments, including one from a correct
party that holds and retains it.  For (ii), requests go to every party and a
correct holder answers.  Starting at $h$, each verified block authorizes its
predecessor, so induction on height retrieves the lane prefix through genesis;
collision resistance makes each requested hash unambiguous.  For (iii), at
least $n-2f\ge f{+}1$ counted acknowledgments are from correct parties, and their
claim-bearing ECHOs eventually reach every correct party.  A counted claim
authorizes its anchor's predecessor walk at every receiver; the correct claim
author retains that chain, so fair repair eventually derives the same exact
tuples everywhere.  For (iv), at least $n-2f\ge f{+}1$ of the exact
claims counted by time $\beta$ came from correct ECHO authors.  Their formed
claim-bearing ECHOs were already broadcast and arrive everywhere by
$\beta+\delta$ under the in-flight convention; exact positions require no
anchor repair.
\end{proof}

\begin{lemma}[Author witness]\label{lem:author-witness}
If a correct party satisfies $\textsc{AuthorOK}_j(a,k,h)$, then some correct
party permanently holds the valid lane prefix ending at $h$ and received every
block in that lane prefix through \rstat{publish} messages directly from $p_a$.
In particular, every block in the lane prefix satisfies $\textsc{BlockOK}$ and
$p_a$ published every such block to a correct party.
\end{lemma}
\begin{proof}
The claim is immediate in the $\textsc{DirectPub}$ case.  Otherwise, among the
$f{+}1$ acknowledgers counted by $p_j$ at least one is correct.  A correct
party contributes that logical author-bound acknowledgment in its ECHO only
after its own $\textsc{DirectPub}$ predicate holds, and the retention rule is
permanent.
\end{proof}

A count of $\quorum$ acknowledgments is therefore local evidence that at
least $n-2f\ge f{+}1$ correct parties hold the directly published valid lane prefix and,
by \Cref{lem:availability-witness}(iii), that both the $(f{+}1)$-availability and $\textsc{AuthorOK}$
conditions used for a core entry below eventually hold at every correct party.

\subsection{Interface}\label{sec:agb-interface}

The direct interface has one sender input, two idempotent environment events,
one completion event that makes the core irrevocable, and an optional
direct-result event:
\begin{itemize}[leftmargin=1.4em,topsep=2pt,itemsep=2pt]
  \item $\rstat{propose}(v,C,T)$ --- invoked once by $\proposer(v)$; a
        correct caller supplies a value satisfying $\textsc{Formed}_v(C,T)$.
        In Vantage, a correct proposer places
        in $C$ only
        references $(a,k,h)$ for which
        $\textsc{DirectPub}_i(a,k,h)$ holds and the reference is
        exact-$\quorum$-available there, and in $T$ only references for which
        $\textsc{DirectPub}_i(a,k,h)$ holds.  Concretely, it lists in $C$, for
        each author, the newest such exact-$\quorum$-available reference.  For
        a non-quarantined author, $T$ contains the least-height generally
        $\quorum$-available but not exact-$\quorum$-available directly
        published prefix strictly above $C$, when one exists, and otherwise
        the newest directly published lane tip whose valid lane prefix
        strictly contains the author's $C$ entry --- an author with no $C$
        entry is anchored at genesis, so any directly published tip qualifies --- at most
        one $C$ and one $T$ entry per
        author.  Inclusion is thus deterministic in the proposer's first-hand
        state.  Here \emph{newest} means greatest verified lane height, with
        the lexicographically smallest hash breaking an equal-height tie.  If
        an author forked its lane, the $C$ entry pins the kept branch below
        itself (the same height/hash order resolves locally visible forks; a
        forking author is Byzantine), and a
        directly published lane tip on any other branch is never included.  Authors
        are processed in a canonical order, skipping any hash already listed
        under an earlier index, so all entries are globally distinct and no
        hash appears in both manifests.
        A correct Vantage proposer also keeps one local quarantine witness and
        two scalar retry counters per author.  Ordinary inclusion and an
        initial READY-mix do not fill this slot.  When a READY quorum first
        completes a proposal while neither READY grade has quorum, the party
        stores each tip reference that is not exact-$\quorum$-available in its
        author's empty slot.  It omits $a$ from ordinary later $T$ manifests
        until the witness becomes exact-$\quorum$-available or a containing
        prefix becomes its current $C$ candidate, then clears the slot.

        Number this party's own proposer opportunities.  A new witness is due
        at the next opportunity; after an actual retry, its successive gaps are
        $1,2,4,\ldots$ opportunities.  In one proposal the party selects at
        most one due quarantined author with a directly published prefix above
        the current core, choosing least due opportunity and then author order.
        If the author has a generally $\quorum$-available but not exact-$\quorum$-available
        prefix above the core, it puts the least-height such prefix in $T$ as a
        stable confirmation target; otherwise it puts the newest directly
        published tip there.  The same confirmation preference applies to
        ordinary, non-quarantined $T$ selection.  A due author
        without such a directly published tip is skipped without obstructing
        another eligible author or advancing its retry counter.  Re-completion does not replace a
        nonempty slot or reset its backoff.  This caller-side discipline is not
        part of $\textsc{Formed}$ and cannot constrain a Byzantine proposer.
        It never activates on a homogeneous completing quorum.  Proposal ECHOs
        eventually clear any slot whose witness a correct author actually
        published.  Even a fabricated witness cannot censor a correct author's
        lane: later due probes use that author's real directly published
        prefixes, and a general claim quorum freezes the least unconfirmed
        coordinate until exact-position ECHOs confirm it, without resetting
        the witness or its backoff.
        Over-inclusion elsewhere is harmless since sequencing outputs every
        data block once (\Cref{sec:protocol}).  Each manifest thus has at
        most $n$ entries.  Vantage instantiates $\textsc{Formed}$ by
        requiring that each of $C$ and $T$ list at most one entry
        per author, all hashes across $C$ and $T$ are distinct, and every
        reference has a canonical positive height and hash.
        Correct parties never echo a malformed proposal
        (R2), so the per-author entry bound --- and with it the $n$-entry
        bound --- holds for every echoed proposal,
        Byzantine ones included.  These are syntactic checks: when bytes or
        acknowledgments are available, $\textsc{AuthorOK}_i(a,k,h)$ can hold
        only for a valid lane prefix whose encoded author and height are exactly
        the reference coordinate $(a,k)$.
  \item $\rstat{activate}(v)$ --- permits the fixed proposal's positive echo
        gate to fire.  In Vantage this local event occurs when authenticated
        receipt advances the contiguous responsive proposal frontier through
        $v$; formal entry activates separately through $\rstat{enter}(v)$.
  \item $\rstat{enter}(v)$ --- records the local entry time, activates the fixed
        proposal if present, and arms the echo and ready fallback deadlines.
        Vantage generates this local event with the WISH pacemaker of
        \Cref{sec:pacemaker}.  Neither environment event is a counted statement
        or transferable evidence.
  \item $\rstat{complete}(v)\to B_v$ --- emitted when the first
        proposal-ready quorum is counted; it fixes $B_v=(C,T)$ and makes
        $C$ an irrevocable prefix of the view's eventual non-skip
        output.  A completed but unsealed view is $\gopen(C,T)$.
  \item $\rstat{direct-seal}(v)\to X$ --- the optional direct output, where $X$
        is $\gfull(C,T)$ or $\gcore(C)$.  It fires on the first homogeneous
        ready quorum of the corresponding grade.  The mixed grade is
        nonterminal and $\gskip$ is not a Direct AGB output.
\end{itemize}

Direct AGB never reads resolver state.  Vantage connects its
$\rstat{direct-seal}$ output to a caller-owned $\rstat{try-seal}(v,X)$ arbiter.
Every direct result, Vantage fast or skip result, and resolver outcome is
submitted to it.  While empty, the
arbiter accepts the first submission, stores $X$, and emits the terminal
$\rstat{seal}(v)\to X$ event; afterward it retains that outcome and ignores
later submissions.  Compatibility ensures that two submissions, if both occur,
name the same outcome.  This output multiplexer belongs to the composition, not
to Direct AGB.  The lifecycle vocabulary of \Cref{sec:overview} applies
throughout: the arbiter is what turns a direct, fast, skip, or resolver result
into the terminal \emph{seal} event.

\subsection{Vantage Pacemaker and Activation}\label{sec:pacemaker}

The following service is Vantage's concrete implementation of Direct AGB's
$\rstat{enter}$ and $\rstat{activate}$ environment events.  It is not used by
the asynchronous safety proof of \Cref{thm:agb-safety}.

The pacemaker uses first-hand $\langle\wish,x\rangle$ messages, following the
bounded-space high-watermark construction of Bravo, Chockler, and
Gotsman~\cite{bravo2022liveness}.  Party $p_i$ stores, for every author $p_j$,
the largest directly received wish $\omega_i[j]$.  Let $\omega_i^+$ and
$\omega_i^Q$ be the $(f+1)$st and $(\quorum)$th largest entries, respectively.
When $\omega_i^+$ exceeds $p_i$'s own largest wish, it broadcasts a wish for
$\omega_i^+$; independently, whenever $\omega_i^Q$ increases it raises its
entry target to $\omega_i^Q$.  Wishes are
high-watermarks, and the two updates are independent: a wish for $x$ supports
every view at most $x$, and target advancement never waits for
$\omega_i^Q=\omega_i^+$.  On a target $x$, the party records entry to every
missing view
through $x$ immediately and in order.

Initially $\omega_i[j]=0$ for every $j$, the largest entered view is $0$, and
the own-wish high-watermark is $0$.  Genesis enters view~$1$, after which
$p_i$ sets $\omega_i[i]=2$, records own wish~$2$, and broadcasts
$\langle\wish,2\rangle$ (self-delivery is immediate).  Upon receiving
$\langle\wish,x\rangle$ first-hand from $p_j$, including as a piggybacked
high-watermark, $p_i$ sets
$\omega_i[j]\gets\max\{\omega_i[j],x\}$ and recomputes both order statistics.
It first performs any enabled $f{+}1$ amplification, updating
$\omega_i[i]$ and its own-wish high-watermark when it broadcasts and
recomputing both order statistics after that update.  It then records entry
through any newly supported quorum target.  Stale wishes cause no transition.

Write $E_i(u)$ for $p_i$'s unique echo-stage response and $R_i(u)$ for its
initial ready-stage response in view $u$.  The local trigger for wishing follows Starfish's
two-response progress condition~\cite{starfishs}: before emitting whichever of
$E_i(v-1)$ and $R_i(v-2)$ completes the pair, $p_i$ raises its own wish to at
least $v+1$.  Every response piggybacks the sender's current own-wish
high-watermark outside its counted response identity.  An optional READY-mix
refinement carries the same field but does not run this progress trigger a
second time.
Amplification may send a wish standalone.  A wish only schedules views; it is
never an echo, a ready, an availability acknowledgment, an origin bit, or a
resolution justification.

Let $e_i(v)$ be the time $p_i$ enters view $v$.  Entry is strictly increasing
locally.  Recording entry to $v$ arms both fallback schedulers for $v$;
positive responses may already have fired through the responsive path, and
pending responses may finish after later entries.  R2 and R3 close them by
$e_i(v)+\theta_E$ and $e_i(v)+\theta_R$, respectively, in post-GST
drift-free time, unless their positive
gates fire earlier.  Thus entry never waits for an older response window.

In parallel, $p_i$ maintains a \emph{responsive proposal frontier} $a_i$,
initially $0$.  Formal entry to view $v$ raises $a_i$ to at least $v-1$.
For each view, $p_i$ fixes the first proposal received directly from its
designated proposer and ignores later versions.  When the fixed proposal for
$a_i+1$ is well formed, it advances $a_i$ by one; buffered consecutive fixed
proposals are processed in increasing view order.  A proposal that advances
the frontier may activate its positive response path before formal entry, but
it never causes formal entry and is not evidence for any quorum.  Thus WISH
remains the only recovery and synchronization mechanism, while a contiguous
stream of first-hand proposals supplies a one-hop responsive path.  As in the
rest of the unbounded-execution model, parties retain future per-view messages;
the non-Zeno convention makes the retained set finite at every finite time.

The designated proposer also uses WISH evidence directly: R1 invokes
$\rstat{propose}(v,\cdot)$ already when $\omega_i^+$ first supports $v$ at
$p_i=\proposer(v)$, one amplification hop before the entry quorum threshold
$\omega_i^Q$.  The early proposal is passive --- receivers fix it and
authorize repair of its references, but no response fires before
activation --- and it cannot outrun genuine progress.  For $v\ge3$, by
induction over raise events, the earliest correct wish supporting $v$ is an
own raise of the two-response trigger, by whose local step the sender has
emitted its echo-stage response for a view at least $v-2$; an
amplification raise to $\omega^+\ge v$ counts a wish of at least one
correct author with a strictly earlier send and inherits the same
provenance.  Genesis wishes support only views up to~$2$.  Since
$\omega_i^+\ge v$
counts $f{+}1$ authors, at least one correct, the trigger certifies that
progress.  Its payoff appears in the fallback analysis: the wish quorum
behind the first correct entry to $v$ contains $f{+}1$ correct wishes
supporting $v$, which reach $\proposer(v)$ within one delay of that entry,
so the proposer invokes within $\delta$ of the first entry rather than
within the $2\delta$ entry spread (\Cref{thm:agb-live}).  A smaller entry
spread would require a smaller entry threshold: the $f{+}1$ correct wishes
guaranteed behind the first entry reach every correct party within one delay, but
completing a $\quorum$ quorum from them costs the amplification hop, since
the up to $f$ Byzantine wishes inside the first quorum reach only their chosen
recipient.  An entry threshold that may contain fewer than $f{+}1$ correct
wishes would instead let the adversary drive a single correct party's
deadlines arbitrarily far ahead, spending its one-shot refusals.  Signed and certified pacemakers escape
this dilemma by relaying transferable evidence; the proposer-side trigger
recovers the lost delay exactly where the liveness analysis needs it,
because proposing early is harmless where entering early is not.

\begin{lemma}[Pacemaker progress and synchronization]
\label{lem:pacemaker}
\leavevmode
\begin{enumerate}[label=\textup{(\alph*)},leftmargin=2.0em,topsep=2pt,itemsep=3pt]
  \item Every correct party eventually enters every view and emits one
        echo-stage and one initial ready-stage response for every view; a
        provisional READY-mix may additionally refine once.
  \item There is a finite view $V$ such that, for every $v\ge V$, if the first
        correct entry to $v$ occurs at time $s_v$, then every correct party
        enters $v$ by $s_v+2\delta$.
  \item If every party is correct, the same cutoff can be chosen so that every
        correct party enters each $v\ge V$ by $s_v+\delta$.
\end{enumerate}
\end{lemma}
\begin{proof}
For (a), all correct parties enter view~$1$ from genesis; responses for
nonpositive views are fixed, and every correct party initially broadcasts
$\langle\wish,2\rangle$.  Suppose all correct parties eventually enter $v$.
The deadlines for views $v-1$ and $v-2$ force every correct party to emit both
responses, thereby broadcasting a wish of at least $v+1$.  Every correct party
therefore receives a wish quorum for $v+1$ and enters every missing view
through $v+1$.  Induction proves entry to all views; the stage deadlines then
prove emission of every response.

For (b), assume the standard non-Zeno execution convention and choose $V$
larger than every view entered, proposed, responded to, or wished for by a
correct party before GST; this is a finite prefix.  Fix $v\ge V$ and its first
correct entry time $s=s_v\ge\mathrm{GST}$.
Let $x\ge v$ be the target whose traversal caused that entry.  The entering
party counted a wish quorum supporting $x$, including at least $f{+}1$ correct
authors.  Their wishes reach every correct party by $s+\delta$, causing every
correct party to broadcast a wish supporting $x$ if it has not already done
so.  These wishes arrive everywhere by $s+2\delta$, making every $\omega_i^Q$ support
$x$.  Sequential traversal therefore enters $v$ by that time, even if the
target is higher.

For (c), every author in the wish quorum counted at time $s_v$ is correct.
Those same wishes reach every correct party by $s_v+\delta$, directly giving
each of them a wish quorum for the same target.  No amplification hop is
needed.
\end{proof}

The synchronizer also bounds how long a view can hold every correct party.

\begin{lemma}[Eventual worst-case view duration]\label{lem:view-duration}
Let $v$ be a view whose $t_1$ --- the earliest time at which every correct
party has entered view $v$ or higher --- satisfies $t_1\ge\mathrm{GST}$.
Then every correct party enters a view greater than $v$ by
$t_1+\theta_R+\delta=t_1+4\Delta+\delta$.  Vantage's eventual worst-case
view duration, in the sense of~\cite{simpleit}, is therefore
$4\Delta+\delta$.
\end{lemma}
\begin{proof}
Fix a correct party $p_i$.  By $t_1$ it has recorded entry, in order, to
every view through $v$, so $e_i(v-2)\le e_i(v-1)\le e_i(v)\le t_1$; the
entries themselves may precede GST.  The fallback rules close its
echo-stage response for $v-1$ by $\max\{e_i(v-1),\mathrm{GST}\}+\theta_E$
and its initial ready-stage response for $v-2$ by
$\max\{e_i(v-2),\mathrm{GST}\}+\theta_R$: a timer still armed at GST has at
most its full duration remaining once clocks are drift-free.  Hence the
progress pair for $v-1$ completes by $t_1+\theta_R$; for $v=2$ the
  initial ready-stage member of the pair is a fixed genesis response, so the pair
completes with the echo-stage emission alone, and for $v=1$ both members
are fixed and the genesis wish already supports~$2$.  Before emitting the
completing response, $p_i$ raises its
own wish to at least $v+1$ unless its own wish already supports $v+1$; in
either case $p_i$ has broadcast a wish supporting $v+1$ by $t_1+\theta_R$.
The correct parties' wishes supporting $v+1$ --- at least $n{-}f$ of
them --- reach every correct party by
$t_1+\theta_R+\delta$ under the in-flight convention, make $\omega^Q$
support $v+1$ there without an
amplification hop, and sequential traversal enters every view through
$v+1$.
\end{proof}

\subsection{Protocol}\label{sec:agb-rules}

All statements below are broadcast to all parties over authenticated channels.
A party sends at most one echo-stage statement and one \emph{initial}
ready-stage statement for a view.  Echo-stage statements are either proposal
echoes or \echoskip{} and are immutable.  The initial ready-stage statement is
proposal-ready or \noready{}; \noready{} and homogeneous proposal-ready are
final, whereas an initial READY-mix may be followed by at most one
same-proposal homogeneous refinement.  A party counts a statement only after
receiving it directly from its author.  It counts the first ready-stage
statement from an author once toward the stage-response and proposal-ready
censuses.  If that first statement is READY-mix for $B$, it additionally
accepts one later READY-$g$ for the same $B$, $g\in\{0,1\}$, into the
grade-$g$ tally without counting the author again; every other later or
conflicting statement is ignored.  If the homogeneous statement is the first
countable version from that author --- whether due to network reordering or
because digest-named statements awaited body validation --- it is the final
initial statement and a later READY-mix is ignored.  Where READY-mix was counted first, it remains historical
non-grade-$1$ evidence for resolver candidate selection even after refinement.
The sender's own responses are delivered and counted immediately under the
self-delivery convention of \Cref{sec:model}.

In the Vantage wrapper, responses for views $v'\le0$ are fixed genesis responses
treated as received from every party.  Genesis enters every correct party into
view~$1$, raises its responsive proposal frontier to $0$, and makes it broadcast
$\langle\wish,2\rangle$.  Entry, frontier state, and WISH messages are not
evidence for a broadcast instance; safety holds for arbitrary entry and
timeout timings.

\begin{enumerate}[label=\textbf{R\arabic*.},leftmargin=2.8em,topsep=3pt,itemsep=3pt]
  \item \textbf{Propose.}  On its first sender input
        $\rstat{propose}(v,C,T)$, party $p_i=\proposer(v)$ sends
        $\langle\rstat{propose},v,C,T\rangle$ once.  Vantage invokes this input
        when its responsive frontier first satisfies $a_i\ge v-1$ or its wish
        statistic first satisfies $\omega_i^+\ge v$ (\Cref{sec:pacemaker}),
        and constructs
        $(C,T)$ as in \Cref{sec:agb-interface}.  In particular, formal entry
        to $v$ raises that frontier to $v-1$, so a correct Vantage proposer
        invokes the input no later than its own entry; receipt of the preceding
        direct proposal or early wish evidence may make it invoke earlier.
  \item \textbf{Echo stage.}  Proposal messages may arrive before or after a
        party enters the proposal's view.  The party fixes the first view-$v$
        proposal received directly from $\proposer(v)$ and ignores every later
        version.  It checks $\textsc{Formed}_v(C,T)$; a malformed fixed
        proposal is not echoable.  Otherwise the party stores $B_v=(C,T)$ and
        invokes the value domain's authorization and repair hooks for $C$ and
        $T$.  In Vantage, a malformed proposal also cannot advance the
        responsive frontier, and authorization requests from all parties every
        missing block of each named lane prefix.

        A stored proposal becomes \emph{active} after the local environment
        invokes $\rstat{activate}(v)$ or $\rstat{enter}(v)$.  In Vantage,
        processing the fixed proposal invokes $\rstat{activate}(v)$ exactly when
        it advances the responsive frontier to $v$; a buffered proposal can
        therefore activate when the missing contiguous prefix arrives.  From
        activation onward,
        while the echo-stage response is pending, the party continuously
        evaluates the positive gate:
        $\textsc{CoreOK}_i(C)$ and $\textsc{TipOK}_i(C,T)$ both hold.  In
        Vantage, the tip predicate means
        $\textsc{AuthorOK}_i$ holds for every entry of $T$ and the
        \emph{tip-anchoring check} holds: for every author index present in
        both manifests, the party holds the tip entry's full lane prefix and
        that lane prefix strictly contains the author's $C$ entry.  This is a pure local hash
        walk, so counted acknowledgments alone never substitute for a paired
        tip.  When the gate first holds, the party broadcasts
        $\langle\rstat{echo},v,C,T,1\rangle$, even if formal entry has not yet
        occurred.

        Formal entry to $v$ additionally arms the fallback scheduler.  If the
        response is still pending and a well-formed proposal is fixed, let
        $t$ be the later of entry and its direct
        receipt.  The fallback window is
        $[t,\,\min\{t+\Delta,e_i(v)+\theta_E\}]$.  If no grade-$1$ echo has
        fired by its end, the party broadcasts
        $\langle\rstat{echo},v,C,T,0\rangle$ if $\textsc{CoreOK}_i(C)$ holds, and
        $\langle\echoskip,v\rangle$ otherwise.  If no echoable proposal becomes
        active by the absolute fallback deadline $e_i(v)+\theta_E$, the party
        instead broadcasts $\langle\echoskip,v\rangle$.  A well-formed
        proposal delivered
        after that deadline can no longer be echoed; it is still fixed and
        its references are authorized for repair, and in Vantage the
        responsive frontier processes it.  No later echo-stage
        statement for $v$ is allowed; the pre-entry and post-entry routes share
        this single immutable response state.
  \item \textbf{Ready stage.}  Entry to view $v$ arms its ready-stage refusal
        deadline.  Independently of entry, while the initial response is pending the
        party continuously evaluates the first-hand echoes it has counted.
        As soon as a quorum of proposal
        echoes names a common view proposal $B_v=(C,T)$, it broadcasts the
        proposal-ready statement $\langle\rstat{ready},v,C,T,g\rangle$,
        where, over all echoes counted at emission, the grade $g$ is $1$ if it has counted $\quorum$ grade-$1$
        proposal echoes for $B_v$, $0$ if it has counted $\quorum$ grade-$0$
        proposal echoes for $B_v$, and the mixed grade $\mathsf{mix}$
        otherwise.  Grades $0$ and $1$ finalize the sender's ready stage.  A
        mixed grade is provisional: while fewer than all $n$ echo-stage
        responses have been counted and the ready deadline has not fired, the
        sender continues evaluating echoes for the same $B_v$.  If one grade
        reaches $\quorum$, it broadcasts one READY-$g$ refinement for $B_v$
        and finalizes.  Counting all $n$ echo-stage responses without such a
        quorum, or reaching $e_i(v)+\theta_R$, instead finalizes the mixed
        response without another statement; deadline ties process delivered
        messages first.  The refinement is not a second initial response and
        does not invoke the WISH progress hook.  The Vantage wrapper does not
        quarantine on this provisional response.  If formal entry has occurred and
        no initial proposal-ready has fired by $e_i(v)+\theta_R$, it
        broadcasts the refusal $\langle\noready,v\rangle$.  Wishes and entry
        notifications never satisfy the echo gate or count toward a ready
        quorum.
  \item \textbf{Completion and direct result.}  Upon counting a ready quorum
        for a common $B_v=(C,T)$, disregarding grades and counting each
        author once, record
        $\rstat{complete}(v)\to B_v$.  On the first homogeneous ready quorum,
        emit $\rstat{direct-seal}(v)\to\gfull(C,T)$ for grade $1$ or
        $\rstat{direct-seal}(v)\to\gcore(C)$ for grade $0$.
        A completed view with no homogeneous ready quorum remains
        $\gopen(C,T)$ at the Direct AGB boundary.  Vantage may close it through
        the resolver adapter of \Cref{app:protocol}.
        In the same local transition that first completes a proposal, if
        neither grade has a ready quorum, the Vantage wrapper applies the
        quarantine update of \Cref{sec:agb-interface} to its non-exact-quorum tip
        entries.  If a homogeneous quorum completes the proposal, completion
        and direct seal occur atomically and no quarantine is created.
        Ready counting continues after completion, so a late-arriving ready
        that closes a homogeneous quorum produces a direct result.  Vantage
        submits every such result to its caller-owned $\rstat{try-seal}$ arbiter;
        a compatible result arriving after a terminal Vantage seal is ignored
        by that arbiter.
        Neither completion nor the direct result waits for formal entry.  When a
        homogeneous quorum is the first completing quorum, completion and the
        corresponding $\rstat{direct-seal}$ event occur in the same local step.
\end{enumerate}

Vantage adds the following caller-side rule without changing Direct AGB.  For
a proposal $B=(C,T)$, a \emph{matching response} is a grade-$1$ proposal echo
for exactly $B$.  Immediately before a correct
party sends a matching response, it records an optimistic lock $L_i(v,B)$.
The lock is \emph{active} exactly while the party has counted fewer than
$f{+}1$ first-hand echo-stage responses for $v$ that are not matching responses
for $B$.  Thus a lock may be born inactive, and once inactive it never becomes
active again.  Upon counting matching responses from all $n$ parties, the
wrapper emits
\[
  \fastseal(v)\to\gfull(C,T)
\]
once and submits that result to $\rstat{try-seal}$.  It does not fire
$\rstat{complete}$ or $\rstat{direct-seal}$, expose a core separately, or create
a resolution witness; R3 and R4 continue normally.  The resolver-side use of
$L_i$ is defined in \Cref{sec:protocol}.

The interface and R1--R4 above define Direct AGB.  Each caller initializes
$\mathit{quarantine}_i[a]=\bot$, its two retry counters to zero, and
$\mathit{ownProposalTurn}_i=0$ once per session; this caller state persists
across views.  \Cref{alg:agb} gives
Vantage's concrete realization in event form; its frontier, WISH, manifest
construction, and $\rstat{try-seal}$ lines are caller-side wrappers rather than
additional Direct AGB requirements.  \Cref{alg:fastseal} separately gives the
Vantage-only fast- and skip-seal and lock-release hooks so they are not
mistaken for Direct AGB rules.
\begin{algorithm*}[tp]
\SetAlgoSkip{}
\caption{Direct Availability-Graded Broadcast as instantiated by Vantage for
  view $v$ at party $p_i$ (R1--R4; $Q=n-f$)}
\label{alg:agb}
\scriptsize\linespread{0.74}\selectfont
\Glob{}{
  $\mathit{fixed}_v\gets\bot$;
  $\mathit{echo}_v\gets\mathsf{pending}$;
  $\mathit{ready}_v\gets\mathsf{pending}$
    \Comment*[r]{later $\mathsf{provisional}(B)$ or $\mathsf{final}$}
  $\mathit{completed}_v,\mathit{directed}_v\gets\bot$\;
  $\mathit{active}_v,\mathit{activated}_v,\mathit{proposed}_v
    \gets\mathsf{false}$
    \Comment*[r]{$\mathit{active}_v$: activated with a well-formed fixed proposal}
  $e_i(v),\rho_i(v)\gets\bot$;
  $\theta_E=3\Delta$;
  $\theta_R=4\Delta$;
  caller state (\Cref{sec:agb-interface,sec:pacemaker}): responsive frontier
    $a_i$, wish statistic $\omega_i^+$, $\mathit{ownProposalTurn}_i$, and
    per-author quarantine slots with $(\mathit{probeAttempt},\mathit{probeDue})$
}

\Upon{$p_i=\proposer(v)$, not $\mathit{proposed}_v$, and either
  $a_i\ge v-1$ or $\omega_i^+\ge v$}{
  $\mathit{proposed}_v\gets\mathsf{true}$;
  $\mathit{ownProposalTurn}_i\gets\mathit{ownProposalTurn}_i+1$;
  clear quarantine slots that are exact-$Q$-available or core-contained;
  $(C,T)\gets\textsc{FormManifests}_i(\mathit{ownProposalTurn}_i)$
    \Comment*[r]{exact-$Q$ core; stable tips; at most one due probe}
  for the selected probe, if any: $\mathit{probeAttempt}\gets\mathit{probeAttempt}+1$
    and $\mathit{probeDue}\gets\mathit{ownProposalTurn}_i+2^{\mathit{probeAttempt}-1}$\;
  \textbf{broadcast} $\langle\rstat{propose},v,C,T\rangle$
    \Comment*[r]{R1}
}

\Upon{local $\rstat{enter}(v)$ or $\rstat{activate}(v)$ from the Vantage
  pacemaker/frontier wrapper}{
  \If{this is $\rstat{enter}(v)$ and $e_i(v)=\bot$}{
    $e_i(v)\gets\mathsf{now}$; $\mathit{activated}_v\gets\mathsf{true}$;
      raise $a_i$ to at least $v-1$;
    arm deadlines $e_i(v)+\theta_E$ and $e_i(v)+\theta_R$
  }
  \If{this is $\rstat{activate}(v)$}{
    $\mathit{activated}_v\gets\mathsf{true}$
  }
  \If{$\mathit{activated}_v$ and $\mathit{fixed}_v=(C,T)$}{
    $\mathit{active}_v\gets\mathsf{true}$
  }
}

\Upon{$\mathit{fixed}_{a_i+1}$ is a well-formed proposal}{
  $a_i\gets a_i+1$; invoke $\rstat{activate}(a_i)$
    \Comment*[r]{responsive frontier; repeats through buffered proposals}
}

\Upon{the first direct $\langle\rstat{propose},v,C,T\rangle$ from
  $\proposer(v)$ while $\mathit{fixed}_v=\bot$}{
  $\rho_i(v)\gets\mathsf{now}$\;
  \eIf{$(C,T)$ is malformed}{
    $\mathit{fixed}_v\gets\mathsf{reject}$
      \Comment*[r]{later proposer versions remain ignored}
  }{
    $\mathit{fixed}_v\gets(C,T)$; store $B_v$;
    invoke \textsc{Authorize} for each lane prefix named by $C$ or $T$
      \Comment*[r]{typed repair; never provenance}
    \lIf{$\mathit{activated}_v$}{$\mathit{active}_v\gets\mathsf{true}$}
  }
}

\Upon{$\mathit{echo}_v$ is pending and $\mathit{active}_v$,
  evaluated whenever local state changes}{
  let $(C,T)\gets\mathit{fixed}_v$\;
  \If{$\textsc{CoreOK}_i(C)$ and $\textsc{TipOK}_i(C,T)$}{
    $\mathit{echo}_v\gets\mathsf{sent}$
      \Comment*[r]{one immutable echo-stage response}
    $\textsc{Respond}(\langle\rstat{echo},v,C,T,1\rangle)$
      \Comment*[r]{R2 positive path}
  }
}

\Upon{$e_i(v)\ne\bot$, $\mathit{fixed}_v=(C,T)$, $\mathit{echo}_v$ is pending,
  and the fallback deadline
  $d_i^E(v)=\min\{\max\{e_i(v),\rho_i(v)\}+\Delta,\,e_i(v)+\theta_E\}$
  passes}{
  $\mathit{echo}_v\gets\mathsf{sent}$;
  \eIf{$\textsc{CoreOK}_i(C)$}{
    $\textsc{Respond}(\langle\rstat{echo},v,C,T,0\rangle)$
  }{
    $\textsc{Respond}(\langle\echoskip,v\rangle)$
  }
}

\Upon{$e_i(v)\ne\bot$, the deadline $e_i(v)+\theta_E$ passes, $\mathit{echo}_v$
  is pending, and no active well-formed fixed proposal exists}{
  $\mathit{echo}_v\gets\mathsf{sent}$;
  $\textsc{Respond}(\langle\echoskip,v\rangle)$
    \Comment*[r]{no active echoable proposal}
}

\Upon{the counted echo set changes and $\mathit{ready}_v$ is pending}{
  \If{some $B=(C,T)$ has $Q$ proposal echoes}{
    $g\gets 1$ if $Q$ counted echoes for $B$ have grade $1$;
      $0$ if $Q$ have grade $0$; otherwise $\mathsf{mix}$\;
    $\mathit{ready}_v\gets\mathsf{provisional}(B)$ if
      $g=\mathsf{mix}$, else $\mathsf{final}$;
    $\textsc{Respond}(\langle\rstat{ready},v,C,T,g\rangle)$
      \Comment*[r]{R3; no entry or own-echo guard}
  }
}

\Upon{the counted echo set changes and
  $\mathit{ready}_v=\mathsf{provisional}(B)$}{
  \If{$Q$ counted echoes for $B$ have one grade $g\in\{0,1\}$}{
    $\mathit{ready}_v\gets\mathsf{final}$;
    \textbf{broadcast} $\langle\rstat{ready},v,C,T,g\rangle$ with the own
      \wish{} high-watermark piggybacked
      \Comment*[r]{same-value homogeneous refinement; no second WISH hook}
  }
  \lIf{$\mathit{ready}_v$ is still provisional and all $n$ echo-stage
    responses are counted}{$\mathit{ready}_v\gets\mathsf{final}$}
}

\Upon{$e_i(v)\ne\bot$ and the deadline $e_i(v)+\theta_R$ passes}{
  \If{$\mathit{ready}_v=\mathsf{provisional}(B)$}{
    $\mathit{ready}_v\gets\mathsf{final}$
      \Comment*[r]{finalize the mixed response locally; no second statement}
  }
  \If{$\mathit{ready}_v=\mathsf{pending}$}{
    $\mathit{ready}_v\gets\mathsf{final}$;
    $\textsc{Respond}(\langle\noready,v\rangle)$
  }
}

\Upon{the counted ready set changes}{
  count each author once; a same-$B$ refinement after READY-mix changes only
    that author's grade tally\;
  \ForEach{$B=(C,T)$ named by counted proposal-ready statements}{
    \If{$\mathit{completed}_v=\bot$ and at least $Q$ name $B$}{
      $\mathit{completed}_v\gets B$; store $B$; authorize lane prefixes named
        by $C$;
      \If{fewer than $Q$ grade-$1$ and fewer than $Q$ grade-$0$ READYs name $B$}{
        put each non-exact-$Q$ tip in its empty quarantine slot and set its
          $(\mathit{probeAttempt},\mathit{probeDue})\gets(0,
          \mathit{ownProposalTurn}_i+1)$
      }
      fire $\rstat{complete}(v)\to B$
        \Comment*[r]{R4; makes the core irrevocable; the composition runs the cursor}
    }
    \If{$\mathit{directed}_v=\bot$ and at least $Q$ grade-$1$ readies name $B$}{
      $\mathit{directed}_v\gets\gfull(C,T)$;
      fire $\rstat{direct-seal}(v)\to\gfull(C,T)$ and submit it to
        the caller's $\rstat{try-seal}$
    }
    \ElseIf{$\mathit{directed}_v=\bot$ and at least $Q$ grade-$0$ readies name $B$}{
      $\mathit{directed}_v\gets\gcore(C)$;
      fire $\rstat{direct-seal}(v)\to\gcore(C)$ and submit it to
        the caller's $\rstat{try-seal}$
    }
  }
}

\Fn{\textsc{Respond}$(m)$}{
  attach $\textsc{AckClaims}_i(B)$ outside the identity of any proposal ECHO;
  apply \Cref{sec:pacemaker}'s \wish{} hook and piggyback the own \wish{}
    high-watermark outside $m$'s identity; \textbf{broadcast} $m$
}

\end{algorithm*}

\begin{algorithm}[tp]
\SetAlgoSkip{}
\caption{Vantage's fast- and skip-seal wrappers and resolution stances at $p_i$}
\label{alg:fastseal}
\footnotesize
\Comment{Hooks into Direct AGB; $Q=n-f$; count only the first direct response
  per author}
\Glob{}{
  $\mathit{fastLock}_v\gets\bot$;
  $\mathit{fastReleased}_v,\mathit{fasted}_v\gets\mathsf{false}$;
  for every target $u$: the persistent stance $z_i(u)\gets\mathsf{free}$ and
  $\mathit{skipSent}_u,\mathit{skipSubmitted}_u\gets\mathsf{false}$
}

\Upon{about to emit grade-$1$ ECHO for $B=(C,T)$ in view $v$}{
  $\mathit{fastLock}_v\gets(C,T)$\;
  $\mathit{fastReleased}_v\gets\mathsf{true}$ iff already $f+1$ counted
    echo-stage responses are not matching responses for $B$\;
  emit the Direct AGB ECHO
    \Comment*[r]{the lock is recorded first}
}

\Upon{the counted view-$v$ echo set changes}{
  \If{$\mathit{fastLock}_v=(C,T)$ and $f+1$ counted responses are not
    matching responses for $(C,T)$}{
    $\mathit{fastReleased}_v\gets\mathsf{true}$
      \Comment*[r]{the lock is active iff $\mathit{fastLock}_v\ne\bot$ and not $\mathit{fastReleased}_v$}
  }
  \If{not $\mathit{fasted}_v$ and all $n$ responses are matching responses
    for some $(C,T)$}{
    $\mathit{fasted}_v\gets\mathsf{true}$; fire
      $\fastseal(v)\to\gfull(C,T)$; submit it to $\rstat{try-seal}$
      \Comment*[r]{no Direct AGB completion}
  }
}

\Upon{$R_i(u)=\noready$, at least $Q$ counted target-view responses are
  \echoskip{}, $z_i(u)=\mathsf{free}$, the known terminal outcome
  for $u$ is absent or $\gskip$, and not $\mathit{skipSent}_u$}{
  $z_i(u)\gets\mathsf{skip\mbox{-}voted}$;
  $\mathit{skipSent}_u\gets\mathsf{true}$\;
  \textbf{broadcast} $\langle\skipvote,u\rangle$
    \Comment*[r]{persist the stance first; no $f+1$ amplification}
}

\Upon{at least $Q$ direct $\langle\skipvote,u\rangle$ statements are counted
  and not $\mathit{skipSubmitted}_u$}{
  $\mathit{skipSubmitted}_u\gets\mathsf{true}$; fire
    $\skipseal(u)\to\gskip$; submit it to $\rstat{try-seal}$
    \Comment*[r]{no Direct AGB completion}
}

\end{algorithm}

Grades, including a permitted homogeneous refinement, are outside the counted
view-proposal identity: mixed grades can never prevent completion.  A residual
mixed grade leaves only the tip open; a refinement may instead close a later
homogeneous quorum.  Grade $0$ is positive support for a
terminal core-only outcome, not merely the absence of grade $1$.  A ready
carries no transferable echo certificate; its support is checked only by its
sender.

\begin{lemma}[Responsive proposal pacing]\label{lem:responsive-pacing}
Suppose every party is correct and $v$ is past the cutoff of
\Cref{lem:pacemaker}(c).  Let $\beta_v$ be the time $B_v$ is sent.  Then
\[
  \beta_{v+1}\le \max_{u\le v}\beta_u+\delta ,
\]
so the least not-yet-proposed view always follows within $\delta$ of the
latest send.  Call $B_v$ \emph{prefix-extending} when every lower view was
sent by $\beta_v$.  Suppose $B_v$ is prefix-extending, the
$\textsc{AuthorOK}$ and tip-anchoring checks for $B_v$ pass everywhere by
$\beta_v+\delta$.  Then every
correct party obtains the ordinary Direct AGB full result by
$\beta_v+3\delta$, and every correct party also emits and
submits the Vantage fast result $\gfull(C_v,T_v)$ by $\beta_v+2\delta$.
The send times need not be monotone in the view number: the early wish
trigger of R1 may send a proposal before an older one, and receivers buffer
it until its contiguous prefix arrives or formal entry activates it.
\end{lemma}
\begin{proof}
Because $v$ is past the cutoff, $\beta_v\ge\mathrm{GST}$.  Let
$z=\max_{u\le v}\beta_u$.  In an all-correct execution every view $u\le v$
is proposed, and by $z$ every such $B_u$ has been sent; the in-flight
convention delivers all of them everywhere by $z+\delta$, and in-order
processing of buffered fixed proposals from genesis raises every correct
frontier to $v$ by $z+\delta$ --- a proposal delivered after its echo
deadline is still fixed and processed by the frontier (R2).  Hence $\proposer(v+1)$ sends $B_{v+1}$ by
$z+\delta$ unless it already did, proving the first claim.

For a prefix-extending $B_v$, the same delivery argument with $z=\beta_v$
activates $B_v$ at every correct party by $\beta_v+\delta$, so
under the stated checks every positive echo gate holds by $\beta_v+\delta$.  No
fallback response can preempt it: fallback timers exist only after formal
entry, and the first correct formal entry satisfies $s_v\ge \beta_v-\delta$,
since the all-correct spread of \Cref{lem:pacemaker}(c) makes the proposer
enter and send by $s_v+\delta$.  Party $p_i$'s echo fallback window starts
at $t\ge$ its direct receipt of $B_v\ge \beta_v$ and therefore ends at
$\min\{t+\Delta,\,e_i(v)+\theta_E\}\ge \beta_v+\delta$, using
$e_i(v)+\theta_E\ge \beta_v+3\Delta-\delta$ from $e_i(v)\ge s_v$; the absolute
\echoskip{} deadline is the same bound, and the \noready{} deadline
$e_i(v)+\theta_R\ge \beta_v+4\Delta-\delta$ exceeds the ready gate time
$\beta_v+2\delta$ below.  Messages delivered by a deadline are processed before
its timeout branch, so every echo is grade~$1$ and fires by $\beta_v+\delta$.
The echoes arrive everywhere by $\beta_v+2\delta$; these are all $n$ matching
responses for $B_v$, so the Vantage wrapper emits its fast full result then.
Independently, positive proposal-ready fires even before formal entry, and
those statements arrive everywhere by $\beta_v+3\delta$, producing the
ordinary Direct AGB full result.
All proposal echoes in this all-correct execution have grade $1$, so no
completed-open quarantine is created.
\end{proof}

\subsection{Guarantees}\label{sec:agb-guarantees}

Every quorum below is first-hand and counts distinct authors.  Manifests travel
by value; replacing them by hashes is an optimization outside these proofs.
Every decisive threshold is a quorum or $f{+}1$, and WISH never contributes to
one.  R2 treats the eligibility hooks as opaque local
predicates; only their witness contract and the composition results depend on
their content.

Every correct proposal-ready sender counted a proposal-echo quorum for the
named proposal; a correct grade-$g$ sender, $g\in\{0,1\}$, counted a
grade-$g$ echo quorum.  Hence any two correct proposal-readies name the same
proposal, and correct grade-$1$ and grade-$0$ proposal-readies cannot coexist:
their grounding echo quorums would intersect in a correct immutable echoer.

The direct safety theorem below treats the eligibility hooks
opaquely.  It concerns only completions, the initial-response and refinement
discipline, and direct R4 seals.  Cross-party validity and agreement of resolver-supplied terminal seals
are separate composition properties proved in
\Cref{lem:direct-resolution,lem:vantage-resolver,thm:resolved-agb}.

\begin{theorem}[Safety]\label{thm:agb-safety}
In every execution, regardless of timing:
\begin{enumerate}[label=\textup{(\alph*)},leftmargin=2.0em,topsep=2pt,itemsep=3pt]
  \item \textup{(Local uniqueness and sender integrity)} Each correct party
        completes $v$ at most once and emits at most one $\rstat{direct-seal}$
        result.  Every completion names a proposal
        received directly from $\proposer(v)$ by a correct echoer.  If
        $\proposer(v)$ is correct and proposes
        $B_v=(C,T)$, every completion names that proposal and every direct
        result for view $v$ names the payload $(C,T)$.
  \item \textup{(Value agreement)} All correct proposal-ready statements,
        proposal-ready quorums, and completions for $v$ derive from one view
        proposal $B_v=(C,T)$.
  \item \textup{(Grade compatibility)} A grade-$1$ proposal-ready quorum and a
        grade-$0$ proposal-ready quorum cannot both exist.  Hence correct
        parties never emit conflicting direct outcomes $\gfull$ and
        $\gcore$; either may coexist with the same $\gopen(C,T)$.  Parties
        that have not counted a completion quorum may still remain unknown.
  \item \textup{(Eligibility witness)} Completion implies
        $\textsc{Witnessed}(C)$; a grade-$1$ proposal-ready quorum also implies
        $\textsc{Witnessed}(T)$.  In Vantage these properties mean that every
        named valid lane prefix is author-backed and retrievable.
  \item \textup{(Ready--refusal exclusion)} A proposal-ready quorum and a
        no-ready quorum for view $v$ cannot both exist.
\end{enumerate}
\end{theorem}

\begin{proof}
\emph{(a)} A correct party acts only on its first completion and first direct
result.  Every triggering proposal-ready quorum contains a correct sender,
whose grounding echo quorum contains a correct echo.  That echo was sent only
for a proposal received directly from $\proposer(v)$, proving sender integrity.
Against a correct proposer it names the sender's unique $B_v$.

\emph{(b)} The grounding fact gives one proposal for all correct
proposal-readies.  Every proposal-ready quorum contains a correct sender, so
all such quorums and the completions they trigger use that proposal.

\emph{(c)} Quorums of grade-$1$ and grade-$0$ proposal-readies would contain
correct senders of both grades, which the grounding fact excludes.

\emph{(d)} A completing quorum contains a correct ready sender, whose grounding
echo quorum contains a correct party that passed $\textsc{CoreOK}$; the
eligibility contract gives $\textsc{Witnessed}(C)$.  For a grade-$1$ ready
quorum, choose a correct grade-$1$ sender and apply the same argument to its
grade-$1$ echo quorum and $\textsc{TipOK}$.  The concrete Vantage consequence
follows from \Cref{lem:availability-witness,lem:author-witness}.

\emph{(e)} The two quorums count $\quorum$ distinct initial ready-stage
authors each, so they share at least one correct party.  That party's initial
choice cannot be both proposal-ready and \noready{}; a READY-mix refinement
does not add another author to either quorum.
\end{proof}

\paragraph{Graded Complete--Adopt relation.}
Reading a correct party's initial proposal-ready response as
$\rstat{adopt}(v,B,g)$, $g\in\{0,\mathsf{mix},1\}$, with a possible monotone
update from $\mathsf{mix}$ to a homogeneous grade, and reading its \noready{}
response as $\rstat{no-adopt}(v)$ gives a graded analogue of BBCA's
Complete--Adopt interface~\cite{bbcachain}, directly from
\Cref{thm:agb-safety} and the grounding argument above: all correct adopters
name one proposal, and grade-$1$ and grade-$0$ adopters cannot coexist; a
completion implies at least $f{+}1$ correct adopters of its proposal, while
a quorum of $\rstat{no-adopt}$ responses excludes every completion; and each
correct adopter's grounding echo quorum establishes $\textsc{Witnessed}(C)$
through the eligibility contract, a grade-$1$ adopter's likewise
$\textsc{Witnessed}(T)$.  These are names for already emitted responses, not
a probe operation, and no totality is added: without the Vantage resolver,
some correct parties may never complete or obtain a direct result for a
Byzantine-proposer instance.

\begin{theorem}[Conditional termination]\label{thm:agb-live}
Let $s\ge\mathrm{GST}$ be the first correct $\rstat{enter}(v)$ event, assume
every correct party enters $v$ by $s+2\delta$, and assume the view proposal is
broadcast no earlier than GST.  Suppose $\proposer(v)$ is correct and invokes
$\rstat{propose}(v,C,T)$ with $\textsc{Formed}_v(C,T)$ no later than
time $s+\delta$.  Assume
that at every correct party $\textsc{CoreOK}_i(C)$ holds at some instant no
later than its echo-window closure; persistence then keeps it true through the
fallback action.  Then
every correct party fires
$\rstat{complete}(v)$ by $s+4\delta+\Delta\le s+5\Delta$, returning $B_v$.
\end{theorem}

\begin{proof}
By premise, all correct parties formally enter by $s+2\delta$ and the sender
invokes its well-formed input by $s+\delta$.  No correct receiver
rejects it, so the proposal is received everywhere by $s+2\delta$ and, since
formal entry activates any fixed proposal, active everywhere by $s+2\delta$.
For each correct party
$p_i$, both $e_i(v)\le s+2\delta$ and its direct proposal receipt occur by
$s+2\delta$.  Thus the fallback-window start $t$ is at most $s+2\delta$, while
$e_i(v)+\theta_E\ge s+3\Delta\ge t+\Delta$; its window therefore ends at
$t+\Delta$, at least $\Delta$ after direct receipt.  The assumed core check
passes within every echo window, so every correct party sends a
proposal echo by $s+2\delta+\Delta$.  This is no later than
the earliest absolute echo deadline $s+3\Delta$; R2's absolute \echoskip{}
branch requires that no echoable proposal be active by that deadline, so
under the core-check premise the window branch emits a proposal echo instead.

The echoes arrive everywhere by
$s+3\delta+\Delta\le s+4\Delta$, no later than the earliest ready deadline
$s+4\Delta$; messages delivered by a deadline are processed before its
timeout branch.  Every correct party therefore sends proposal-ready, and those
statements arrive everywhere by $s+4\delta+\Delta$, proving completion.
\end{proof}

For every view beyond the cutoff of \Cref{lem:pacemaker}(b), Vantage's WISH
service supplies the entry-spread premise of \Cref{thm:agb-live}, and the
cutoff's definition excludes a pre-GST correct proposal for that view.  The
R1 construction supplies a formed value, and its early wish trigger supplies
the invocation premise: the wish quorum counted at the first correct entry
contains $f{+}1$ correct wishes supporting $v$, sent by $s$ and
delivered to $\proposer(v)$ by $s+\delta$ under the in-flight convention,
so $\omega^+$ supports $v$ there and the proposer invokes by $s+\delta$
unless it already has.  For the Vantage eligibility hooks, every
correct-proposer core entry is exact-$\quorum$-available before the send time
$\beta\ge\mathrm{GST}$ and therefore becomes $(f{+}1)$-available everywhere
by $\beta+\delta$ by \Cref{lem:availability-witness}(iv).  A quorum assembled using
integer claims has the same eventual property by
\Cref{lem:availability-witness}(iii), but another receiver may first have to
retrieve the anchoring chain; it therefore cannot promote a correct
proposer's core until exact-position support is also present.  The theorem
otherwise depends only on Direct AGB's environment
events and message rules.

For a view with a correct proposer beyond this cutoff, none of these
premises requires $\textsc{TipOK}$.  If a correct party cannot verify the tip,
R2 emits a grade-$0$ proposal echo once its fallback window closes, while R3
and R4 count proposal echoes and readies regardless of grade.  Thus every
correct party still completes by $s+4\delta+\Delta$, regardless of which
correct parties fail which tip checks.  Completion fixes the witnessed core,
and once the output cursor reaches the view and obtains its lane prefixes, it emits
$\mathsf{Core}_D(C_v)$ without waiting for the tip; only the tip may remain
open (\Cref{thm:agb-safety,sec:protocol}).

\begin{corollary}[Transaction latency]\label{cor:tx-latency}
Suppose every party is correct, each party places arriving transactions in a
new valid data block and publishes it at once.  Suppose the execution is in a
non-Zeno responsive suffix beginning at a view $r$ past the cutoff of
\Cref{lem:pacemaker}(c): every view below $r$ has already been proposed,
and sealed and fully output at every correct party, by $\beta_r$.  Let $\beta_v$
be the send time of $B_v$.  A transaction arriving at time $t\ge \beta_r$ is sequenced
within
\[
  \begin{gathered}
    L_{\mathrm{tip}}\le\delta+\delta+2\delta=4\delta,\\[-2pt]
    \text{and, ignoring tips,}\\[-2pt]
    L_{\mathrm{core}}\le
      \delta+\delta+2\delta+\delta+2\delta=7\delta.
  \end{gathered}
\]
If the next prefix-extending proposal opportunity is sent at time $t+\delta$
after processing the new block, the tip path is aligned and sequences
it by $t+3\delta$.  If an anchoring proposal is sent at $t+\delta$ after
processing the new block and a later prefix-extending proposal is sent at
$t+3\delta$ after processing its ECHO acknowledgment quorum, the core path is
aligned and sequences it by $t+5\delta$ through fast seal, or by $t+6\delta$
through the Direct AGB READY path.
Over any finite union of complete responsive-opportunity intervals, if the
eligibility time is wall-clock-uniform and independent of the fixed proposal
schedule, the expected wait for the next prefix-extending opportunity is at
most $\delta/2$.  Under this additional convention,
\[
  \mathbf{E}[L_{\mathrm{tip}}]\le
    \delta+\frac{\delta}{2}+2\delta=3.5\delta,
  \qquad
  \mathbf{E}[L_{\mathrm{core}}]\le
    \delta+\frac{\delta}{2}+2\delta+\delta+2\delta=6.5\delta.
\]
\end{corollary}
\begin{proof}
We first factor out the proposal wait.  Fix either path and let
$x\ge\beta_r$ be a time by which every proposer satisfies
$\textsc{DirectPub}$ for a correct-author reference covering the block; on the
core path, require in addition that every proposer has counted $\quorum$
acknowledgments for that reference.
Every view is eventually proposed by \Cref{lem:pacemaker}(a) and R1, while
non-Zeno execution permits only finitely many proposal sends by $x$.  Hence
there is a least $k\ge r$ with $\beta_k>x$.  Since $x\ge\beta_r$, we have
$k>r$.  The suffix premise bounds $\beta_u\le\beta_r\le x$ for $u<r$, minimality
bounds $\beta_u\le x$ for $r\le u<k$, and
\Cref{lem:responsive-pacing} gives
\[
  x<\beta_k\le\max_{u<k}\beta_u+\delta\le x+\delta.
\]
The inclusion rule makes $B_k$ cover the block.  On the core path its newest
exact-$\quorum$-available core does so.  On the tip path, either the core
already covers it or the newest directly published lane tip does.  Indeed, in
this all-correct suffix every claim at its named position is exact: once a
general quorum is reached for a named tip, that same tip is exact-quorum
available and becomes the core candidate, so no stable-confirmation target
above the core displaces the newest tip.  The author's lane is unforked.

Every proposal in the all-correct suffix passes its
$\textsc{AuthorOK}$ and tip-anchoring checks everywhere within one delay of
its send: every named lane prefix was directly published by its correct
author before the send and reaches every party within one further delay.
$B_k$ is prefix-extending by the send bounds above, so
\Cref{lem:responsive-pacing} fast-seals it $\gfull$ within $2\delta$ of its
send; thus $B_k$ seals by $x+3\delta$.  Every lower view was sent by $x$
--- below $r$ by the suffix premise, in $[r,k)$ by those bounds --- so, by
the delivery argument of that lemma, every proposal $B_u$ with $u<k$ is
fixed and, through in-order frontier processing, active at every correct
party by $x+\delta$, without waiting for formal entry.  For a suffix view
$u\in[r,k)$ the checks pass by $\beta_u+\delta\le x+\delta$, so every
correct party's grade-$1$ gate holds by $x+\delta$.  No fallback response
preempts it: a fallback window exists only after entry, starts at
$t=\max\{e_i(u),\rho_i(u)\}$, and ends no earlier than the gate time
$\max\{t,\beta_u+\delta\}$ on both arms.  Indeed, $t+\Delta$ dominates both
$t$ and $\beta_u+\Delta$ from $\rho_i(u)\ge\beta_u$; and
$e_i(u)+\theta_E\ge\beta_u+3\Delta-\delta$ dominates $\beta_u+\delta$ as
well as $t\le\max\{e_i(u),\,\beta_u+\delta\}$, using
$e_i(u)\ge s_u\ge\beta_u-\delta$ --- valid for every suffix view, since a
correct
proposer invokes no later than its own entry and the all-correct entry
spread past the cutoff is $\delta$ --- and $\rho_i(u)\le\beta_u+\delta$.  Hence
all $n$ matching responses for each suffix view in $[r,k)$ are sent by
$x+\delta$, counted by $x+2\delta$, and fast-seal it.  The drained earlier prefix therefore makes the block
sequenced by $x+3\delta$.

For the tip path, direct publication makes the block eligible at
$x=t+\delta$, giving $4\delta$.  If a prefix-extending proposal is sent exactly
at that time after processing the block, it covers the block and all lower
views have already been proposed, so sequencing occurs by $t+3\delta$.

For the core path, let $x_0=t+\delta$, when every party has directly received
the correct-author block.  The first prefix-extending proposal after $x_0$ is
sent by $x_0+\delta$.  It names a tip covering the block, and every correct
party's ECHO carries the exact-position acknowledgment.  Those ECHOs are
counted everywhere within $2\delta$ of that proposal's send, so the block is a
core candidate everywhere by $t+4\delta$.  Applying the preceding proposal
wait and fast-seal argument once more gives $t+7\delta$.  In the aligned case,
the anchoring proposal is sent at $t+\delta$, its exact claims are counted by
$t+3\delta$, and a core proposal sent then fast-seals by $t+5\delta$.
The same send-to-seal theorem gives $t+6\delta$ when the core proposal instead
finishes through the three-delay READY path.

For the expected bounds, let $S_j$ be the chronological proposal-send events
that fill the least not-yet-proposed numerical view and hence extend the fully
proposed prefix.  Already-sent higher-view proposals may enlarge the resulting
jump, but are not fresh opportunities.  If the prefix reaches view $q$ at
$S_j$, then the responsive-pacing rule gives
$\beta_{q+1}\le\max_{u<q+1}\beta_u+\delta\le S_j+\delta$; hence the next fresh
event satisfies $g_j=S_{j+1}-S_j\le\delta$.  Conditional on the fixed schedule,
a wall-clock-uniform eligibility time over any finite union of complete
intervals $[S_j,S_{j+1})$ therefore gives
\[
  \mathbf E[W]
    =\frac{\sum_j g_j^2}{2\sum_j g_j}
    \le \frac{\delta}{2}.
\]
For the tip path, adding publication and the $2\delta$ fast-sealing time gives
$3.5\delta$.  For the core path, the wait to the first anchoring proposal has
the same $\delta/2$ expectation; its ECHO census costs $2\delta$, the wait to
the subsequent core proposal is at most $\delta$, and fast sealing costs
$2\delta$, giving $6.5\delta$.  Any limiting time average, when it exists,
inherits the same inequality.
\end{proof}

\paragraph{Communication.}
Per instance, Direct AGB exchanges $O(n^2)$ messages: one proposal
broadcast ($n$ messages), one echo-stage and one initial ready-stage broadcast
per party ($n^2$ messages each), and at most one homogeneous READY refinement
per party after READY-mix (at most another $n^2$ messages).  Thus the explicit
worst-case count is $3n^2+O(n)$, while the all-grade-$1$ common path has no
refinement and retains $2n^2+O(n)$.  \echoskip{} and \noready{} replace
initial statements and never add any.  In Vantage, responses piggyback the sender's
current WISH; threshold amplification adds at most one own WISH
broadcast per raised view, within the same $O(n^2)$ per-view message bound.
Fast sealing only reads the existing ECHO set and adds no message; its lock
is one payload identifier plus the already maintained first-response counts.
Completed-open quarantine adds one reference and two scalar counters per
author and no message; its probes reuse ordinary proposals.
The grounded skip shortcut adds at most one short \skipvote{} broadcast per
party on an eligible failed view, another $O(n^2)$ messages in that view and
none on the all-correct favorable path.  Its per-target stance is one
persistent trit.  Thus Vantage's complete per-view statement layer remains
$O(n^2)$ messages.
In Vantage's proved by-value instantiation, manifests are carried by value and well-formedness
bounds each at $n$ entries --- for Byzantine proposals too, since correct
parties never echo a malformed proposal --- so every statement is
$O(n\lambda)$ bits and a view costs
$O(n^3\lambda)$ bits, including the proposal-relative ACK fields.\footnote{Scalar fields --- view numbers, heights,
lengths --- cost $O(\log\nu)$ bits, where $\nu$ bounds those quantities in
the execution prefix.  The displayed bounds use the deployment regime
$\log\nu=O(\lambda)$; otherwise each $O(\lambda)$ scalar term reads
$O(\lambda+\log\nu)$.  This accounting convention does not bound the
unbounded-execution liveness model.}  The default implementation names a
single proposal by hash in later statements and puts an $O(n)$-bit
availability bitmap on each ECHO.  Its aligned common case is therefore
$O(n^2\lambda+n^3)$ bits per view; sparse $(\text{position},\text{distance})$
exceptions cost $O(n^3(\log n+\log\nu))$ bits.
\Cref{sec:optimizations} gives the accounting.  Digest-naming is a wire
optimization; the ECHO acknowledgment semantics are part of the protocol.
Data dissemination is accounted separately: a data
block of $\ell_b\le\ell_{\max}$ payload bits occupies
$\ell_b+O(\lambda)$ bits on the wire and costs $O(n(\ell_b+\lambda))$ for its
author's typed publication; acknowledgments add no standalone message and
were counted in the ECHO layer above.  Thus blocks of size
$\ell_b=\Omega(n\lambda)$ cost $O(n)$ bits of communication per payload bit ---
the cost of $n$-fold replication itself.  Repair is off the common path and
absent with correct authors: a repairing party spends $O(n\lambda)$ on
requests per distinct authorized missing block, including ancestors discovered
while walking a named root, and, under the deduplicated typed-repair
response rule, may receive up to $O(n(\ell_b+\lambda))$ per block ---
$O(n^2(\ell_b+\lambda))$ across all parties in the worst case.  Counted
availability claims also authorize roots, but a view counts one echo-stage
response per sender, each naming at most $2n$ references, which bounds the
roots a Byzantine sender can authorize per view.
Repairing a lane prefix sums these per-block costs over its missing ancestors.
Coded dissemination would reduce this and is orthogonal.

The resolution plane of \Cref{sec:commit-rule} is accounted separately.
For one target and resolver view, WISH synchronization, suggestions and
proofs, fresh ECHOs and witnesses, the IT-HS phases, and DONE cost
$O(n^2)$ messages.  A full/core value contains two $O(n\lambda)$ manifests;
because witnesses and DONE attach that value, the proved by-value cost is
$O(n^3\lambda)$ bits, excluding lane-prefix repair.  Skip values are shorter.
Bounded value-body requests to counted phase authors add $O(n^2)$ messages and
$O(n^3\lambda)$ bits per resolver view in the worst case, preserving both
bounds.  Across $R$ entered resolver views of one target these per-view costs
add, giving the $O(Rn^2)$ messages and $O(Rn^3\lambda)$ bits stated in
\Cref{sec:commit-rule}.
Different targets run concurrently, so $k$ overlapping active targets
multiply both bounds by $k$.  The plane is absent from the logical dependency
chain of direct completion and sealing, not from network and processor use
while resolution is active.

\section{Resolution and Output: Full Specification and Proofs}
\label{app:protocol}

\paragraph{Resolver adapter.}
Resolution is composition state outside Direct AGB.  For each target data-plane
view $u$, the adapter may receive a target-local resolver decision: a bounded value
$R$ (below) whose submitted outcome $X$ is $\gfull(C,T)$, $\gcore(C)$, or
$\gskip$.  A non-skip value retains its exact backing payload $(C,T)$; the
core outcome omits $T$ from output but keeps it in the value's identity.  The
adapter authorizes the named lane roots, submits $X$ to the caller-owned
$\rstat{try-seal}(u,\cdot)$ arbiter, and fires $\rstat{seal}(u)\to X$ if no
compatible terminal outcome was already stored.
Direct AGB, fast-seal, grounded skip, and resolver submissions race only at
this arbiter.  The first terminal outcome is persistent and later compatible
submissions are idempotent.

A local terminal seal does not remove a party from resolver participation.
If such a party receives a peer resolver WISH for $u$, it allocates the
target-local instance, relays its own WISH, and may fresh-ECHO only a resolver
value whose submitted outcome equals that seal.  This rule is needed when a direct or skip
seal was selectively delivered: the parties still missing it may need every
correct participant.  A party without either target-local evidence or a local
terminal seal instead requires WISHes from $f+1$ distinct members before it
allocates an instance, so one Byzantine member cannot activate arbitrary
target phases.  Upon each direct WISH, a correct party also sends that requester its
current own WISH watermark and, once per requester and coordinate, every
applicable own
retained SUGGEST, PROOF, PROPOSE, phase statement, and witness through the
requested watermark plus one, together with its own retained DONE, if any.
Thus traffic sent before a late activation is not lost even though unknown-target
phase traffic is otherwise discarded.

\paragraph{Values and local candidates.}
For target $u$, write
\[
 R_u^{\mathrm{full}}(C,T),\qquad
 R_u^{\mathrm{core}}(C,T),\qquad
 R_u^{\mathrm{skip}}
\]
for the three bounded resolver values.  Their submitted outcomes are,
respectively, $\gfull(C,T)$, $\gcore(C)$, and $\gskip$.  The second value keeps
$T$ so that full and core values derived from the same target proposal have an
unambiguous identity.  The identifier $\hash{R}$ abbreviates a
domain-separated digest that binds the session, target, outcome kind, and
complete value.

Let $h_i=\max\{a_i+1,\omega_i^+\}$ be $p_i$'s local data-plane proposal
horizon.  Once $h_i\ge u+3$, each party repeatedly derives a monotone
candidate set from its first-hand target records; a party that has terminally
sealed $u$ does not start an instance from them but joins on a peer WISH.  Every candidate
requires $Q=n-f$ initial ready-stage responses.  In addition:
\begin{enumerate}[label=\textup{(C\arabic*)},leftmargin=2.7em,topsep=2pt,itemsep=2pt]
  \item $R_u^{\mathrm{full}}(C,T)$ requires $f+1$ matching grade-$1$ target
        ECHOs for $(C,T)$;
  \item $R_u^{\mathrm{core}}(C,T)$ requires $Q$ ready-stage responses whose
        first counted form is not grade $1$, and $f+1$ matching target ECHOs of
        either grade for $(C,T)$;
  \item $R_u^{\mathrm{skip}}$ requires $Q$ \noready{} responses.
\end{enumerate}
Only the first authenticated target ECHO and initial READY from each sender are
counted.  A READY-mix counted first remains historical non-grade-$1$ candidate
evidence after a permitted same-proposal homogeneous refinement.  The sender's
own provisional READY-mix nevertheless cannot pass fresh acceptance until it closes.
Because $f+1$ matching ECHOs for distinct payloads use disjoint first-response
authors, a local candidate set contains at most
$2\lfloor n/(f+1)\rfloor+1$ values.  Its canonical order is payload order,
full before core for the same payload, and skip last.

Fresh acceptance is deliberately stronger than candidacy.  Immediately before
sending a resolver ECHO for a zero-key proposal, $p_i$ checks:
\begin{enumerate}[label=\textup{(F\arabic*)},leftmargin=2.7em,topsep=2pt,itemsep=2pt]
  \item its own immutable echo-stage response exists and its own initial READY state is
        closed;
  \item full is rejected by an own grade-$0$ READY or a different ready
        payload; core is rejected by an own grade-$1$ READY or a different
        ready payload; skip requires the own response \noready{};
  \item full has locally available $C,T$ and passes the paired-tip ancestry
        check; core has locally available $C$;
  \item the value is compatible with any local terminal seal and active
        optimistic lock; and
  \item before a non-skip ECHO, the persistent stance changes atomically from
        $\mathsf{free}$ to $\mathsf{non\mbox{-}skip}$, or is already
        non-skip; a skip-voted stance rejects it.
\end{enumerate}
A fresh full ECHO carries origin bit $1$ only if its sender previously emitted
the exact matching grade-$1$ target ECHO.  A fresh core ECHO carries origin bit
$1$ after a matching target ECHO of either grade.  Skip carries no origin bit.
A fresh ECHO quorum for a non-skip value is usable only with at least $f+1$
origin-$1$ ECHOs.  Every resolver ECHO carries the exact proposal key $q$ and
is counted at coordinate $(u,r,q,x)$.  Thus $q=0$ ECHOs that executed these
fresh checks never combine with $q>0$ ECHOs for a carried \backed{} value, even
when both proposals name the same digest $x$.

\paragraph{Stable backing.}
After a usable quorum of resolver ECHOs for $(u,r,0,x)$, a correct party
broadcasts one witness $\langle\reswit,u,r,x,R\rangle$, attaching the bounded
value with $\hash{R}=x$.  A party that has not sent that witness relays it after
counting $f+1$ matching witnesses.  It sets
$\textsc{Backed}_i(u,x)$ after counting $Q$ matching witnesses.  A positive-key
view-change suggestion also attaches $R$; a receiver accepts the suggestion
only after checking its digest, well-formedness, IT-HS key proof, and
$\textsc{Backed}_i(u,x)$.

\begin{lemma}[Stable backing]\label{lem:stable-backing}
If a correct party sets $\textsc{Backed}(u,x)$, then a usable fresh resolver
ECHO quorum for one value $R$ with $\hash{R}=x$ occurred, at least one correct
party permanently holds $R$, and every correct party that participates in the
instance eventually sets the same \backed{} predicate and obtains $R$.
No correct witness can support two values at one $(u,r)$.
\end{lemma}
\begin{proof}
Consider the causally first correct witness for $(u,r,x)$.  It cannot be a
relay: among the $f+1$ witnesses required for relay, one is correct and would
precede it.  It therefore followed the usable fresh ECHO quorum and attached
the value it retained.  Correct witnesses for two different values at the same
$(u,r)$ would each descend from a fresh ECHO quorum; those quorums intersect in
a correct one-ECHO sender, a contradiction.  A correct party that counts $Q$ witnesses counts at
least $Q-f=n-2f\ge f+1$ correct senders.  Reliable delivery brings those
witnesses and their attached value to every correct participant, causing each
to relay.  If a witness preceded a party's activation, that party's WISH makes
each correct witness sender replay its own witness directly; the same
$f+1$ relay then starts at the late party.  The at least $n-f=Q$ correct parties
therefore supply a witness quorum everywhere.  Per-coordinate first-send and
per-requester replay rules give the final claim.
\end{proof}

\subsection{Resolved AGB Contract}\label{sec:resolver-contract}

Direct AGB intentionally leaves residual mixed-grade and silent instances
unsealed.  A resolver is a composition service with the following per-target
contract.  Every non-skip result also supplies the exact backing payload
$(C,T)$.
\begin{enumerate}[label=\textbf{RS\arabic*.},leftmargin=3.2em,topsep=3pt,itemsep=3pt]
  \item \textbf{Single choice.}  Resolver decisions at correct parties for
        target $u$ name one exact value and outcome.
  \item \textbf{Future-closed direct compatibility.}  A full result for
        $(C,T)$ excludes every grade-$0$ READY quorum and every READY quorum
        for a different payload.  A core result backed by $(C,T)$ excludes
        every grade-$1$ READY quorum and every READY quorum for a different
        payload.  A skip result excludes every READY quorum.
  \item \textbf{Output validity.}\newline A core result satisfies
        $\textsc{Witnessed}(C)$; a full result satisfies both
        $\textsc{Witnessed}(C)$ and $\textsc{Witnessed}(T)$.
\end{enumerate}
For termination we use one additional property:
\begin{enumerate}[label=\textbf{RL1.},leftmargin=3.2em,topsep=3pt,itemsep=3pt]
  \item \textbf{Per-target totality.}  For every target $u$, one outcome
        $X_u$ eventually reaches every correct party that lacks a terminal
        seal; every party already terminally sealed at $u$ holds that same
        compatible outcome.
\end{enumerate}
This service has agreement power; the definition does not claim that AGB
implements it.

\begin{theorem}[Resolved AGB composition]\label{thm:resolved-agb}
Direct AGB composed with any resolver satisfying RS1--RS3 has at most one
terminal outcome at each correct party, and no two correct parties terminally
seal different outcomes for a target $u$.  If completion exists, every
terminal outcome is non-skip, is backed by the completion's unique payload,
and contains its core.  Every non-skip outcome has the witness promised by
RS3 or by the direct path.  If RL1 also holds, every correct party eventually
terminally seals $u$.
\end{theorem}
\begin{proof}
The caller-owned arbiter gives local uniqueness.  Two direct results agree by
\Cref{thm:agb-safety}(b)--(c).  Two resolver results agree by RS1, and a direct
result agrees with a resolver result by RS2 and AGB value agreement.  A
completion has a READY quorum, so RS2 excludes skip and binds every non-skip
resolution to its payload.  Its core is fixed by
\Cref{thm:agb-safety}(b),(d), and the value-domain expansion contract makes
$\mathsf{Core}_D(C)$ a prefix of $\mathsf{Full}_D(C,T)$ for every prior-output
set $D$.  Direct outcomes obtain their witness from
\Cref{thm:agb-safety}(d); resolver outcomes obtain it from RS3.  RL1 supplies
the common outcome to every party that did not already seal it.
\end{proof}

\paragraph{Ordered output.}
The cursor processes data-plane views in increasing order.  At a locally
completed but open cursor view, it retrieves and verifies the core prefixes,
emits the canonical core expansion once, and stops before the tip.  A full
seal then retrieves and emits the canonical tip expansion before advancing; a
core seal advances without the tip; a skip seal emits nothing.  A resolver
decision received before local completion supplies the same manifests.  Late
compatible completion or sealing is idempotent.  Later AGB and resolver
instances continue while the cursor is stopped, but no payload from a later
view crosses that position.

Call a sequence of instances \emph{input-fair} when every admissible item from
a correct source eventually occurs in the core of some completed instance.

\begin{corollary}[Atomic broadcast from resolved AGB]
\label{cor:resolved-agb-log}
Run one Direct AGB instance per increasing data-plane view, compose each with a
resolver satisfying RS1--RS3 and RL1, and assume input fairness and fair
repair.  Under the cursor above, deterministic view-order expansion, and
first-occurrence deduplication, correct output sequences are prefix-comparable,
have one common limit, and eventually contain every finite prefix of that
limit.  Every output item is emitted at most once, satisfies external
validity, is witnessed as required by Direct AGB or RS3, and is retrieved
before local output.  Every admissible correct-source item is eventually
output.  Applying a deterministic transition function implements state-machine
replication.
\end{corollary}
\begin{proof}
By \Cref{thm:resolved-agb}, all correct parties use one outcome per view, and
a completed open core is a prefix of either compatible non-skip outcome.  The
cursor therefore exposes only prefixes of one deterministic view-order
expansion.  RL1 advances every cursor through every finite set of views; RS3,
Direct AGB availability, and fair repair provide the referenced prefixes
before output.  Deduplication gives at-most-once output, and input fairness
gives correct-input inclusion.  Equal input prefixes and a deterministic
transition function give equal state prefixes.
\end{proof}

\section{Concurrent Target-Local Resolver: Construction and Proofs}
\label{app:commit-rule}

Each unresolved target $u$ owns an independent sequence of resolver views
$r=1,2,\ldots$.  Resolver view~1 starts with the committee member immediately
after $\proposer(u)$; subsequent views rotate through the committee.  The
wire format binds each message's type and applicable target, resolver
view, and value digest; a resolver ECHO additionally binds the proposal key
$q$.  The authenticated connection binds its direct sender.
For each target and resolver view, parties count only the first statement from
each sender of each message type.  A quorum additionally requires equality of
the full typed coordinate; in particular, ECHOs with different proposal keys
never combine.  Parties count phase traffic at most one resolver view above
their own WISH watermark and act on a view's proposal, ECHO, KEY, and LOCK
quorums only while that view is current, a quorum counted before entry firing
on entry; witness generation and relay, \backed{}, WISH, and DONE relay and
decision are view-independent.

WISH uses the same high-watermark pattern as the data pacemaker, independently
for each target: $f+1$ WISHes amplify the own watermark and $Q$ enter all
missing views through that watermark.  A locally justified initial WISH is
retransmitted every $5\Delta$ until entry.  On entering view $r$, a party sends
the primary its KEY3/KEY2 suggestion, broadcasts its KEY1 proof, and arms both
a $5\Delta$ no-proposal timer and an $11\Delta$ full-view timer.  Either timer
raises the own WISH to $r+1$ when its condition still applies.
Each new WISH also elicits the responder's current own WISH and at most one
direct replay per coordinate of every applicable own retained view-scoped send
through $r+1$ and of any own DONE.  This instantiates IT-HS's send-upon-join
discipline under our non-FIFO channels.

The resolver then uses the information-theoretic HotStuff view-change core of
Abraham et al.~\cite{iths}.  KEY1 supplies lock-opening proofs, KEY2 supplies
\textsc{AcceptKey} proofs, and KEY3 is the candidate carried across views.  A
positive suggestion attaches its value body.  The primary waits for $Q$
accepted suggestions and re-proposes a maximum positive key, breaking ties by
digest.  If every accepted key is zero, it takes the next value of its
candidate set cyclically under a persistent local cursor and advances that
cursor; with no candidate it proposes nothing and the view times out.  The
cursor moves
on the party's own primary turns, rather than being indexed by the global
resolver view, so a common divisor between committee size and candidate count
cannot skip a candidate forever.

A zero-key proposal is fresh and invokes F1--F5 immediately before ECHO.  A
positive-key proposal must be \backed{}.  The ECHO binds that exact key, and only
ECHOs matching the receiver's stored $(u,r,q,x)$ proposal count together.  In
either case, the IT-HS lock accepts
the locked value or a proposal whose key and $f+1$ lock-opening proofs satisfy
\textsc{OpenLock}.  A syntactically valid designated-primary proposal is stored
even while these guards are false: its arrival stops the shorter no-proposal
timer, while the full-view timer remains armed and the proposal is rechecked
when proofs, backing, or target data change or its view becomes current.  After $Q$ matching ECHOs for one proposal
coordinate, a carried value advances directly; a fresh value first executes the witness step and advances only after it is
\backed{}.  $Q$ matching KEY1, KEY2, KEY3, and LOCK statements advance the
successive phases.  When a new KEY1 value differs from $x_1$, only then does
$\mathit{prev1}$ take the old KEY1 view; KEY2 updates
$\mathit{prev2}$ by the same conditional rule.  $Q$ LOCKs emit DONE.  DONE has
an $f+1$ relay and a $Q$ decision threshold.  The complete executable rules are
in \Cref{alg:direct-resolver,alg:direct-resolver-phases}.

Every digest quorum is inert until the party holds a well-formed value body
with the named digest.  A missing body is requested only from authors of the
counted statements; a reply is accepted only for a pending target/digest and
creates no protocol statement.  This repair rule handles non-FIFO body/phase
delivery and does not weaken the first-hand thresholds.

\begin{algorithm*}[tp]
\SetAlgoSkip{}
\caption{Target-local resolver activation and stable external validity at
  party $p_i$ for target $u$ ($Q=n-f$)}
\label{alg:direct-resolver}
\scriptsize\linespread{0.70}\selectfont
\Comment{For each $(u,r)$, count only the first authenticated statement from a
sender of each message type.  Quorums match the full typed coordinate: a
resolver ECHO uses $(u,r,q,x)$, including the exact proposal key $q$, and a
later phase uses $(u,r,x)$.  $\hash{R}$ abbreviates the domain-separated
digest of a bounded value
$R\in\{R_u^{\mathrm{full}}(C,T),R_u^{\mathrm{core}}(C,T),R_u^{\mathrm{skip}}\}$,
whose submitted outcomes are $\gfull(C,T)$, $\gcore(C)$, and $\gskip$; the
core value keeps $T$ only in its identity.}
\Glob{}{
  monotone candidate set $\mathcal C_i[u]$ and persistent next-candidate
    pointer $\mathit{next}_i[u]$;
  resolver WISH high-watermarks $\mathit{wish}_i[u,j]$ and own watermark
    $\mathit{ownWish}_i[u]$;
  $\textsc{Backed}_i(u,x)\gets\mathsf{false}$;
  persistent AGB stance
    $z_i(u)\in\{\mathsf{free},\mathsf{non\mbox{-}skip},
    \mathsf{skip\mbox{-}voted}\}$
}

\Upon{the data-plane proposal horizon reaches $u+3$, or target evidence changes thereafter}{
  add every locally justified value to $\mathcal C_i[u]$ in canonical order:
  every value needs $Q$ initial ready-stage responses for $u$;
  $R_u^{\mathrm{full}}(C,T)$ also needs $f+1$ matching grade-$1$ ECHOs;
  $R_u^{\mathrm{core}}(C,T)$ needs $Q$ ready-stage responses whose first
    counted form was not grade $1$ and $f+1$ matching ECHOs of either grade;
  $R_u^{\mathrm{skip}}$ needs $Q$ \noready{} responses;
  \If{$\mathcal C_i[u]\ne\emptyset$ and no instance, decision, or terminal
    seal for $u$ exists}{
    allocate the instance and invoke \textsc{RaiseWish}$(u,1)$
  }
}

\Proc{\textsc{RaiseWish}$(u,r)$}{
  \If{an instance exists, no decision exists, and
      $r>\mathit{ownWish}_i[u]$}{
    persist $\mathit{ownWish}_i[u]\gets r$;
    \textbf{broadcast} $\langle\rwish,u,r\rangle$;
    \lIf{no resolver view has been entered}{arm a $5\Delta$ retransmission timer}
  }
}

\Upon{a direct $\langle\rwish,u,r\rangle$ from $p_j$}{
  \If{the resolver already decided $(x,R)$}{send the retained
    $\langle\rstat{done},u,x,R\rangle$ to $p_j$ once}
  \Else{
    record $\mathit{wish}_i[u,j]\gets\max\{\mathit{wish}_i[u,j],r\}$
      \Comment*[r]{one watermark per sender, kept before allocation}
    \If{no instance exists}{
      \eIf{a retained local terminal seal exists}{
        allocate the instance and invoke \textsc{RaiseWish}$(u,1)$
      }{
        \lIf{$f+1$ senders have watermarks}{allocate the instance}
      }
    }
    \If{an instance exists}{
      \If{$f+1$ WISH high-watermarks support $w>\mathit{ownWish}_i[u]$}{
        invoke \textsc{RaiseWish}$(u,w)$
      }
      enter every missing view through the largest $w$ supported by $Q$ WISHes;
      send $p_j$ the current own WISH and, once per requester and coordinate,
        each applicable own SUGGEST, PROOF, PROPOSE, phase statement, and
        witness through $r+1$, plus the own DONE, if any
    }
  }
}

\Upon{the initial-WISH timer expires before any resolver view is entered}{
  rebroadcast the own WISH and re-arm the $5\Delta$ timer
}

\Fn{\textsc{FreshOK}$(u,R)$}{
  require a well-formed value for $u$, the own immutable echo-stage response,
    and a closed own initial READY state\;
  for $R_u^{\mathrm{full}}(C,T)$, reject an own grade-$0$ READY or a
    different own READY payload; require locally available $C,T$ and the
    tip-anchoring check\;
  for $R_u^{\mathrm{core}}(C,T)$, reject an own grade-$1$ READY or a
    different own READY payload; require locally available $C$\;
  for $R_u^{\mathrm{skip}}$, require the own response \noready{}\;
  require compatibility with any terminal seal and active optimistic lock\;
  before a non-skip acceptance, atomically claim
    $z_i(u):\mathsf{free}\to\mathsf{non\mbox{-}skip}$ or keep
    $\mathsf{non\mbox{-}skip}$; reject $\mathsf{skip\mbox{-}voted}$\;
  return origin bit $1$ for a full value after an exact matching grade-$1$
    target ECHO, and for a core value after a matching target ECHO of either
    grade; return bit $0$ for another non-skip value and no bit for skip
}

\Upon{$Q$ resolver ECHOs for the fresh coordinate $(u,r,0,x)$, a held value
  $R$ with $\hash{R}=x$ (requested from counted senders if missing), and,
  when $R$ is non-skip, at least $f+1$ origin-$1$ ECHOs among them; in any
  resolver view}{
  retain $R$ and \textbf{broadcast} $\langle\reswit,u,r,x,R\rangle$
}
\Upon{$f+1$ matching witnesses for $(u,r,x)$, before sending one}{
  retain the attached verifying $R$ and \textbf{broadcast} the own witness
}
\Upon{$Q$ matching witnesses for $(u,r,x)$}{
  latch $\textsc{Backed}_i(u,x)\gets\mathsf{true}$ and retry the current phase
}
\end{algorithm*}

\begin{algorithm*}[tp]
\SetAlgoSkip{}
\caption{One target's WISH/IT-HS views and decision at party $p_i$}
\label{alg:direct-resolver-phases}
\scriptsize\linespread{0.70}\selectfont
\Comment{SUGGEST, PROOF, PROPOSE, ECHO, KEY, and LOCK statements are counted
  through one view above $\mathit{ownWish}_i[u]$.  The ECHO send and the ECHO,
  KEY, and LOCK quorum rules act only while their view is current; a quorum
  already counted at entry fires on entry.  Witness generation and relay,
  \backed{}, WISH, and DONE relay and decision are view-independent.}
\Glob{}{
  $(\mathit{key1},x_1,\mathit{prev1})$,
  $(\mathit{key2},x_2,\mathit{prev2})$,
  $(\mathit{key3},x_3)$, and $(\mathit{lock},x_L)$, initially zero;
  one PROPOSE as primary, one sent ECHO, and one statement of each later
    phase per resolver view; at most one witness per $(u,r)$ and one DONE per
    target
}

\Upon{entering target-local resolver view $r$}{
  send the primary
    $\langle\rstat{suggest},u,r,\mathit{key3},x_3,
    \mathit{key2},x_2,\mathit{prev2},R_3\rangle$, where $R_3=\bot$ iff
    $\mathit{key3}=0$ and otherwise $\hash{R_3}=x_3$;
  \textbf{broadcast}
    $\langle\rstat{proof},u,r,\mathit{key1},x_1,\mathit{prev1}\rangle$;
  arm a $5\Delta$ no-proposal timer and an $11\Delta$ full-view timer
}

\Fn{\textsc{AcceptKey}$(q,x)$}{
  \KwRet{$q=0$, or $f+1$ counted suggestions with KEY2 history
    $(k_2,y_2,p_2)$ satisfy $p_2<k_2<r$ and either
    $q\le p_2$ or $(q\le k_2\land x=y_2)$}
}
\Fn{\textsc{OpenLock}()} {
  \KwRet{$f+1$ counted proofs with KEY1 history $(k_1,y_1,p_1)$ satisfy
    $p_1<k_1<r$ and either $\mathit{lock}\le p_1$ or
    $(\mathit{lock}\le k_1\land y_1\ne x_L)$}
}

\Comment{Accept a zero-key suggestion directly.  Accept a positive suggestion
  only when $0<q<r$, $\textnormal{\textsc{AcceptKey}}(q,x)$,
  $\backed{}_i(u,x)$, and its attached $R$ verifies with
  $\hash{R}=x$.}

\Upon{$p_i$ is the primary of $(u,r)$ and has $Q$ accepted suggestions}{
  choose one with maximum positive key view, breaking ties by digest;
  \uIf{that key is positive}{
    \textbf{broadcast} $\langle\rstat{propose},u,r,q,x,R\rangle$ with its
      attached $(q,x,R)$
  }
  \uElseIf{$\mathcal C_i[u]\ne\emptyset$}{
    take $R$ at position $\mathit{next}_i[u]\bmod|\mathcal C_i[u]|$ of
      canonical $\mathcal C_i[u]$ and advance the persistent pointer;
    \textbf{broadcast} $\langle\rstat{propose},u,r,0,\hash{R},R\rangle$
  }
  \Else{propose nothing in this view
    \Comment*[r]{the view times out}}
}

\Upon{the first designated-primary proposal $(u,r,q,x,R)$ with
  $q<r$ and $\hash{R}=x$}{
  store it, stopping only the no-proposal timer
}

\Upon{$r$ is current, a proposal $(u,r,q,x,R)$ is stored, no ECHO was sent in
  $r$, and either $\mathit{lock}=0$, $x=x_L$, or
  $(r>q\ge\mathit{lock}\land\textsc{OpenLock}())$; evaluated whenever proofs,
  backing, target data, or the current view change}{
  \eIf{$q=0$}{require \textsc{FreshOK}$(u,R)$ and attach its origin bit}{
    require $\textsc{Backed}_i(u,x)$; attach no origin bit
  }
  \textbf{broadcast} $\langle\rstat{echo},u,r,q,x,\mathit{origin}\rangle$
}

\Comment{Act on a digest quorum only while holding a verifying $R$; if it is
missing, request it from counted senders and accept only the pending digest.}

\Upon{$Q$ ECHOs matching the stored proposal coordinate $(u,r,q,x)$}{
  if $q=0$, wait until the witness rule of
    \Cref{alg:direct-resolver} makes $x$ \backed{};
  \If{$x\ne x_1$}{$\mathit{prev1}\gets\mathit{key1}$ and $x_1\gets x$}
  $\mathit{key1}\gets r$;
  \textbf{broadcast} $\langle\rstat{key1},u,r,x\rangle$
}
\Upon{$Q$ matching KEY1s for $(u,r,x)$}{
  \If{$x\ne x_2$}{$\mathit{prev2}\gets\mathit{key2}$ and $x_2\gets x$}
  $\mathit{key2}\gets r$;
  \textbf{broadcast} $\langle\rstat{key2},u,r,x\rangle$
}
\Upon{$Q$ matching KEY2s for $(u,r,x)$}{
  $(\mathit{key3},x_3)\gets(r,x)$;
  \textbf{broadcast} $\langle\rstat{key3},u,r,x\rangle$
}
\Upon{$Q$ matching KEY3s for $(u,r,x)$}{
  $(\mathit{lock},x_L)\gets(r,x)$;
  \textbf{broadcast} $\langle\rstat{lock},u,r,x\rangle$
}
\Upon{$Q$ matching LOCKs for $(u,r,x)$ and a verifying $R$ is held}{
  \textbf{broadcast} $\langle\rstat{done},u,x,R\rangle$
}
\Upon{$f+1$ matching DONEs for $(u,x)$ before sending one}{
  retain a verifying attached $R$ and \textbf{broadcast} the own DONE
}
\Upon{$Q$ matching DONEs for $(u,x)$}{
  decide $R$, retain the decision, and submit its outcome to
    $\rstat{try-seal}(u,\cdot)$
}

\Upon{the no-proposal timer expires with no stored designated-primary proposal, or the
  full-view timer expires before a decision}{
  invoke \textsc{RaiseWish}$(u,r+1)$
}
\end{algorithm*}

\begin{lemma}[Fresh external validity]\label{lem:direct-resolution}
If a usable fresh resolver ECHO quorum exists for target $u$, its value has the
following properties.
\begin{enumerate}[label=\textup{(\roman*)},leftmargin=2.2em,topsep=2pt,itemsep=2pt]
  \item A full value backed by $(C,T)$ excludes every grade-$0$ target READY
        quorum and every target READY quorum for a different payload; $C,T$
        are retrievable and author-backed.
  \item A core value backed by $(C,T)$ excludes every grade-$1$ target READY
        quorum and every target READY quorum for a different payload; $C$ is
        retrievable and author-backed.
  \item A skip value excludes every target READY quorum.
\end{enumerate}
These claims remain true when the value is carried into later resolver views.
\end{lemma}
\begin{proof}
Let $E$ be the fresh resolver ECHO quorum.  Any target READY quorum intersects
$E$ in at least $n-2f\ge f+1$ parties and therefore in a correct party.  That
party evaluated its own closed immutable target history in F1--F2.  For full,
it excludes own grade $0$, a different payload, and any later refinement; for
core it analogously excludes own grade $1$ and a different payload; for skip
it requires its own \noready{}.  The corresponding conflicting READY quorum
is therefore impossible.

Every correct ECHOer checked the required local prefixes, so at least one
correct permanent holder exists.  A usable non-skip quorum also contains
$f+1$ origin-$1$ ECHOs, at least one from a correct party.  In the full case
that party's matching grade-$1$ target ECHO checked $\textsc{AuthorOK}$ for
$C\cup T$ and tip ancestry; in the core case its matching target ECHO checked
$\textsc{AuthorOK}$ for $C$.  The author-witness and availability lemmas give
publication provenance and retrievability.  Carrying changes no target record
or ECHO fact, and \Cref{lem:stable-backing} binds the later value to this fresh
quorum.
\end{proof}

\begin{lemma}[Fast-seal compatibility]\label{lem:fast-seal}
For a proposal $B=(C,T)$, any two correct fast seals name
the same payload.  A correct $\fastseal(v)\to\gfull(C,T)$ is compatible with
every Direct AGB result and every resolver decision for $v$, and $C,T$ are
witnessed.
\end{lemma}
\begin{proof}
A fast observer counted the unique response of every author as a matching
response for $B$, a grade-$1$ ECHO for exactly $B$.  Hence all correct parties
emitted that response and installed the optimistic lock for $B$.  Two all-$n$ censuses cannot name
different proposals.  Quorum intersection excludes a grade-$0$ or
different-payload READY quorum, so direct results are compatible, and the
fast observer's checks witness $C,T$.

Suppose an incompatible fresh resolver ECHO quorum existed.  Its intersection
with the correct parties contains a correct sender.  That sender's optimistic
lock could accept the incompatible value only after $f+1$ nonmatching target
responses released it.  At least one such response is correct and also occurs
in the all-$n$ census, contradicting that every response there is a matching
response for $B$.  By \Cref{lem:stable-backing}, every carried or decided
resolver value descends from a fresh quorum, so it is compatible as well.
\end{proof}

\begin{lemma}[Grounded skip compatibility]\label{lem:skip-seal}
For every target $u$:
\begin{enumerate}[label=\textup{(\roman*)},leftmargin=2.2em,topsep=2pt,itemsep=2pt]
  \item if a correct party broadcasts $\skipvote(u)$, no correct party ever
        broadcasts a proposal-READY for $u$;
  \item a quorum of skip votes and a usable fresh non-skip resolver ECHO
        quorum cannot both exist; and
  \item every correct $\skipseal(u)\to\gskip$ is compatible with every direct,
        fast, or resolver submission.
\end{enumerate}
\end{lemma}
\begin{proof}
Let $b\le f$ be the actual number of Byzantine parties.  A correct skip voter
counted $Q$ target \echoskip{} responses, including at least $n-f-b$ correct
echo-skippers.  A correct proposal-READY sender would have a grounding ECHO
quorum with at least $n-f-b$ correct proposal echoers.  The two correct sets
are disjoint, yet
\[
  2(n-f-b)>n-b
\]
because the margin is $n-2f-b\ge f+1-b>0$.  Thus no correct proposal-READY,
completion, or direct result exists; an \echoskip{} also excludes the all-$n$
fast census.

A skip-vote quorum and a fresh non-skip resolver ECHO quorum intersect in at
least $n-2f\ge f+1$ parties, including a correct one.  Before the vote that
party persistently claimed skip-voted; before the ECHO it instead claimed
non-skip.  The claims are exclusive in either temporal order.  A carried
resolver value has a \backed{} fresh quorum by \Cref{lem:stable-backing}, so the
same argument applies to it and proves compatibility.
\end{proof}

\begin{lemma}[One value per resolver view]\label{lem:resolver-view-uniqueness}
Fix a target $u$ and resolver view $r$.  If two correct parties send any
statements among KEY1, KEY2, KEY3, and LOCK in view $r$, then the statements
name the same value.
\end{lemma}
\begin{proof}
A correct KEY1 sender counted $Q$ ECHOs for one exact proposal coordinate
$(u,r,q,x)$.  Two such ECHO quorums intersect in at least
$2Q-n=n-2f\ge f+1$ parties and hence in a correct party.  A correct party sends
at most one ECHO in a resolver view, so two correct KEY1 statements in that
view have the same value.  Induct on the remaining phases.  A correct sender
of the next phase counted $Q$ matching statements of the preceding phase; any
two such quorums again intersect in a correct one-shot sender, and every such
statement agrees with the already unique value.  This gives the claim through
LOCK.
\end{proof}

\begin{lemma}[Persistent lock safety]\label{lem:resolver-lock-safety}
Suppose a correct party sends DONE for $x$ from a quorum of LOCK statements in
resolver view $r$.  No correct party sends KEY1 for $y\ne x$ in a later view.
Consequently, all correct DONE statements, including relays, name one value,
and no two correct parties decide differently.
\end{lemma}
\begin{proof}
The LOCK quorum contains at least $Q-f=n-2f\ge f+1$ correct senders.  Fix a
set $L$ of $f+1$ of them; each member of $L$ held lock $(r,x)$ when it sent its
LOCK statement.

Assume for contradiction that $s>r$ is the smallest view in which a correct
party sends KEY1 for some $y\ne x$, and choose the causally first such KEY1 in
view $s$.  Its sender counted an ECHO quorum $E$ for $y$; every ECHO in $E$
therefore precedes that chosen KEY1.
Because $|E|=Q$ and $|L|=f+1$, $E\cap L$ is nonempty; let $p_h$ be a party in
the intersection.  Since $p_h$ sent its LOCK while $r$ was current and its
ECHO after entering $s>r$, before that ECHO it still has an $x$-valued lock
at some view at least $r$.  Indeed, changing that lock to a different
value in an intermediate view would require a correct KEY1 for that value in
the same view by \Cref{lem:resolver-view-uniqueness}, contradicting the
minimality of $s$.  It cannot first change the lock within view $s$ before
this ECHO either: a correct LOCK for the new value causally follows a correct
KEY1 for that value, which would precede the chosen KEY1.

Thus $p_h$ can ECHO $y$ only by counting $f+1$ proofs that satisfy
\textsc{OpenLock}; one proof sender $p_j$ is correct.  Write its proof as
$(k_1,z,p_1)$, where $p_1<k_1<s$.  If $p_h$'s lock view is at most $k_1$ and
$z\ne x$, then $p_j$ sent a non-$x$ KEY1 at view $k_1\ge r$.  Otherwise the
accepted proof has the lock view at most $p_1$.  By the conditional previous-
key update, $p_j$ sent KEY1 at both $p_1$ and $k_1$ with different values, so
one of them is non-$x$.  In either case a correct non-$x$ KEY1 occurs in a
view between $r$ and $s$: equality with $r$ contradicts
\Cref{lem:resolver-view-uniqueness}, and a larger view contradicts the
minimality of $s$.

It follows that phase-originated correct DONEs in different views agree: take
the smaller view and apply the first part, while equal views use
\Cref{lem:resolver-view-uniqueness}.  A causally first correct DONE for any
other value cannot be a relay, because its $f+1$ triggering DONEs include an
earlier correct sender.  Hence relays introduce no new value.  Finally, every
$Q$-DONE decision quorum contains a correct sender, so all correct decisions
have that unique value.
\end{proof}

\begin{lemma}[Persistent key support]\label{lem:resolver-key-support}
If a correct party sets KEY3 to $(k,x)$ with $k>0$, then at least $f+1$
correct parties' suggestions in every later view support
\textsc{AcceptKey}$(k,x)$.  Once stable backing and the attached values have
reached the correct participants, a correct primary eventually accepts every
correct suggestion.  Moreover, if some correct party has a lock at view
$\ell>0$, every set of $Q$ accepted suggestions in a later view contains a
positive KEY3 view at least $\ell$.
\end{lemma}
\begin{proof}
Acting only in the current view keeps a party's KEY1, KEY2, KEY3, and lock
views nondecreasing.  Setting KEY3 to $(k,x)$
requires $Q$ KEY2 statements for $(k,x)$, including
at least $f+1$ from correct parties.  Immediately afterward, each such sender
has either $k_2\ge k$ with $x_2=x$, or already has $p_2\ge k$.  A later KEY2
with the same value preserves the first condition.  At the first change of
value, the conditional update sets $p_2$ to the previous $k_2\ge k$, after
which the second condition is permanent.  Their future suggestions therefore
continue to support \textsc{AcceptKey}$(k,x)$.

A zero-key correct suggestion is accepted directly.  A positive correct
suggestion carries its retained value; the $f+1$ correct supporters above
eventually reach the primary, and stable backing eventually holds there, so
the positive suggestion is accepted as well.

Finally, a correct lock at $\ell$ followed a quorum of KEY3 statements in
view $\ell$.  At least $n-2f$ of those senders are correct and thereafter
report KEY3 views at least $\ell$.  Any set of $Q=n-f$ senders intersects this
set because its complement has at most $2f<Q$.  Hence a $Q$-suggestion set
contains a positive key at least $\ell$.
\end{proof}

\begin{lemma}[Persistent lock-opening support]\label{lem:resolver-open-support}
In resolver view $r$, let a correct primary select a maximum positive accepted
suggestion $(q,x)$.  For every correct lock $(\ell,z)$, $q\ge\ell$; if
$x\ne z$, at least $f+1$ correct proofs eventually satisfy
\textsc{OpenLock} for that lock.
\end{lemma}
\begin{proof}
The inequality $q\ge\ell$ follows from the final part of
\Cref{lem:resolver-key-support}.  Because $(q,x)$ was accepted, the primary
counted $f+1$ supporting KEY2-history tuples.  At least one belongs to a
correct party $p_j$.  Its tuple satisfies either $p_2\ge q$, or
$k_2\ge q$ with $x_2=x$.  Since $q\ge\ell$ and $x\ne z$, $p_j$ therefore has
either $p_2\ge\ell$, or $k_2\ge\ell$ with $x_2\ne z$.

In the second case, when $p_j$ set $(k_2,x_2)$ it counted a KEY1 quorum,
including $f+1$ correct KEY1 senders at view $k_2\ge\ell$ with value different
from $z$.  In the first case, the conditional previous-key update records two
KEY2 views $p_2<k_2$ with different values.  At least one of those values
differs from $z$; the corresponding KEY1 quorum again supplies $f+1$ correct
senders at a view at least $\ell$ with a non-$z$ value.

For each of these correct KEY1 senders, the condition needed by
\textsc{OpenLock} persists.  Repeating the same value leaves
$k_1\ge\ell$ and $x_1\ne z$; changing value sets $p_1$ to the old
$k_1\ge\ell$, after which $p_1\ge\ell$ remains true.  These parties broadcast
their proof tuples on entering $r$, and send-upon-join replays them if needed.
Thus every correct party with lock $(\ell,z)$ eventually counts the required
$f+1$ proofs.
\end{proof}

\begin{lemma}[One timely resolver view]\label{lem:resolver-timing}
Suppose the first correct entry to target-local view $r$ occurs at
$t\ge\mathrm{GST}$, all correct parties enter by $t+2\delta$, the primary is
correct, and stable backing and attached values have reached all correct
participants.  Suppose either that the primary selects a positive key, or
that all accepted keys are zero and its fresh candidate is already acceptable
with its required data at every correct party, with at least $f+1$ correct
parties attaching origin bit $1$ when that candidate is non-skip.  Then all
correct parties receive the proposal before their $5\Delta$ no-proposal
deadlines.  A fresh proposal decides by $t+11\delta$; a carried proposal
decides by $t+10\delta$.
\end{lemma}
\begin{proof}
Suggestions and proofs reach the correct primary by $t+3\delta$.
\Cref{lem:resolver-key-support} makes the $Q$ correct suggestions acceptable.
If their maximum key is positive, it is at least every correct lock and
\Cref{lem:resolver-open-support} makes the proposal acceptable at every
correct party.  If their maximum is zero, the final part of
\Cref{lem:resolver-key-support} implies that no correct lock is positive, and
the premise supplies fresh acceptance and, for a non-skip value, the required
origin support.  The proposal reaches all correct
parties by $t+4\delta$, and their ECHOs arrive everywhere by $t+5\delta$.

For a fresh proposal, the resulting witnesses arrive by $t+6\delta$ and make
the value \backed{}.  KEY1, KEY2, KEY3, LOCK, and DONE then arrive by,
respectively, $t+7\delta,t+8\delta,t+9\delta,t+10\delta,t+11\delta$.  A
carried value is already \backed{} and omits the witness delay, so the same five
arrivals end at $t+10\delta$.  At least $n-f=Q$ correct parties participate.
The $5\Delta$ proposal deadline therefore has one delay of slack even when
$\delta=\Delta$, and the fresh full-view deadline is met at equality only in
that limiting case.  A timeout raises WISH without abandoning the current
view, so a deadline tie does not suppress the proposal or decision.
\end{proof}

\begin{theorem}[Target-local resolver properties]\label{thm:direct-resolver}
For each target $u$:
\begin{enumerate}[label=\textup{(D\arabic*)},leftmargin=2.7em,topsep=2pt,itemsep=2pt]
  \item no two correct parties decide different resolver values;
  \item every decided non-skip value is author-backed and satisfies the
        future-closed compatibility of \Cref{lem:direct-resolution};
  \item if one correct party decides $R$, every correct party eventually
        either decides $R$ or already holds its compatible terminal seal; and
  \item under partial synchrony, every target not already terminally sealed at
        all correct parties eventually obtains one common resolver decision or
        the compatible terminal seal at every correct party.
\end{enumerate}
\end{theorem}
\begin{proof}
For D1, \Cref{lem:resolver-view-uniqueness,lem:resolver-lock-safety} prove
agreement directly for the stated message and update rules.  The target
coordinate separates concurrent instances, and the value digest binds every
nonzero value to the session.  Every value entering KEY1 is \backed{}: a fresh proposal
waits for its witness quorum, and a carried proposal checks the stable
predicate before ECHO.  Thus no phase advances a value outside the resolver's
external-validity domain.
D2 follows from \Cref{lem:stable-backing,lem:direct-resolution}.

We first establish a candidate-convergence fact used by both D3 and D4.  By
the Direct AGB response deadlines, every correct
party eventually broadcasts its target echo-stage and closed initial
ready-stage response.  Let $b\le f$ be the actual Byzantine count.  For each
payload $P$, let $G_P$ be the correct parties whose target response is a
grade-$1$ ECHO for $P$.  Call $P$ large when
$|G_P|\ge n-f-b$; at most one payload is large because
$2(n-f-b)>n-b$.

If a large $P=(C,T)$ exists, the universally acceptable value is
$R_u^{\mathrm{full}}(C,T)$.  A correct grade-$0$ or different-payload READY
would require a grounding ECHO quorum whose at least $n-f-b$ correct members
are disjoint from $G_P$, contradicting the same inequality.  Active locks for
other payloads eventually count at least $f+1$ correct nonmatching responses
and release.  A correct skip voter would likewise require a disjoint set of at
least $n-f-b$ correct echo-skippers, so no correct stance is skip-voted.  The
at least $n-f-b\ge f+1$ correct members of $G_P$ supply the candidate evidence,
origin bits, data, and provenance.

Otherwise every nonempty $G_P$ is small and every optimistic lock eventually
releases after the other correct target responses arrive.  No correct
grade-$1$ READY exists: its grounding quorum would contain at least
$n-f-b$ correct grade-$1$ ECHOers and make that payload large.  The closed
correct ready histories therefore fall into two exhaustive cases: some
correct response is a proposal-READY, or all correct responses are
\noready{}.  AGB value agreement makes all correct proposal-READY responses
name the same payload.  In the first case, its grounding quorum provides $f+1$ correct matching
ECHOs, all correct READYs are non-grade-$1$, and the universal value is core.
By \Cref{lem:skip-seal}(i), that correct proposal-READY also excludes every
correct skip vote.
In the second, the at least $Q$ correct \noready{} responses make skip universal.
Fair
repair eventually supplies every required prefix, so in every case one fixed
value $R^*$ is in every unresolved correct party's candidate set and passes
F1--F5.  By \Cref{lem:fast-seal,lem:skip-seal,lem:direct-resolution}, any
correct party already sealed holds the same submitted outcome and can help
vote for $R^*$.

For D3, every unresolved correct party eventually derives the candidate
$R^*$ above and therefore starts WISH even if it received none of the earlier
resolver traffic.  $Q$ DONEs at one correct party
contain at least $n-2f\ge f+1$ correct senders.  Their reliable broadcasts
cause every already active correct party to relay DONE.  A party activated
later broadcasts a WISH; each of those correct senders replies with its own
retained DONE, so the $f+1$ relay starts there as well.  The at least $Q$ correct
relays then decide everywhere.  A party with a local terminal seal already
satisfies the alternative in D3; on a peer WISH it can also activate and vote
only for the value whose submitted outcome equals that seal.

It remains to prove D4.  If a correct resolver decision already exists, D3
gives the claim, so assume none exists.
The candidate sets are finite and eventually stabilize.  The first local
candidate broadcasts WISH.  Every correct party that already holds the
compatible terminal seal responds to that first peer WISH by broadcasting its
own initial WISH; every unresolved correct party eventually derives $R^*$ and
broadcasts its own.  Thus the $Q$ correct parties supply a view-$1$ WISH quorum
without Byzantine participation.  Reliable delivery, WISH amplification,
timers, and round-robin leaders then make correct parties enter arbitrarily
high common resolver views.  The send-upon-join replies
return both a responder's own watermark and any earlier phase traffic, so
pre-activation delivery alone cannot leave a correct party below an entry or
phase quorum.  Applied per target, the
high-watermark argument of \Cref{lem:pacemaker}(b) gives an eventual
$2\delta$ correct-entry spread.  Explicitly, the $Q$ WISHes seen by the first
correct entrant contain at least $f+1$ correct senders.  Their WISHes reach
every correct party within one delay and cause amplification; the resulting
$Q$ correct WISHes arrive within one more delay.

Choose a sufficiently late common view after candidate sets, \backed{} facts,
and retained value bodies have stabilized.  By
\Cref{lem:resolver-key-support}, a correct primary accepts $Q$ suggestions.  If
one is positive, it selects a maximum positive key; the same lemma and
\Cref{lem:resolver-open-support} make that carried proposal acceptable through
every correct lock.  If all accepted keys are zero, no correct lock is
positive.  Each unresolved correct primary advances its persistent cursor on
its own turns through the finite stabilized candidate set.  At least one
correct party is unresolved in the case under consideration, and $R^*$
belongs to every such party's set, so one of its primary turns selects $R^*$.
If an intervening fresh value instead reaches KEY3, it is
\backed{} and the positive-key case applies thereafter.  Static round-robin
rotation supplies infinitely many correct-primary turns.  Therefore a view
eventually satisfies \Cref{lem:resolver-timing} and decides, proving D4.
\end{proof}

\begin{lemma}[Vantage resolver contract]\label{lem:vantage-resolver}
The target-local resolver satisfies RS1--RS3 and RL1 of
\Cref{sec:resolver-contract}.
\end{lemma}
\begin{proof}
RS1 and RL1 are D1, D3, and D4 of \Cref{thm:direct-resolver}.  RS2 and RS3
follow from D2 and \Cref{lem:direct-resolution}.  Fast and skip compatibility
are given separately by \Cref{lem:fast-seal,lem:skip-seal}.
\end{proof}

\begin{theorem}[Prefix and seal agreement]\label{thm:grade-agreement}
Each correct party seals a view at most once, and no two correct parties seal
different outcomes.  If completion exists, every terminal outcome is
$\gfull(C,T)$ or $\gcore(C)$ for its unique payload and contains its fixed
core.  Before sealing, a correct completion may expose $\gopen(C,T)$; for
every prior-output set $D$, $\mathsf{Core}_D(C)$ is a prefix of either
compatible non-skip expansion.
\end{theorem}
\begin{proof}
Apply \Cref{thm:resolved-agb} using \Cref{lem:vantage-resolver}.
\Cref{lem:fast-seal,lem:skip-seal} make the caller-side shortcuts compatible
with direct and resolver submissions.  The persistent arbiter gives local
uniqueness, and \Cref{thm:agb-safety} gives the completed-open core prefix.
\end{proof}

\begin{theorem}[Sealing liveness]\label{thm:resolution-live}
With the data pacemaker and target-local resolver above, every correct party
eventually seals every data-plane view.  A completed open view is eventually
sealed as full or core everywhere.
\end{theorem}
\begin{proof}
Fix a target $u$.  If compatible local direct, fast, or grounded skip routes
seal it everywhere, the claim holds.  Otherwise
\Cref{thm:direct-resolver}(D4) supplies a common resolver value to every party
that lacks the compatible terminal seal, and the adapter submits it to the
arbiter.  \Cref{thm:grade-agreement} makes all outcomes equal.  A completed
view has a READY quorum; RS2 excludes skip, so its result is full or core.
\end{proof}

\begin{corollary}[Crash-only silent-view bound]\label{cor:crash-skip}
Let $s\ge\mathrm{GST}$ be the first correct entry to data-plane view $u$, and
suppose every correct party enters by $s+2\delta$.  Suppose the proposer and
every other faulty party crashed before sending any view-$u$ or later protocol
message, so no view-$u$ proposal or proposal ECHO exists.  Then every correct
party submits $\skipseal(u)\to\gskip$ by
\[
  s+4\Delta+3\delta.
\]
If all correct parties enter simultaneously at time $e$, the bound is
$e+4\Delta+\delta$.
\end{corollary}
\begin{proof}
Every correct party emits \echoskip{} by its absolute ECHO deadline, at latest
$s+2\delta+3\Delta$.  These at least $Q$ correct responses arrive everywhere by
$s+3\delta+3\Delta$.  Every correct party emits \noready{} by at latest
$s+2\delta+4\Delta$, no earlier than the former time because
$\delta\le\Delta$.  By then it therefore holds both its own \noready{} and Q
\echoskip{} responses.

No correct party can justify a non-skip resolver value because no target
proposal ECHO exists, and the faulty parties send no resolver message.  Every
stance remains free until the vote; an independently learned compatible skip
does not suppress it.  All correct parties broadcast $\skipvote(u)$ by
$s+2\delta+4\Delta$, and those votes arrive everywhere one actual delay later.
Aligned entry removes the $2\delta$ entry spread.  Each later crashed view has
its own deadlines; the corollary does not collapse a burst into one timeout.
\end{proof}

\section{End-to-End Correctness Proof}\label{app:correctness}

\begin{theorem}[End-to-end correctness]\label{thm:correctness}
Every correct party eventually seals every data-plane view with exactly one
common terminal outcome.  At every time, the output sequences of any two
correct parties are prefix-comparable; they have the same limit sequence, and
every finite prefix is eventually output by every correct party.  Every output
block occurs at most once, satisfies $\textsc{BlockOK}$, was published with
its valid lane prefix by its encoded author to at least one correct party, has a
permanent correct holder, and is obtained by each correct party before that
party locally outputs it.
Moreover, every valid data block
published by a correct author is eventually output.
\end{theorem}
\begin{proof}
\Cref{thm:grade-agreement,thm:resolution-live} give exactly one eventual common
seal per view.  For the ordered-output claims, Direct AGB supplies the local
properties, \Cref{lem:vantage-resolver} supplies RS1--RS3,
\Cref{thm:resolution-live} supplies per-target totality, and fair repair is
\Cref{lem:availability-witness}(ii).  The only extension beyond
\Cref{cor:resolved-agb-log} is the pair of caller-side local routes:
\Cref{lem:fast-seal} proves that a fast seal names the same witnessed full
outcome as every Direct AGB or resolver submission, while
\Cref{lem:skip-seal} proves compatibility of a skip seal.  Hence the
corollary's cursor argument applies unchanged.  It
remains to establish input fairness.

Let a correct author $a$ publish a valid data block $b$ at height $k_b$.
Every correct party eventually receives the full lane prefix through $b$ in
the author's \rstat{publish} messages.  Fix a correct party $p_i$.  By
\Cref{lem:pacemaker}(a), R1, round-robin assignment, and non-Zeno execution,
$p_i$ has infinitely many data-plane proposer turns; resolution runs in
separate target-local instances and never consumes a data-plane proposal.  If
$a$ is quarantined, its exponential probe schedule is also infinite: among the
finite author set, earliest-due selection is fair because every selected
author's next due opportunity moves strictly forward.  These claims remain
true when the stored witness was fabricated, since eligibility tests the
proposer's real directly published lane state.

Consider such turns after the pacemaker cutoff.  If $p_i$ already has an
exact-$\quorum$ core prefix covering $b$, a subsequent proposal carries it in
$C$.  Otherwise, if a stable confirmation target exists above the current
core, the next ordinary inclusion or fair due probe for $a$ names the
least-height such target, not a fresher head.  It is a fixed coordinate of the
correct author's unforked lane, so every correct party eventually has
$\textsc{DirectPub}$ for it.  A sufficiently late proposal naming it therefore
receives exact-position claims from all correct parties and makes it
exact-$\quorum$-available.  Either that proposal full-seals the prefix, or a
subsequent proposal carries the new candidate in $C$ and advances the core
beyond its previous coordinate.

If no stable target exists, take an ordinary inclusion or due probe after
every correct party has received $b$.  Unless the core already covers $b$, R1
names a real directly published tip covering $b$.  Each correct ECHO claim for
that tip acknowledges a direct prefix covering $b$ --- exactly at the named
position or through its verified integer-derived prefix.  Consequently every
coordinate above the current core through $b$ becomes generally
$\quorum$-available.  If the proposal full-seals, it already outputs $b$;
otherwise the least such coordinate becomes the stable confirmation target of
the preceding paragraph.  Each confirmation either outputs the target or
makes it the exact-$\quorum$ candidate carried by a subsequent core proposal;
that proposal advances the core.  Induction over the finite coordinates
through $k_b$ therefore eventually covers $b$.  Continuous publication cannot move the least
unconfirmed target, and no periodic availability message is used.

Choose a final full-sealing confirmation, or the subsequent proposer turn
that carries the confirmed prefix in $C$, after the cutoff.  In the latter
case \Cref{thm:agb-live} makes the proposal complete everywhere even if some
unrelated tip spoils its grades; completion makes its core irrevocable and the
resolver eventually seals the view.  The output cursor therefore emits the
confirmed prefix.  This proves input fairness.  All
remaining premises of the cursor argument in \Cref{cor:resolved-agb-log} now
hold, and it gives the remaining claims.
\end{proof}

\begin{corollary}[State-machine replication]\label{cor:smr}
Fix an initial state and a deterministic transition function defined on
$\textsc{BlockOK}$ blocks.  Applying output blocks in order gives correct
parties identical states after every common output prefix, and every finite
state prefix is eventually reached by every correct party.
\end{corollary}

\section{Implementation and Evaluation Details}\label{app:evaluation}

\subsection{Implementation and measured encoding}\label{app:implementation}

Our Rust implementation uses one deterministic protocol task and separate
networking and storage.  It performs no public-key or signature operation.
Validators retain 32-byte identifiers internally but use canonical one-byte
wire indices; the transport binds every connection to a committee identity and
authenticates each frame with a pairwise MAC, on by default
(\Cref{sec:implementation}).

The same binary runs both Autobahn modes and two Simple-IT variants on a shared
TCP, batching, worker, storage, client, and metric stack
~\cite{autobahn,simpleit,narwhal}.  Simple-IT follows its paper, adds a safe
notarization bound, and fetches missing proposals; Bluestreak and Sailfish++
run on the separate artifact~\cite{bluestreak,starfish}.  Timer-calibrated in-tree runs set
$\Delta=200$\,ms: the largest injected RTT is 309\,ms, hence the largest
one-way delay is 154.5\,ms; we round this to 160\,ms and reserve 40\,ms for
processing and scheduling.  Vantage uses
$\theta_E=3\Delta=600$\,ms and $\theta_R=4\Delta=800$\,ms.  The fault-free
Q1 and Q2 runs do not consume a fallback decision.  The target-local resolver
used in Q4 and Q5 has a $5\Delta=1$\,s no-proposal timer and an
$11\Delta=2.2$\,s full-view timer.
Both optimistic all-to-all and
seamless Autobahn use the paper's $10\Delta=2$\,s consensus timeout;
all-to-all routing and tip selection do not alter it.  Simple-IT uses
$8\Delta=1.6$\,s with Opt-RBC and $5\Delta=1$\,s with Bracha RBC.  Autobahn's
separate 500\,ms unanimous-fast-path wait is an implementation tuning, not a
view timeout.  The upstream artifact implements linear Autobahn; we add its
paper's all-to-all exchange for the optimistic variant.  Other shared settings
are a 100\,ms header delay, 64\,KiB/5\,ms network flushing, and 20\,ms worker
batching; Vantage and Simple-IT use BLAKE3 and round-robin leaders.

\subsection{Measurement methodology}\label{app:methodology}

\paragraph{Network and workload.}
Validators use private addresses in one availability zone.  Netem assigns
validator $i$ row $i\bmod 10$ of \Cref{tab:rtt-matrix}, placing half the RTT
on each direction without jitter~\cite{starfish}.  Variants at one committee
size reuse instances and rules; sizes use new instances, so cross-size trends
include AWS placement variation.  Each validator offers $R/n$ open-loop load
in 50\,ms bursts.  Transactions are timestamped, random, effectively
incompressible 512-byte values~\cite{starfish,bluestreak}.

Workers store and disseminate batches before passing digests to the primary.
Vantage puts them in unsigned lane blocks, broadcasts headers every 100\,ms,
and fetches missing batches; the interval approximates the uncertified-DAG
round duration.

\paragraph{Metrics.}
After all validators advance, runs warm up for 30 seconds and measure for 120.
Latency spans submission through local sequencing and byte retrieval
~\cite{starfish,simpleit}; throughput is sequenced transactions per window.
We report the median validator's p50/p99 latency, total primary-plus-worker
CPU, and outbound framed bytes per sequenced byte.  Cloud points run as up to
three campaign repetitions on fresh fleets; every reported metric is the median
across the available repetitions of its per-run value.  The $n{=}100$
committee-scaling cell joins from the throughput campaign's 100~tx/s point,
which every repetition measures.

\subsubsection{Baseline fleet}\label{app:baseline-fleet}
We measured both artifact baselines on \texttt{c5d.2xlarge} and
\texttt{m5d.2xlarge} (equal 8 vCPUs; 16 vs 32 GiB RAM; 200 vs 300 GB local
NVMe) and report each from its better fleet.  Bluestreak runs on
\texttt{m5d.2xlarge}: on \texttt{c5d.2xlarge} its throughput collapses well
before the highest accepted point it reaches on \texttt{m5d.2xlarge}, with six
validators killed for exhausting memory.  The limit is memory rather than CPU.
Sailfish++ stays on the shared \texttt{c5d.2xlarge} fleet.  Bluestreak therefore runs
on a stronger instance than \sysname{}; no baseline runs on a weaker one.

\begin{table}[!htbp]
  \centering
  \scriptsize
  \setlength{\tabcolsep}{3pt}
  \caption{Injected round-trip latency matrix (ms).  Row and column order is
  the committee-index assignment order.}
  \label{tab:rtt-matrix}
  \begin{tabular}{lrrrrrrrrrr}
    \toprule
     & USE1 & CAC1 & EUS2 & SAE1 & APSE2 & USW1 & EUW1 & EUN1 & APS1 & APNE1 \\
    \midrule
    USE1  &   1 &  14 & 104 & 112 & 198 &  65 &  68 & 110 & 201 & 146 \\
    CAC1  &  14 &   1 & 106 & 122 & 196 &  78 &  67 & 103 & 189 & 142 \\
    EUS2  & 104 & 106 &   1 & 215 & 281 & 163 &  29 &  50 & 143 & 238 \\
    SAE1  & 112 & 122 & 215 &   1 & 309 & 175 & 176 & 220 & 299 & 254 \\
    APSE2 & 198 & 196 & 281 & 309 &   1 & 137 & 254 & 268 & 150 & 101 \\
    USW1  &  65 &  78 & 163 & 175 & 137 &   1 & 127 & 172 & 226 & 108 \\
    EUW1  &  68 &  67 &  29 & 176 & 254 & 127 &   1 &  38 & 125 & 199 \\
    EUN1  & 110 & 103 &  50 & 220 & 268 & 172 &  38 &   1 & 148 & 245 \\
    APS1  & 201 & 189 & 143 & 299 & 150 & 226 & 125 & 148 &   1 & 140 \\
    APNE1 & 146 & 142 & 238 & 254 & 101 & 108 & 199 & 245 & 140 &   1 \\
    \bottomrule
  \end{tabular}
\end{table}

\begin{table}[!htbp]
  \centering
  \scriptsize
  \setlength{\tabcolsep}{3pt}
  \caption{End-to-end latency as p50/p99 in milliseconds, for the median
  validator, taken as the median across the available campaign repetitions.
  These are the cells the corresponding figures plot, with the tail added.
  Columns are \sysname{}, optimistic all-to-all Autobahn (A2A), seamless
  Autobahn (Seamless), Simple-IT with Opt-RBC (IT-Opt) and with Bracha RBC
  (IT-Bracha), Bluestreak, and Sailfish++.  A dash
  marks a size or offered-load point the variant does not reach, or a variant the relay
  stress does not cover.  Overload endpoints are excluded, since a variant past
  overload reports a cumulative percentile over pre-collapse traffic.}
  \label{tab:latency-percentiles}
\begin{tabular}{@{}lrrrrrrr@{}}
\toprule
 & \sysname{} & A2A & Seamless & IT-Opt & IT-Bracha & Bluestreak & Sailfish++ \\
\midrule
\multicolumn{8}{@{}l}{\emph{Committee scaling at 100 tx/s} (\Cref{fig:committee-scaling})} \\
$n{=}10$ & 446/599 & 474/677 & 706/992 & 642/802 & 774/970 & 463/627 & 682/1957 \\
$n{=}20$ & 446/601 & 484/677 & 684/950 & 666/833 & 782/985 & 466/628 & 704/2268 \\
$n{=}50$ & 442/603 & 490/694 & 692/969 & 669/864 & 777/985 & 473/633 & 773/1518 \\
$n{=}100$ & 452/613 & 560/776 & 795/1097 & 699/868 & 778/985 & 480/640 & 786/2091 \\
\midrule
\multicolumn{8}{@{}l}{\emph{Throughput ladder at $n=100$, accepted rungs} (\Cref{fig:throughput-sweep})} \\
100 tx/s & 452/613 & 560/776 & 795/1097 & 699/868 & 778/985 & 480/640 & 786/2091 \\
10k & 456/616 & 566/788 & 800/1112 & 692/878 & 780/983 & 481/641 & 788/1992 \\
150k & 475/640 & 604/832 & 838/1150 & 704/897 & 794/1001 & 512/786 & 840/2910 \\
200k & 483/646 & 616/850 & 852/1163 & 708/904 & 800/1010 & 522/870 & --- \\
225k & 487/653 & 629/860 & 864/1178 & 712/908 & 804/1018 & 542/1283 & --- \\
250k & 492/654 & 636/868 & 877/1198 & 718/916 & 810/1028 & 541/1465 & --- \\
275k & 509/698 & 1130/3306 & 1442/3553 & 742/992 & 825/1064 & --- & --- \\
\midrule
\multicolumn{8}{@{}l}{\emph{Leader-relay stress at $n=20$, $f=6$} (\Cref{fig:leader-relay})} \\
100 tx/s & 450/1395 & 809/1154 & --- & 674/861 & --- & --- & --- \\
1k & 452/1354 & 860/1144 & --- & 668/863 & --- & --- & --- \\
10k & 454/1392 & 1706/2592 & --- & 672/868 & --- & --- & --- \\
20k & 460/1426 & 1836/2694 & --- & 679/872 & --- & --- & --- \\
100k & 458/1317 & 1886/3562 & --- & 674/867 & --- & --- & --- \\
200k & 462/1226 & 2336/5660 & --- & 672/859 & --- & --- & --- \\
\bottomrule
\end{tabular}

\end{table}

\subsection{Representation optimizations}\label{sec:optimizations}

The protocol places lane-frontier acknowledgments in the authenticated
envelope of each proposal ECHO, outside its immutable Direct AGB identity.
Since $\hash{B_v}$ binds an ordered reference vector, a bitmap marks exact
directly published positions.  A sender behind a named tip may instead give a
sparse pair containing the position and positive integer distance to its
greatest directly published ancestor.  The pair carries no digest: a receiver
credits it only after the named tip's verified predecessor chain derives the
exact ancestor, and otherwise keeps it pending.  Once verified, either form
acknowledges the sender's whole directly received prefix through that
coordinate, so one claim can backfill older lane coordinates.  Every resulting $(a,k,h)$
remains attributed to and counted first-hand from the envelope sender; the
envelope is neither a certificate nor a forwarded-MAC vector.  Batch-recovery
ECHOs remain by value and use the same claims over their core and tip vectors.
Only a quorum of exact-position claims promotes a correct proposer's core;
integer-derived counts remain valid for $\textsc{AuthorOK}$ and eventual
availability but are not assumed to become common within one delay.  A general
quorum therefore freezes the least-height unconfirmed direct prefix as an
exact-position confirmation target; it neither promotes the core nor clears
completed-open quarantine by itself.

The implementation retains periodic digest-bearing availability watermarks
behind the benchmark switch that disables ECHO claims.  This is an A/B oracle,
not the default protocol or a liveness dependency.

The proposal travels by value; later single-proposal ECHOs and READYs name
$\hash{B_v}$, a grade, and any origin bit.  Before a digest quorum acts, the
receiver must hold the canonical body and execute the unchanged by-value rule.
Statement authors retain and serve it, but served bytes create no provenance
or vote.  Transport batches to 64\,KiB or 5\,ms.  Bounded recovery and state
reclamation in the artifact are outside the protocol and theorem statements.

With $\lambda$ bounding a digest, identifier, and reference, the proposal and
digest statements cost $O(n^2\lambda)$ bits over all deliveries; the
$O(n)$-bit ECHO bitmaps add $O(n^3)$.  Thus the aligned single-proposal
encoding uses $O(n^2\lambda+n^3)$ control bits and $O(n^2)$ logical messages.
Sparse integer exceptions cost $O(n^3(\log n+\log\nu))$ bits, where $\nu$
bounds heights and distances; by-value recovery retains the proved
$O(n^3\lambda)$ statement bound under the paper's scalar-size convention.

\subsection{Experiment-specific records}

\subsubsection{Committee scaling}\label{app:committee-scaling}
The seven protocol variants at $n\in\{10,20,50\}$ use one on-demand
\texttt{c5d.2xlarge} validator per instance (8 vCPUs, 16 GiB RAM, 200 GB local
NVMe) in \texttt{eu-west-1a}, the matrix above, 100 aggregate tx/s, five-second
Prometheus scrapes, and the common warmup/window.  All 21 first attempts
complete at approximately the offered load, with no reported netem drop or
panic.  The $n=100$ row instead joins each variant's 100~tx/s point from the
throughput campaign and therefore inherits that figure's fleets
(\Cref{app:throughput-sweep,app:baseline-fleet}); this is why Bluestreak's
$n=100$ cell is \texttt{m5d.2xlarge} while its smaller sizes are
\texttt{c5d.2xlarge}.

\subsubsection{Fresh-reference cadence}\label{app:cadence}
For the $n=100$, 100~tx/s points, we subtract each validator's progress counter
at the measurement-window boundaries, divide by that run's actual window, and
take the median validator.  The available runs give about 14 Vantage proposal
views/s, about 10 Bluestreak DAG rounds/s, about four Sailfish++ rounds/s, and
about eight slots/s for both Autobahn modes.  Autobahn has no direct slot gauge,
so its rate is derived from received \texttt{ConsensusRequest} counts: one request per
optimistic slot and two per seamless slot.  The generator records every
run's median-validator delta in \texttt{carrier-cadence.json}.  These counters
measure the next opportunity to reference fresh data, not equivalent protocol
work or capacity.

\subsubsection{View anatomy and network-delay control}\label{app:decision-anatomy}
The $n=20$ mechanism study uses the target-local resolver implementation with
no fault-mode switch enabled; the resolver therefore remains inactive.
All validators run in one release binary on the 32-vCPU, 124-GiB native-Linux
host used for Q5, with matrix delays injected in process rather than through
the Docker netem path, the optimized encoding, and 100 tx/s.  Each of three
60-second runs submits 6,000 transactions and materializes all of them after
drain.  A reported percentile is first taken across validators and then the
median across runs.

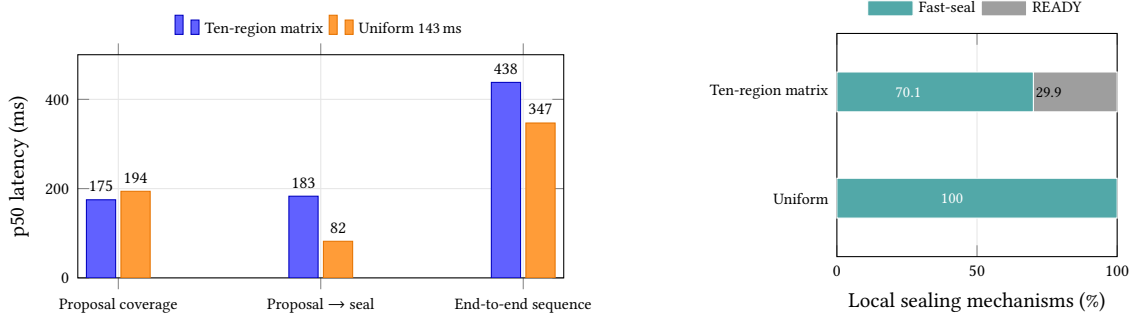
\begin{figure}[t]
  \centering
  \begin{tikzpicture}
\begin{axis}[
  ybar,
  bar width=11pt,
  width=0.53\textwidth,
  height=0.30\textwidth,
  ymin=0,
  ymax=500,
  symbolic x coords={Coverage,Seal,Ordered},
  xtick=data,
  xticklabels={Proposal coverage,Proposal $\rightarrow$ seal,End-to-end sequence},
  x tick label style={font=\scriptsize,align=center},
  ylabel={p50 latency (ms)},
  tick label style={font=\scriptsize},
  label style={font=\small},
  grid=major,
  grid style={black!10},
  legend columns=2,
  legend style={font=\scriptsize,draw=none,fill=none,
    at={(0.5,1.03)},anchor=south},
  nodes near coords,
  every node near coord/.append style={font=\scriptsize},
]
  \addplot[fill=blue!62,draw=blue!75!black]
    coordinates {(Coverage,175) (Seal,183) (Ordered,438)};
  \addlegendentry{Ten-region matrix}
  \addplot[fill=orange!78,draw=orange!90!black]
    coordinates {(Coverage,194) (Seal,82) (Ordered,347)};
  \addlegendentry{Uniform 143\,ms}
\end{axis}
\end{tikzpicture}
\hfill
\begin{tikzpicture}
\begin{axis}[
  xbar stacked,
  bar width=15pt,
  width=0.35\textwidth,
  height=0.30\textwidth,
  xmin=0,
  xmax=100,
  ymin=-0.55,
  ymax=1.55,
  y dir=reverse,
  ytick={0,1},
  yticklabels={Ten-region matrix,Uniform},
  xlabel={Local sealing mechanisms (\%)},
  xtick={0,50,100},
  tick label style={font=\scriptsize},
  yticklabel style={font=\scriptsize},
  label style={font=\small},
  grid=major,
  grid style={black!10},
  legend columns=2,
  legend style={font=\scriptsize,draw=none,fill=none,
    at={(0.5,1.03)},anchor=south},
  area legend,
]
  \addplot[
    fill=teal!68,draw=white,
    nodes near coords,point meta=explicit symbolic,
    every node near coord/.append style={font=\scriptsize,text=white,
      anchor=east,xshift=-2pt}
  ] coordinates {(70.1,0) [70.1] (100,1) [100]};
  \addlegendentry{Fast-seal}
  \addplot[
    fill=black!38,draw=white,
    nodes near coords,point meta=explicit symbolic,
    every node near coord/.append style={font=\scriptsize,
      anchor=east,xshift=-2pt}
  ] coordinates {(29.9,0) [29.9]};
  \addlegendentry{READY}
\end{axis}
\end{tikzpicture}
  \caption{Matched network-delay control at $n=20$.  Left: p50 stage and
  end-to-end latency.  Right: local sealing mechanisms summed over validators.  The
  uniform value is the 143\,ms median off-diagonal RTT after assigning the
  ten-region matrix to 20 validators.}
  \Description{Grouped bars compare stage and end-to-end latency under the
  ten-region matrix and a uniform RTT.  Stacked bars show about 70 percent fast
  seals under the matrix and 100 percent under the uniform RTT.}
  \label{fig:vantage-rtt-shape}
\end{figure}

\begin{figure}[t]
  \centering
  \begin{tikzpicture}
\begin{axis}[
  xbar stacked,
  bar width=13pt,
  width=0.90\textwidth,
  height=0.20\textwidth,
  xmin=0,
  xmax=100.1,
  ymin=-0.55,
  ymax=1.55,
  y dir=reverse,
  ytick={0,1},
  yticklabels={Serialized bytes,Logical messages},
  xlabel={Share of typed outgoing traffic (\%)},
  xtick={0,25,50,75,100},
  grid=major,
  grid style={black!10},
  tick label style={font=\scriptsize},
  yticklabel style={font=\scriptsize},
  label style={font=\small},
  legend columns=4,
  legend style={font=\scriptsize,draw=none,fill=none,
    at={(0.5,-0.62)},anchor=north,
    /tikz/every even column/.append style={column sep=0.30cm}},
  area legend,
]
  \addplot[fill=orange!78,draw=white] coordinates {(39.354,0) (58.164,1)};
  \addlegendentry{AGB}
  \addplot[fill=blue!62,draw=white] coordinates {(30.529,0) (9.403,1)};
  \addlegendentry{Payload}
  \addplot[fill=teal!62,draw=white] coordinates {(20.018,0) (18.617,1)};
  \addlegendentry{Lane blocks}
  \addplot[fill=purple!62,draw=white] coordinates {(5.484,0) (2.283,1)};
  \addlegendentry{Sequence state}
  \addplot[fill=green!55!black,draw=white] coordinates {(4.501,0) (11.265,1)};
  \addlegendentry{Primary--worker}
  \addplot[fill=red!68,draw=white] coordinates {(0,0) (0,1)};
  \addlegendentry{Resolver (none)}
  \addplot[fill=black!38,draw=white] coordinates {(0.114,0) (0.268,1)};
  \addlegendentry{Other}
\end{axis}
\end{tikzpicture}
  \caption{Outgoing typed bytes and logical unicast messages by module in a
  representative matrix run of the $n=20$ mechanism study.  The inactive
  resolver sends nothing.}
  \Description{Two stacked bars show the byte and logical-message shares of
  AGB, payload, lane blocks, sequence state, primary--worker traffic, and other
  traffic; the resolver share is zero.}
  \label{fig:vantage-traffic-shares}
\end{figure}
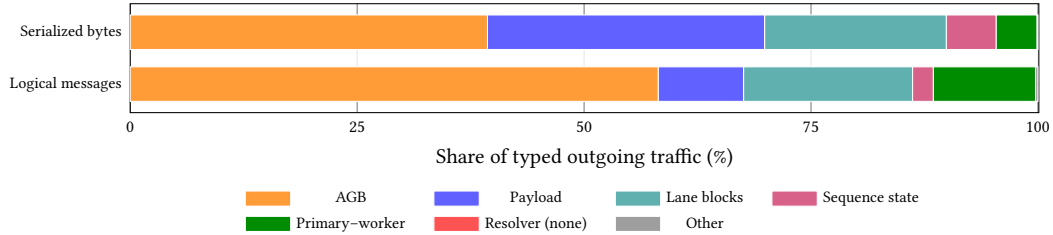

Fast-seal waits for all 20 matching ECHOs, while Direct AGB may finish from
$n-f=14$ READYs.  Under heterogeneous delays, the Direct AGB result sometimes
precedes the slowest ECHO, so the fast seal supplies about 70\% of local seals;
at a uniform 143\,ms RTT every observed seal uses fast-seal.  All seals are full.
Uniformization changes coverage from 175 to 194\,ms, proposal-to-seal from 183
to 82\,ms, and reduces end-to-end sequencing latency from 438 to 347\,ms;
these are direct submission-to-sequencing medians, which the chained stage
medians below, taken over different event populations, need not add up to.
\Cref{fig:vantage-traffic-shares} gives the outgoing traffic shares by module
in the representative matrix run.

Under the matrix, the run-median p50/p99 is about 444/596~ms from
submission to materialization.  From block publication, the send of the first
proposal naming the block as tip is about 119/212~ms, measured on the shared
clock of the single-process run in three repeat runs whose other stage medians
agree with the values here within 3~ms; author-observed first proposal
coverage is about 175/320~ms, exact-$Q$ availability 298/404~ms, and sequencing
371/466~ms.  Proposal receipt to local full seal is about 183/275~ms.  Under
uniform RTT, submission-to-materialization is about 353/434~ms and every seal
uses fast full.  The remaining local p50 intervals are about 11~ms from
submission to batch seal, 51~ms from the first batch digest to author-block
publication, and 6~ms from sequencing to materialization.
The 51~ms is the expected half of the 100~ms header delay, and the 47~ms by
which tip naming exceeds the 72~ms median one-way delay is the wait for the
next proposal opportunity.

The representative matrix run sends about 1.2 million typed logical messages,
200~MB of typed-message bodies, and 230~MB physically over 60 seconds.  AGB
contributes about 40\% of typed bytes and 60\% of messages; payload plus lane
blocks contribute about half of the bytes and 30\% of messages.  Sequence state
and primary--worker traffic account for nearly all of the remainder.  The
target-local resolver contributes no bytes or messages.  Optional tracing
retains at most 4,096 block timestamps per validator.

For a hardware-consistency check, we retain the same six repetitions on the
original Apple M4 Pro.  Its matrix p50/p99 is about 445/599~ms, its uniform
p50/p99 359/445~ms, and it again emits no resolver traffic.  Q4 reports
the native-Linux measurements and uses the workstation set only as a
consistency check; the two machines are not pooled.

\subsubsection{Throughput sweep}\label{app:throughput-sweep}
The $n=100$ sweep offers 100, 10,000, 150,000, 200,000, 225,000, 250,000,
275,000, and 300,000 tx/s on \texttt{c5d.2xlarge} validators (8 vCPUs,
16 GiB RAM, 200 GB NVMe) in \texttt{eu-west-1a}; the Sailfish++ campaign adds
a 170,000 tx/s point.

The campaign strictly validates all validators at 100 and 10k and treats
higher loads as exploratory.  Acceptance
means at least 95\% median throughput, with the first lower point retained as
overload.  Accepted points end with 100 responsive validators and no reported
netem drops.  \sysname{}, optimistic all-to-all Autobahn, seamless Autobahn,
and both Simple-IT variants pass through 275k and overload at 300k;
Bluestreak overloads at 275k and Sailfish++ at 170k.  These are
deployment endpoints, not capacity bounds.

\subsubsection{Leader-relay stress}\label{app:leader-relay}
The $n=20,f=6$ campaign offers 100, 1,000, 10,000, 20,000, 100,000, and
200,000 total tx/s.  Six Byzantine publishers receive the same per-author
load as every correct publisher.  Each creates one batch per $\Delta$, sends
it directly to a fixed five-party correct-holder group, exposes the lane
header, and refuses repair.  The publisher retains its own copy and correct
holders serve normally, so each batch has six direct holders but remains
below the $f+1=7$ PoA threshold.  The six holder groups jointly cover all 14
correct consensus leaders while preserving a complete prefix at every holder.
In the Autobahn variant, Byzantine consensus leaders propose their certified
cut rather than self-inflicting a timeout.  Across the three variants, the
only emulated faulty behavior is partial data publication and repair refusal.

Each protocol uses a fresh fleet of 20 on-demand \texttt{c5d.2xlarge}
validators plus one control instance in \texttt{eu-west-1a}.  Protocol traffic
uses private addresses and the matrix of \Cref{tab:rtt-matrix} through netem;
observed probes span approximately 1--309\,ms.  The common 30-second warmup,
120-second window, 512-byte transactions, calibrated timers, and positional
Vantage encoding apply.
Useful throughput excludes transactions from the six Byzantine publishers;
their Autobahn commitment is measured separately.  All points except the
200k optimistic Autobahn retry report 20/20 healthy validators, zero netem
drops, and zero panics.  The retry ends at 19/20 healthy; the overloaded host
does not expose its final netem counter.  A first 200k attempt, retained only
as a diagnostic, stalls two validators and is not plotted.

\subsubsection{Target-local resolver recovery}\label{app:resolver-evaluation}
Q5 runs Docker on a separate native-Linux host with 32 virtual CPUs and
124~GiB RAM.  Both profiles use $\Delta=200$~ms, the ten-region netem matrix,
a 100,000-packet netem limit, 512-byte transactions, compact committee IDs,
authenticated channels, and one-second Prometheus reads.  These local
mechanism runs do not share a fleet or measurement window with the cloud
campaigns.

The transient profile uses $n=10,f=3$ for 130 seconds.  The counted
1,000~tx/s load is assigned only to the seven non-victims.  At second~20, the
other three otherwise-correct containers receive SIGKILL simultaneously; they
restart with their stores after a configured 30-second blackout, about
31 seconds in the observed timeline.  Each process exports an exact latency
histogram cleared every second, and Prometheus reads every second.  The figure plots the non-victim
median at each read, not a repetition aggregate: throughput is a trailing
five-second rate and latency is the latest-one-second p50.  All ten validators
are live at the end.  This profile has one run.

The mixed-open profile uses $n=10,f=3$ for 240 seconds and three fresh
committees.  From seconds 20--110, each Byzantine publisher proposes only its
own partially disseminated current tip, ignoring its local quarantine, and
suppresses its outgoing AGB, skip, WISH, resolver-phase, witness, DONE, and
resolver-serve messages.  Outside that interval it resumes ordinary behavior
immediately.  Correct publishers offer the counted 1,000~tx/s; the faulty
publishers add 600 uncounted tx/s.  Resolver completion is read from protocol events at
every correct validator; the plotted backlog is the committee median of
completed-open injected targets not yet sealed locally.  The arrival/seal
comparison uses the common steady interior of each attack, excluding its
boundary transients.  All reported Q5 runs end with every correct validator
live, no panic, and the injected backlog drained.

Prometheus evaluates both mixed-open metric series every second and excludes
the three Byzantine validators.  Throughput is a trailing-five-second rate;
latency is the median of the correct validators' latest-one-second
materialization p50.  The figure first computes those per-run committee
medians and then shows their three-run median and min--max band.  Empty
one-second latency windows remain missing rather than repeating a stale gauge.

The plotted mixed-open repetitions use the complete resolver path, including
admission, candidate cycling, retained-terminal participation, and
send-upon-join replay.  Resolver unit/integration tests and an $n=31$ smoke run
provide additional functional coverage.

Resolver instances activate once the proposal horizon reaches $u+3$.  In this
profile a Byzantine primary sends no resolver proposal, so its view ends at the
$5\Delta=1$~s no-proposal timeout.  A target can encounter up to two such
primaries before a correct-primary execution.  The separate $11\Delta=2.2$~s
deadline applies after a proposal appears and need not expire in a successful
correct-primary view.  The traces report target completion and output release,
but do not divide the observed delay among these stages.

\subsubsection{Inactive-path check}\label{app:resolver-load}
A separate one-run-per-arm, fault-free $n=30$ diagnostic at 1,000~tx/s checks
that merely enabling the target-local resolver does not perturb the inactive
path.  Both builds sustain the offered load at about 0.44~s p50, 0.60~s p99,
and 0.21 median validator CPU cores, with no observable inactive-path
regression.

\subsubsection{Artifact scope}\label{app:artifact-scope}
The implementation and harness repositories named in \Cref{sec:implementation}
contain the implementation, experiment configurations, generated figure
inputs, and per-run records; the figure data of this paper are under
\texttt{recorded/paper-fixedbuild-20260822} in the harness repository.  Q1 and Q2 terminate through
homogeneous AGB and consume no
fallback decision; Q3 stresses partial-publication admission.  Q4 keeps the
target-local resolver inactive, whereas Q5 exercises the recovery path: the
grounded skip shortcut under crashes and the resolver under mixed-open faults.  The matched $n=30$ check above compares the inactive path
before and after enabling the target-local resolver.

\paragraph{Reproduction.}
In the artifact's \texttt{wan-bench} repository, run
\texttt{wanbench campaign {-}{-}config CONFIG {-}{-}execute}; its paper
configuration index names the committee, throughput, baseline, and
leader-relay campaigns.  Regenerate the figure inputs with
\texttt{python3 vantage/gen\_paper\_dat.py}.  For the local controls, build the
artifact's \texttt{vantage-bft} repository with
\texttt{cargo build {-}{-}release -p node {-}{-}features pipeline-tracing} and run
\texttt{./target/release/node local-benchmark}; the repository contains the Q4
flags for the matrix and uniform-RTT arms.  Repeat each arm three times with a
fresh data directory.

For Q5, run the scripts in \texttt{vantage-bft/docker-bench}.  The sustained
profile uses \texttt{./direct\_resolver\_q5.sh}; the transient profile uses
\texttt{./transient\_crash\_q5.sh}.  The recorded configurations supply the
durations and repetition count.  The latter fixes
$n=10,f=3$, keeps the counted load off validators 0--2, schedules their
30-second blackout.
The plotting and validation commands are
\texttt{direct\_resolver\_q5.py} and \texttt{transient\_crash\_q5.py} in the same
directory; each accepts the archived run root and an output directory.

\section{Extended Related Work}\label{sec:related-work}

\Cref{sec:closest-comparisons} gives a normalized quantitative comparison with
five baselines.  Here we place those protocols and other related work in the
broader design space, organized by where a protocol spends strong agreement
and whether its evidence is transferable.

\paragraph{Signature-free BFT}
Agreement without digital signatures is classical: Castro's MAC-vector
Algorithm BFT reaches a $3\delta$ good case, with cubic authenticator
communication from recipient-specific MAC vectors~\cite{pbft,castrothesis}.
Those vectors provide forwarded per-recipient evidence and hence lie outside
our channel-only model.  Mostéfaoui, Moumen, and Raynal gave signature-free
asynchronous binary consensus with $O(n^2)$
messages~\cite{mmr14}.  Information-Theoretic HotStuff brought the rotating
leader-and-view discipline to authenticated channels without signatures or a
PKI~\cite{iths};
TetraBFT reduces the good case to $5\delta$~\cite{tetrabft}; Forget-IT
attains the good-case optimum of $3\delta$~\cite{forgetit}.  Simple-IT makes
the line practical with reliable notification and an optimistic broadcast path
under a stronger honesty threshold~\cite{simpleit}; speculative pipelining can
narrow leader-proposal spacing at the cost of requiring runs of correct
leaders.  Its shared-mempool variant is compared directly in
\Cref{sec:closest-comparisons}.  Vantage instead instantiates fresh IT-HS
agreement independently for each unresolved target.  A one-type first-hand
witness relay supplies the stable external-validity predicate that the
standalone agreement protocol assumes, and typed WISH messages synchronize
each target's resolver views.  No instance runs before target-local evidence
or $f+1$ peer WISHes activate it.  This
dependency can gate output after an open tip but remains outside direct
completion and homogeneous sealing.
A different route to the same goal buys authentication from the
infrastructure: SwitchBFT drops protocol signatures by relying on packet
source authentication in a trusted datacenter network, and programs a switch
to enforce agreement decisions and check that safety is not violated, matching
the speed of an in-switch crash-fault-tolerant protocol~\cite{switchbft}.  Its
trust therefore rests on the network --- a fabric in one administrative domain
with a trusted switch on the path --- where \sysname{} assumes only pairwise
authenticated channels between validators and no trusted network element,
which is what lets it run over the wide-area committee of
\Cref{sec:evaluation}.

\paragraph{Certified-data designs.}
Narwhal separates availability-certified dissemination from
ordering~\cite{narwhal}; DAG-Rider~\cite{dagrider},
Bullshark~\cite{bullshark}, Shoal and Shoal++~\cite{shoal,shoalpp}, and
Sailfish~\cite{sailfish} order DAGs in which every vertex is reliably
broadcast or availability-certified, so ordering adds little communication
but every data unit completes a strong-broadcast or certification round before
it can be referenced; most practical instantiations also use signatures.
DAG-Rider's RBC-DAG is signature-free given an ideal perfect coin, although
its standard practical coin realization uses threshold cryptography for
liveness.  Sailfish++ makes the Sailfish line signature-free by running
optimistic signature-free RBC for every vertex~\cite{sailfishpp};
\Cref{sec:closest-comparisons} details the resulting nonleader classes and the
contrast with Vantage's graded tip.  On the signed side, Autobahn combines
PoA-certified per-author lanes with a PBFT-style core and optionally admits an
uncertified current tip~\cite{autobahn};
\Cref{sec:closest-comparisons} compares that optimization and its seamlessness
tradeoff directly with \sysname{}'s graded tip.  Dumbo-NG
separates continuously growing signed data streams from asynchronous
agreement~\cite{dumbong}, and Raptr integrates prefix certification of
uncertified batches into leader consensus~\cite{raptr}.
Simplex~\cite{simplex} and DispersedSimplex~\cite{dispersedsimplex} are the
signed leader-dissemination baselines against which the signature-free
leader line measures itself.

Concurrent work, Multimmit, independently addresses missing data in a signed
multi-chain protocol~\cite{multimmit}.  Its votes report a supported prefix
per producer chain and may extend the leader's proposed tips, so one missing
producer need not block the other chains, a localization they prove rather than
approximate.  Its per-chain report is a finer split than Vantage's single grade
bit.  Its extension votes omit a step that Vantage retains: a Vantage tip
entry must reach the proposer before the proposal.  In its good case a block
disseminated at $t$ is ordered by $t+3\delta$ in expectation and $t+2\delta$ at
best, measured from dissemination as ours is.  The cost is cryptographic and in
resilience: $n\geq5f+1$, a PKI, and ordinary plus aggregate signatures, together
with two threshold schemes that their optimization section shows can be replaced
by the aggregate scheme alone.  Aggregation is what keeps their certificates
succinct, and the standardized post-quantum signatures have no comparable
aggregate~\cite{fips204,fips205}, so a post-quantum instantiation would keep
the structure and lose the succinctness.  Vantage studies signature-free one-delay admission
with $n\ge 3f{+}1$.  Because it is concurrent work with no reported evaluation and a
different fault threshold, it is not included in \Cref{tab:data-path-comparison}.

Other concurrent work takes proposal pipelining beyond the one-round cadence
of DAGs: Cadence runs independent signed slots at an arbitrary configured
interval $\tau$, with concurrent proposers disseminating coded encrypted
proposals against $\tau$-spaced deadlines, and signed include/exclude
evidence plus multivalue Byzantine agreement (MVBA) resolving a failed fast
path~\cite{cadence}.  Cadence can
configure $\tau$ independently of $\delta$, a stronger pacing result than our
all-correct at-most-$\delta$ fresh-opportunity gap.  Its favorable per-slot finality is
$\Delta+2\delta$ from the formal slot start $D-\Delta$, where $D$ is the slot
deadline, although in-order append may wait for earlier slots.  Its endpoint and mechanisms --- aggregate
signatures, per-slot threshold encryption, and proposer-local input rather
than a distinct data-block author --- place it outside the comparison
experiment of \Cref{tab:data-path-comparison}.  \sysname{} instead studies how
far pipelined multi-author data can go with hashes and non-transferable
first-hand counts alone.

Banyan paired the $n=3f{+}2p{-}1$ two-delay fast path with rotating
leaders~\cite{banyan}; concurrent FinWhale brings it to Mysticeti, using
evidence blocks to reconcile fast and slow decisions across local DAG views
~\cite{finwhale}.  At $p=1$, it keeps $n=3f{+}1$ and survives one non-voter,
unlike Vantage's all-$n$ fast-seal shortcut, but its blocks remain signed.  The two systems
share the $2\delta$ carrier-relative endpoint; their aligned totals and
FinWhale's round bound are in \Cref{tab:data-path-comparison}.

\paragraph{Recent hybrid and sparse DAGs.}
Recent work reduces certified-DAG metadata without removing the data-side
dissemination gate.  Angelfish replaces some nonleader RBC vertices with
one-delay best-effort votes, but those votes carry no transactions; every
transaction-carrying vertex remains reliably broadcast, and the protocol uses
BLS certificates~\cite{angelfish}.  Clownfish sparsifies certified-DAG edges
in RBC- and signed consistent-broadcast variants and explicitly leaves a
signature-free adaptation to future work~\cite{clownfish}.  DAGs for the
Masses uses verifiable sampling to reduce Bullshark's parent set, while every
vertex remains reliably broadcast and references carry signed availability
evidence~\cite{dagsmass}.  Morpheus is a closer signed endpoint: its
Autobahn-style mode lets a leader hash-reference a freshly received author
block, but voters proceed only after obtaining that block's signed
$0$-QC~\cite{morpheus}.  Hashgraph's constant-degree gossip DAG likewise
uses signed two-parent events~\cite{hashgraph}; the broader certified and
uncertified DAG design space is surveyed in~\cite{dagsok}.

\paragraph{Uncertified data with signatures.}
Cordial Miners introduced leader-based ordering of a best-effort block
DAG~\cite{cordialminers}; Mysticeti improves latency by pipelining leaders,
committing a signed uncertified DAG with a
three-round leader path~\cite{mysticeti}.  Its idealized aligned
cross-author block path is $4\delta$; the familiar $4.5\delta$ starts at
transaction arrival and adds the average wait for the author's next block.
The Starfish paper presents the pacemaker under which this line is provably
live, and its Mysticeti-L variant improves consensus bandwidth with
multisignatures.  Its explicit all-correct post-GST construction for
Mysticeti approaches $5\tfrac{1}{3}\delta$ average block latency; adding the
$\delta/3$ admission induced by that synchronized arrival schedule gives
$5\tfrac{2}{3}\delta$ average transaction latency, a different experiment.
Starfish itself further decouples availability through coded
dissemination~\cite{starfish}.
Bluestreak reduces nonleader metadata without
multisignatures~\cite{bluestreak}, and concurrent FinWhale adds a two-round
fast path to this family~\cite{finwhale}.  In these protocols, signatures keep
relayed blocks and votes attributable.  Vantage instead distinguishes direct
author publication from repair and counts author-receipt acknowledgments
first-hand.

Starfish-S is the closest mechanism lineage: its optimistic, standard, and
pending grades use later committed anchors~\cite{starfishs}.  Vantage replaces
its signed DAG evidence with first-hand response histories and locally
validated entries.

\paragraph{Graded and abortable broadcast.}
Feldman--Micali gradecast grades confidence in one value for Byzantine
agreement~\cite{feldmanmicali}; AGB instead grades optional-suffix inclusion
after fixing a common core and omits faulty-proposer totality.  BBCA gives an
indivisible value Complete--Adopt, and BBCA-Chain recovers a backbone with
transferable adopt/no-adopt evidence~\cite{bbcachain}.  AGB's
READY/\noready{} and READY-mix refinement imply a graded counterpart, but Vantage
resolves open and general silent views through target-local agreement, using its
grounded shortcut only for clean crash silence.  Optimistic RBC retains
reliable-broadcast totality~\cite{bracha,sailfishpp}.  The BBCA-Chain comparison
in \Cref{sec:discussion} gives Vantage no bounded fallback.
GradedDAG grades a different object in an asynchronous signed DAG
~\cite{gradeddag}.

\paragraph{Foundations.}
AGB uses Bracha's echo/ready pattern~\cite{bracha}; Abraham, Ren, and Xiang
classify unauthenticated broadcast latency~\cite{arxunauth}.  Its pacemaker
combines the bounded-space synchronizer of Bravo, Chockler, and Gotsman with
Starfish's two-response trigger~\cite{bravo2022liveness,starfishs}; verifiable
dispersal could compress data and repair independently~\cite{avid}.
The on-demand resolver adopts IT-HS's three-key and lock-opening discipline
for first-hand authenticated messages~\cite{iths}.  Separately, Vantage
excludes overlap between grounded skip votes and ECHOs for fresh non-skip
resolver values, without $f+1$ skip-vote amplification.  Its $f+1$ relay
applies to resolver witnesses only after a genuine fresh ECHO quorum.

\end{document}